\documentclass[11pt,fleqn]{article}
\usepackage{amssymb,latexsym,amsmath,amsfonts,graphicx, ebezier}

    \advance\textheight by \topskip

    \numberwithin{equation}{section}

\makeatletter
\def\eqalign#1{\null\vcenter{\def\\{\cr}\openup\jot\m@th
  \ialign{\strut$\displaystyle{##}$\hfil&$\displaystyle{{}##}$\hfil
      \crcr#1\crcr}}\,}
\makeatother

\newcommand{\be}{\begin{equation}}
\newcommand{\ee}{\end{equation}}

\newcommand{\wt}{\widetilde}
\newcommand{\wh}{\widehat}
\newcommand{\de}{\delta}

\newcommand{\al}{\alpha}
\newcommand{\bt}{\beta}

\newcommand{\Ga}{\Gamma}

\newcommand{\si}{\sigma}

\newcommand{\om}{\omega}

\newcommand{\lb}{\lambda}

\renewcommand{\th}{\theta}

\newcommand{\ep}{\varepsilon }
\newcommand{\eit}{e^{it}}

    \def\e{{\epsilon}}

    \def\Re{{\rm Re \,}}
    \def\Im{{\rm Im \,}}

    \def\bigO{{\cal O}}
    
    \def\Res{{\rm Res}}
    
    \def\P2n{{\rm P}_{{\rm II}}^{(n)}}

    \newtheorem{theorem}{Theorem}[section]
    \newtheorem{lemma}[theorem]{Lemma}
    
    \newtheorem{proposition}[theorem]{Proposition}
    
    \newtheorem{Definition}[theorem]{Definition}
    
    \newtheorem{Remark}[theorem]{Remark}
    \newenvironment{remark}{\begin{Remark}\rm}{\end{Remark}}
    \newtheorem{Example}[theorem]{Example}
    
    \newtheorem{Assumptions}[theorem]{Assumptions}

    \newenvironment{proof}%
    {\rm \trivlist \item[\hskip \labelsep{\bf Proof. }]}%
    {\hspace*{\fill}$\Box$\endtrivlist}
    {\rm \trivlist \item[\hskip \labelsep{\bf Proof}]}%
    {\hspace*{\fill}$\Box$\endtrivlist}

\begin{document}
\title{Toeplitz determinants with merging singularities}
\author{T. Claeys\footnote{Institut de Recherche en Math\'ematique et Physique,  Universit\'e
catholique de Louvain, Chemin du Cyclotron 2, B-1348
Louvain-La-Neuve, BELGIUM}
\ and I. Krasovsky\footnote{Department of Mathematics, Imperial College London,
London SW7 2AZ, United  Kingdom
}}
\maketitle

\begin{abstract}
We study asymptotic behavior for the determinants of $n\times n$ Toeplitz matrices corresponding to symbols with two Fisher-Hartwig singularities at the distance $2t\ge0$ from each other on the unit circle.
We obtain large $n$ asymptotics which are uniform for $0<t<t_0$ where $t_0$ is fixed. They describe the transition as $t\to 0$ between the asymptotic regimes of 2 singularities and 1 singularity. The asymptotics involve a particular solution to the Painlev\'e V equation. We obtain small and large argument expansions 
of this solution. As applications of our results we prove a conjecture of Dyson on the largest occupation number in the ground state of a one-dimensional Bose gas, and a conjecture of Fyodorov and Keating on the second moment of powers of the characteristic polynomials of random matrices. 
\end{abstract}

\tableofcontents 

\section{Introduction}

Consider Toeplitz matrices 
\begin{equation}\label{Toeplitz matrix}
T_n(f_t)=(f_{t,j-k})_{j,k=0}^{n-1}, \qquad
f_{t,j}=\frac{1}{2\pi}\int_{0}^{2\pi}f_t(e^{i\theta})e^{-ij\theta}d\theta,
\end{equation}
where the complex-valued symbol $f_t(z)$ depends on a parameter $t$
and has the form
\be\label{symbol}
f_t(z)=e^{V(z)} z^{\bt_1+\bt_2}
\prod_{j=1}^2  |z-z_j|^{2\al_j}g_{z_j,\bt_j}(z)z_j^{-\bt_j},\qquad z=e^{i\th},\qquad
\theta\in[0,2\pi),
\ee
where $z_1=e^{it}$, $z_2=e^{i(2\pi-t)}$, $0<t<\pi$,
\begin{eqnarray}
&g_{z_j,\bt_j}(z)=
\begin{cases}
e^{i\pi\bt_j}& 0\le\arg z<\arg z_j\cr
e^{-i\pi\bt_j}& \arg z_j\le\arg z<2\pi
\end{cases},&\label{g}\\
&\Re\al_j>-1/2,\quad \bt_j\in\mathbb{C},\quad j=1,2.&
\end{eqnarray}
The condition on $\al_j$ ensures
integrability of $f_t$. We assume $V$ to be analytic in a neighborhood of the unit circle, with the Laurent series \[V(z)=\sum_{k=-\infty}^{\infty}V_k z^k,\qquad V_k=\frac{1}{2\pi}\int_0^{2\pi}V(e^{i\theta})e^{-ik\theta}d\theta.\]
The function $e^{V(z)}$ allows the standard Wiener-Hopf decomposition:
\be\label{WH}
e^{V(z)}=b_+(z)b_0b_-(z),\quad b_+(z)=e^{\sum_{k=1}^{\infty}V_kz^k},\quad b_0=e^{V_0},\quad
b_-(z)= e^{\sum_{k=-\infty}^{-1}V_kz^k}.
\ee

The function $f_t(z)$ is a standard form of a symbol with 2 Fisher-Hartwig (FH) singularities at
the points $z_1=e^{it}$ and $z_2=e^{i(2\pi-t)}$. The parameters $\alpha_1$ (at
$z_1$) and $\alpha_2$ (at $z_2$) describe power- or root-type
singularities, $\beta_1$ and $\beta_2$ describe jump
discontinuities. We always assume below that the both singular points are genuine, i.e., that either $\al_j\neq 0$ or $\bt_j\neq 0$,
$j=1,2$.  

We are interested in the large $n$ behavior of the Toeplitz determinant
\be\label{Toeplitz}
D_n(f_t)=\det T_n(f_t)
\ee
when the distance between the singularities is small, i.e. $t$ is small.

A study of asymptotics of Toeplitz determinants as $n\to\infty$ was initiated by Szeg\H o in 1915 for symbols without singularities. Singular symbols, however, appear naturally in applications such as exactly solvable models (most notably, the two-dimensional Ising model), random matrices, etc.
The large $n$ behavior of Toeplitz determinants with several FH singularities has been studied by many authors under various assumptions on $V$ and the values of the parameters $\alpha_j$ and $\beta_j$, see, e.g.,\ \cite{FH, W, Basor, Basor2, BS, Ehr, DIK2, DIK22}. 
A recent historical account on Toeplitz determinants is given in \cite{DIK-hist}.

If the weight has two singularities as in (\ref{symbol}), the asymptotics of $D_n(f_t)$ are described as follows.
Following \cite{DIK2}, define
\[
|||\beta|||\equiv |||(\beta_1,\beta_2)|||\equiv|\Re(\beta_1-\beta_2)|.
\]
Let first
\begin{equation}
|||\beta|||<1,
\end{equation}
and assume that $\alpha_j\pm \beta_j\neq -1, -2, \ldots$ for $j=1,2$ (we 
always assume this ``nondegeneracy'' condition throughout this paper).
Then, the asymptotics of $D_n$ as
$n\to\infty$ for a fixed $t>0$ are given by (\cite{Ehr}; see \cite{DIK22} for the estimate of the error term)
\begin{equation}\label{as FH2}
\ln D_n(f_t)=nV_0+(\alpha_1^2+\alpha_2^2-\beta_1^2-\beta_2^2)\ln n + E(V,\alpha_1,\alpha_2,\beta_1,\beta_2,t)+
o(1),
\end{equation}
where
\begin{multline}\label{EV}E(V,\alpha_1,\alpha_2,\beta_1,\beta_2,t)=
\sum_{k=1}^\infty kV_k
V_{-k}+2(\beta_1\beta_2-\alpha_1\alpha_2)\ln |2\sin t|+i(\pi-2t)(\alpha_1\beta_2-\alpha_2\beta_1)\\
-\al_1(V(z_1)-V_0)+\bt_1\ln\frac{b_+(z_1)}{b_-(z_1)}
-\al_2(V(z_2)-V_0)+\bt_2\ln\frac{b_+(z_2)}{b_-(z_2)}\\
+
\ln\frac{G(1+\alpha_1+\beta_1)G(1+\alpha_1-\beta_1)G(1+\alpha_2+\beta_2)G(1+\alpha_2-\beta_2)}{G(1+2\alpha_1)G(1+2\alpha_2)},
\end{multline}
and where $G(z)$ is Barnes' $G$-function (it is an entire function; it satisfies the difference equation $G(z+1)=\Gamma(z)G(z)$ in terms of the Euler Gamma function, with the condition $G(1)=1$;  its zeros are 
$z=0,-1,-2,\dots$).
The error term in (\ref{as FH2}) is $o(1)=\bigO(n^{-1+|||\bt|||})$.
One of the results in the present paper is an extension of the validity of (\ref{as FH2})
(with the corresponding change in the error term estimate): see Theorem \ref{theorem: extension} below.

The case $|||\bt|||\ge 1$ reduces either to $|||\bt|||<1$ or to $|||\bt|||=1$ as follows.

If $|||\beta|||\geq 1$ and $|||\beta|||$ is not an odd integer, we can choose $k\in\mathbb Z$
such that $|||\beta'|||\equiv|||(\beta_1+k,\beta_2-k)|||<1$, where $\bt_1'=\bt_1+k$, $\bt_2'=\bt_2-k$.
 The change of variable $\beta\mapsto \beta'$ leaves the symbol $f$ invariant, except for multiplication by the constant factor $e^{-2ikt}$:
\begin{equation}
\label{symbol2}
f_t(e^{i\theta})\equiv f_t(e^{i\th};\al_1,\al_2,\bt_1,\bt_2)=
e^{2ikt}f_t(e^{i\th};\al_1,\al_2,\bt_1',\bt_2').
\end{equation}
Since $|||\bt'|||<1$, the formula
(\ref{as FH2}) can now be used for the symbol in the r.h.s. of (\ref{symbol2}), 
if $\al_j\pm\bt_j'\neq -1,-2,\dots$, to compute the asymptotics for $D_n(f_t)$:
 \[
 D_n(f_t)=e^{n(V_0+2ikt)}n^{\alpha_1^2+\alpha_2^2-\beta_1'^2-\beta_2'^2}e^{E(V,\alpha_1,\alpha_2,\beta_1',\beta_2')}(1+\bigO(n^{-1+|||\bt'|||})),\qquad n\to\infty.
 \]
If $|||\beta|||\ge 1$ is an odd integer, there exists $k\in\mathbb Z$ such that 
$|||\beta'|||\equiv|||(\beta_1+k,\beta_2-k)|||=1$. Let $(\bt_1'',\bt_2'')=(\bt_1'+\ell,\bt_2'-\ell)$, where
$\ell=1$ if $\Re\bt_1'<\Re\bt_2'$, and $\ell=-1$ if $\Re\bt_1'>\Re\bt_2'$. 
Then $|||\bt''|||=1$ and we have \cite{DIK2} if $\al_j\pm\bt_j',\al_j\pm\bt_j''\neq -1,-2,\dots$:
\begin{multline}\label{as FH2b1}
D_n(f_t)= e^{n(V_0+2ikt)} [
n^{\alpha_1^2+\alpha_2^2-\beta_1'^2-\beta_2'^2}e^{E(V,\alpha_1,\alpha_2,\beta_1',\beta_2')}\\+
e^{2in\ell t}n^{\alpha_1^2+\alpha_2^2-\beta_1''^2-\beta_2''^2}e^{E(V,\alpha_1,\alpha_2,\beta_1'',\beta_2'')}](1+\bigO(n^{-1})),\qquad n\to\infty.
\end{multline}
Note that $\Re\{\beta_1'^2+\beta_2'^2\}=\Re\{\beta_1''^2+\beta_2''^2\}$, 
and therefore the 2 terms in (\ref{as FH2b1}) are of the same order.
As with (\ref{as FH2}), in this paper we also extend the validity of (\ref{as FH2b1}): see the discussion following Theorem \ref{theorem: Toeplitz 3} below.

If we let $t$ decrease towards $0$, the symbol (\ref{symbol}) reduces to
\begin{equation}
f_0(z)=e^{V(z)}|z-1|^{2(\alpha_1+\alpha_2)}z^{\beta_1+\beta_2}e^{-i\pi(\beta_2+\beta_1)},
\end{equation}
which has only one FH singularity at $1$, with
parameters $\alpha_1+\alpha_2$ and $\beta_1+\beta_2$. Then the
asymptotics of $D_n$ as $n\to\infty$ are given by (for $\Re (\al_1+\al_2)>-1/2$, $\al_1+\al_2\pm(\bt_1+\bt_2)\neq -1,-2,\dots$)
\begin{multline}\label{as FH1}
\ln D_n(f_0)=nV_0+\left((\alpha_1+\alpha_2)^2-(\beta_1+\beta_2)^2\right)\ln n \\
+\sum_{k=1}^\infty kV_k
V_{-k}-(\alpha_1+\alpha_2)(V(1)-V_0)
+(\beta_1+\beta_2)\ln\frac{b_+(1)}{b_-(1)}
\\+\ln\frac{G(1+\alpha_1+\alpha_2+\beta_1+\beta_2)G(1+\alpha_1+\alpha_2-\beta_1-\beta_2)}{G(1+2\alpha_1+2\alpha_2)}+\bigO(n^{-1}).
\end{multline}
Comparing (\ref{as FH2}) with (\ref{as FH1}), we observe that the terms proportional to $\ln n$ do not match unless $\alpha_1\alpha_2=\beta_1\beta_2$. Moreover we see that $E(V,\alpha_1,\alpha_2,\beta_1,\beta_2,t)$ is unbounded as $t\to 0$. These observations indicate that as $n\to\infty$ and at the same time $t\to 0$, a transition occurs in the asymptotic behavior of the Toeplitz determinants, and so the asymptotic expansion (\ref{as FH2}) is not uniformly valid for small values of $t$.
Our goal is to describe this transition between (\ref{as FH2}) and (\ref{as FH1}). We will obtain an asymptotic expansion for $D_n(f_t)$ as $n\to\infty$ uniform for
$0<t<t_0$, where $t_0>0$ is fixed and sufficiently small. In particular, this describes the double scaling limit where $n\to\infty$ and simultaneously $t\to 0$.
In the scaling $t=\frac{s}{-2in}$, where $s\in -i\mathbb R^+$ is fixed and $n\to\infty$, we will prove that $D_n(f_t)$ can be expressed asymptotically in terms of a
particular solution $\si(s)$ of the Painlev\'e V equation. Using the expansion of $\si(s)$ in the limits $s\to 0$ and $s\to -i\infty$, we recover the large $n$ asymptotics of $D_n(f_t)$ for $t=0$ and $t$ fixed, respectively.

For the transition we consider in this paper, the first observation of the appearance of a Painlev\'e V solution in the particular case of the Toeplitz determinant with symbol $f_{t/2}(z)=|z-e^{it/2}||z-e^{-it/2}|$ and a study of these objects was in
the work of Lenard, Schultz, Dyson, Tracy and Vaidya, and Jimbo, Miwa, M\^{o}ri, Sato in 1960'th-1970'th  
on the one-dimensional gas of impenetrable bosons. We discuss this in more detail later on in introduction (see the text around (\ref{Ldet})).

Thus, the present paper describes the transition between the two different types of FH asymptotics: one for Toeplitz determinants corresponding to symbols with $2$ FH singularities,
and the other for symbols with $1$ FH singularity formed by the original ones merging together along the unit circle.
This work is closely related to \cite{CIK2}, where the transition was described between the non-singular case
(Strong Szeg\H o asymptotics) and the case with one FH singularity. The transition in that case was also described by a solution to the  Painlev\'e V equation, but with different parameters and different asymptotic behavior.
Critical transitions for Toeplitz determinants have also been studied in \cite{BDJ, MTW, T, WMTB}.

\subsubsection*{Statement of results}

Before stating our results on Toeplitz determinants, we first describe the relevant Painlev\'e V transcendents.
Consider the $\sigma$-form of the Painlev\'e V equation \cite[Formula (2.8)]{ForresterWitte}
\begin{equation}\label{sigmaPV}
s^2\sigma_{ss}^2=\left(\sigma-s\sigma_s+2\sigma_s^2\right)^2
-4(\sigma_s-\theta_1)(\sigma_s-\theta_2)
(\sigma_s-\theta_3)(\sigma_s-\theta_4),
\end{equation}
where the parameters $\theta_1,\theta_2, \theta_3, \theta_4$ are given by
\begin{align}
&\theta_1=-\al_1+\frac{\beta_1+\beta_2}{2},\qquad \theta_2=\al_1+\frac{\beta_1+\beta_2}{2},\label{theta12}\\
&\theta_3=\al_2-\frac{\beta_1+\beta_2}{2},\qquad \theta_4=-\al_2-\frac{\beta_1+\beta_2}{2}.\label{theta34}
\end{align}

\begin{theorem}\label{theorem: PV}
Let $\alpha_1,\alpha_2, \beta_1, \beta_2\in\mathbb C$ be such that
\be\label{gc1}
\Re\alpha_1,\Re\alpha_2>-\frac{1}{2},\qquad \Re(\al_1+\al_2)>-\frac{1}{2},
\qquad |||\beta|||<1,
\ee
 and assume that
\begin{equation}\label{gc2}
\alpha_1\pm \beta_1,\quad \alpha_2\pm\beta_2,\quad \alpha_1+\alpha_2\pm\beta_1\pm\beta_2 \notin \{-1,-2,-3, \ldots\}.
\end{equation}
If $2(\alpha_1+\alpha_2)\notin \mathbb N\cup\{0\}$, there exists a solution $\sigma(s)$ of equation (\ref{sigmaPV})
with the following asymptotic behavior as $|s|\to 0$ along the negative imaginary axis:
\begin{multline}
\sigma(s)=2\al_1\al_2-{\frac{1}{2}}(\bt_1+\bt_2)^2
-\frac{(\al_1-\al_2)(\bt_1+\bt_2)}{2(\al_1+\al_2)}s
+\tau_0|s|^{1+2(\al_1+\al_2)}\\
+\bigO(|s|^2)+\bigO(|s|^{2+4(\al_1+\al_2)}),\qquad s\to -i0_+,\label{as0}
\end{multline}
where
{\small \begin{multline}
\tau_0=-\frac{\Ga(1+\al_1+\al_2+\bt_1+\bt_2)\Ga(1+\al_1+\al_2-\bt_1-\bt_2)}{2\pi\Ga(1+2(\al_1+\al_2))^2}\frac{\Ga(1+2\al_1)\Ga(1+2\al_2)}
{\Ga(2+2(\al_1+\al_2))}\\
\times \left[
e^{i\pi(\al_1-\al_2)}\frac{\sin\pi(\al_1+\al_2+\bt_1+\bt_2)}{\sin 2\pi(\al_1+\al_2)}
+e^{-i\pi(\al_1-\al_2)}\frac{\sin\pi(\al_1+\al_2-\bt_1-\bt_2)}{\sin 2\pi(\al_1+\al_2)}
-e^{i\pi(\bt_1-\bt_2)}\right],
\end{multline}
}
and with the following asymptotic behavior as $|s|\to\infty$ along the negative imaginary axis: 
\be
\sigma(s)
=\frac{\beta_2-\beta_1}{2}s-\frac{1}{2}(\beta_1-\beta_2)^2
\pm\frac{s\gamma(s)}{1+\gamma(s)}+\bigO(|s|^{-1+|||\beta|||}),\qquad s\to -i\infty,\label{asinfty}
\ee
where ``+'' is taken for $\Re(\bt_1-\bt_2)\ge 0$, ``-'' for  $\Re(\bt_1-\bt_2)< 0$, and
\be
\gamma(s)=\begin{cases}
\left|s\right|^{2(-1+\beta_1-\beta_2)}e^{-i|s|}e^{i\pi(\alpha_1+\alpha_2)}\frac{\Gamma(1+\alpha_1-\beta_1)\Gamma(1+\alpha_2+\beta_2)}{\Gamma(\alpha_1+\beta_1)\Gamma(\alpha_2-\beta_2)},& \Re(\beta_1-\beta_2)\geq 0,\\
\left|s\right|^{2(-1+\beta_2-\beta_1)}e^{i|s|}e^{-i\pi(\alpha_1+\alpha_2)}\frac{\Gamma(1+\alpha_2-\beta_2)\Gamma(1+\alpha_1+\beta_1)}{\Gamma(\alpha_2+\beta_2)\Gamma(\alpha_1-\beta_1)},& \Re(\beta_1-\beta_2)< 0.
\end{cases}
\ee
If $2(\alpha_1+\alpha_2)\in \mathbb N\cup\{0\}$, there exists a solution $\sigma(s)$ of equation (\ref{sigmaPV}) satisfying
\begin{equation}
\sigma(s)=2\al_1\al_2-{\frac{1}{2}}(\bt_1+\bt_2)^2+\bigO(|s\ln |s||),\qquad\mbox{ $s\to -i0_+$},\label{as0half}
\end{equation}
and satisfying (\ref{asinfty}) as $s\to -i\infty$.\\
Moreover, if $\alpha_1,\alpha_2\in\mathbb R$ and $\beta_1, \beta_2\in i\mathbb R$, there exists a solution $\sigma(s)$ which is real and free of poles for $s\in -i\mathbb R^+$, and which has the asymptotics
(\ref{as0}) (or (\ref{as0half}) if $2(\alpha_1+\alpha_2)\in\mathbb N\cup\{0\}$) and (\ref{asinfty}).
\end{theorem}

\begin{remark}
We will construct solutions $\sigma(s)$ satisfying the above properties in terms of a Riemann-Hilbert (RH) problem which depends on the parameters $\alpha_1,\alpha_2,\beta_1,\beta_2$ and on $s$. The solutions constructed in this way will be the ones appearing in the asymptotic expansion for the Toeplitz determinants $D_n(f_t)$. However, we do not prove that there is only one solution $\sigma$ which satisfies the properties given in Theorem \ref{theorem: PV}.
\end{remark}

\begin{remark}
Equation (\ref{sigmaPV}) depends, through $\theta_1,\ldots, \theta_4$, on the three independent parameters $\alpha_1,\alpha_2$, and $\beta_1+\beta_2$. On the other hand, the solutions described in the above theorem depend not only on the sum $\beta_1+\beta_2$, but also on $\beta_1$ and $\beta_2$ independently. This means that, given $\alpha_1,\alpha_2,$ and $\beta_1+\beta_2$, the asymptotics (\ref{asinfty}) and (\ref{as0}), (\ref{as0half})  specify a one-parameter family of solutions to the same differential equation (\ref{sigmaPV}).
\end{remark}

\begin{remark}
The function $\si(s)$ has a branching point at zero (any other singularities of $\si(s)$ are poles) and is defined on the plane with a cut from zero to infinity.
The assumption in the theorem that $s$ is on the negative imaginary axis is not essential: it is adopted for simplicity and in view of the application in Theorem \ref{theorem: Toeplitz} below. A simple modification 
of the proof shows that the asymptotics (\ref{as0}), (\ref{asinfty}), (\ref{as0half}) hold along a path from 
zero to infinity in a neighborhood of the negative imaginary axis. This fact is used in Theorem \ref{theorem: Toeplitz 2} below.
\end{remark}

We now state the result about Toeplitz determinants for the case $\al_j,i\bt_j\in\mathbb R$.

\begin{theorem}\label{theorem: Toeplitz} Let $\alpha_1,\alpha_2,\al_1+\al_2>-1/2$
and $\beta_1,\beta_2\in i\mathbb R$.
Let $D_n(f_t)$ be the Toeplitz determinant (\ref{Toeplitz}) corresponding to the symbol (\ref{symbol}).
The following asymptotic expansion holds as $n\to\infty$ with the error term 
uniform for $t\in (0,t_0)$, where $t_0$ is sufficiently small:
\begin{multline}\label{asT}
\ln D_n(f_t)=\ln D_n(f_0)+int(\beta_2-\beta_1)+
\int_{0}^{-2int}\frac{1}{s}\left(\sigma(s)-2\alpha_1\alpha_2+\frac{1}{2}(\beta_1+\beta_2)^2\right)ds\\
+2\left(\beta_1\beta_2-\alpha_1\alpha_2\right)\ln\frac{\sin t}{t}
+2it(\alpha_2\beta_1-\alpha_1\beta_2)
-\alpha_1(V(e^{it})-V(1))\\-\alpha_2(V(e^{-it})-V(1))
+\beta_1\ln\frac{b_+(e^{it})b_-(1)}{b_-(e^{it})b_+(1)}
+\beta_2\ln\frac{b_+(e^{-it})b_-(1)}{b_-(e^{-it})b_+(1)}+o(1),
\end{multline}
where the function $\sigma(s)$ satisfies the conditions of Theorem \ref{theorem: PV}: it solves equation (\ref{sigmaPV}), has the asymptotics (\ref{as0}) if $2(\alpha_1+\alpha_2)\notin\mathbb N\cup\{0\}$
((\ref{as0half}) otherwise) and (\ref{asinfty}), and has no poles for $s\in -i\mathbb R^+$. Here $\ln D_n(f_0)$ is given by (\ref{as FH1}).
\end{theorem}

\begin{remark}
The integral in (\ref{asT}) is well-defined by (\ref{as0}), (\ref{as0half}), and by the fact that $\sigma$ has no poles on the interval of integration.
\end{remark}

\begin{remark}\label{remark: degenerate}
If $\alpha_1=\alpha_2=\beta_1=\beta_2=\frac{1}{2}$, the function $\sigma(s)$ is identically zero, as we show in Section \ref{section: degenerate}. Note that in this case, the parameters $\theta_1,\ldots, \theta_4$ in the Painlev\'e equation (\ref{sigmaPV}) are given by $\theta_1=\theta_3=0$, $\theta_2=1$, 
$\theta_4=-1$, and it is easily verified that $\sigma(s)=0$ solves (\ref{sigmaPV}), and that it satisfies the asymptotic conditions (\ref{as0half}) and (\ref{asinfty}). 
Although $\beta_1,\beta_2\notin i\mathbb R$ in this case, the asymptotic expansion (\ref{asT}) holds and becomes elementary. 
\end{remark}

An extension of the previous theorem to the generic case $|||\beta|||<1$ is the following.

\begin{theorem}\label{theorem: Toeplitz 2}
Let $\Re\alpha_1,\Re\alpha_2,\Re(\al_1+\al_2)>-1/2$,  $|||\beta|||<1$,
and
\[
\alpha_1\pm \beta_1, \alpha_2\pm\beta_2, \alpha_1+\alpha_2\pm\beta_1\pm\beta_2 \notin \{-1,-2,-3, \ldots\}.
\]
Let $D_n(f_t)$ be the Toeplitz determinant (\ref{Toeplitz}) corresponding to the symbol (\ref{symbol}).
There exists a finite set $\Omega=\{s_1,\dots,s_\ell\}\subset -i\mathbb R^+=(0,-i\infty)$ (with $\ell=\ell(\al_1,\al_2,\bt_1,\bt_2)$)
such that the asymptotic expansion (\ref{asT}) holds uniformly for $t\in (0,t_0)$, where $t_0$ is sufficiently small, provided
$-2int$ is bounded away from $\Omega$. The path of integration in the integral on the r.h.s.
of (\ref{asT}) is chosen in the complex $s$-plane to avoid the points of $\Omega$.
The function $\si(s)$ is a solution to (\ref{sigmaPV})
with the asymptotics (\ref{as0}) (or (\ref{as0half}) if $2(\alpha_1+\alpha_2)\in\mathbb N\cup\{0\}$) and (\ref{asinfty}).
\end{theorem}

\begin{remark}
The set $\Omega$ is the set of points where the Riemann-Hilbert problem associated to $\si(s)$ is not solvable. 
$\Omega$ contains the poles of $\si(s)$. 
A pole $s_j$ corresponds to a
zero in the asymptotics of the determinant $D_n(f_t)$ for $t_j=is_j/(2n)$.
Different choices of the integration contour in (\ref{asT}) correspond to different branches of $\ln D_n$.
For $\al_j, i\bt_j\in\mathbb R$, we show in Section \ref{secsolvability} that $\Omega$ has no points on the half-line $-i\mathbb R^+$, and hence a simpler formulation of the result in Theorem \ref{theorem: Toeplitz}. 
\end{remark}

\begin{remark}
An estimate for the error term in (\ref{asT}) for both theorems is given in the Proposition \ref{prop diff id F} below.
\end{remark}

If $t\to 0$ sufficiently fast so that also $nt\to 0$, we immediately obtain (\ref{as FH1})
from (\ref{asT}).
Let us check that we also recover (\ref{as FH2}) from  (\ref{asT}) when $t$ is fixed, and so $nt\to\infty$.
Note first that it follows from the asymptotics for $\sigma$ that, given (\ref{gc1}), (\ref{gc2}),
\begin{align}
&\int_{0}^{-2int}\frac{1}{s}\left(\sigma(s)-2\alpha_1\alpha_2+\frac{1}{2}(\beta_1+\beta_2)^2\right)ds\nonumber \\
&\qquad\quad =-int(\beta_2-\beta_1)-\left(\frac{1}{2}(\bt_1-\bt_2)^2+\sigma(0)\right)\ln(2nt) +\bigO(1)\nonumber \\
&\qquad\quad =
-int(\beta_2-\beta_1)-2(\al_1\al_2-\bt_1\bt_2)\ln(2nt)+\bigO(1),\qquad nt\to \infty.
\label{Omega infty}
\end{align}
Substituting this expression into the right hand side of (\ref{asT}), we
obtain the terms with $n$ and with $\ln n$ in (\ref{as FH2}). Equality of the constant in $n$ terms in both expressions for $t$ fixed
gives the following integral identity for $\sigma(s)$:
{\small \begin{multline}
\lim_{T\to +\infty}\left(\int_0^{-iT}\frac{1}{s}\left(\sigma(s)-2\alpha_1\alpha_2+\frac{1}{2}(\beta_1+\beta_2)^2\right)ds +\frac{iT}{2}(\beta_2-\beta_1)+2(\alpha_1\alpha_2-\beta_1\beta_2)\ln T\right)\\
\qquad =i\pi(\alpha_1\beta_2-\alpha_2\beta_1)-\ln\frac{G(1+\alpha_1+\alpha_1+\beta_1+\beta_2)G(1+\alpha_1+\alpha_1-\beta_1-\beta_2)}{G(1+2\alpha_1+2\alpha_2)}
\\+\ln\frac{G(1+\alpha_1+\beta_1)G(1+\alpha_1-\beta_1)G(1+\alpha_2+\beta_2)G(1+\alpha_2-\beta_2)}{G(1+2\alpha_1)G(1+2\alpha_2)}.
\label{intidentity}\end{multline}
}This identity is a deep result which contains global information about $\sigma$. We believe that it is of independent interest in the study of Painlev\'e transcendents.

The following result extends the expansion (\ref{as FH2}), known for fixed singularities $z_1$, $z_2$ independent of $n$, to the case where the two singularities approach each other at a sufficiently slow rate as $n\to\infty$.

\begin{theorem}\label{theorem: extension}
Let $\Re\alpha_1,\Re\alpha_2 > -1/2$,  $|||\beta|||<1$,
and
$\alpha_1\pm \beta_1$, $\alpha_2\pm\beta_2 \notin \{-1,-2,\ldots\}$.
Let
$D_n(f_t)$ be the Toeplitz determinant (\ref{Toeplitz}) corresponding to the symbol (\ref{symbol});
$\om(x)$ be any positive, smooth for large $x$ function such that $\om(n)\to\infty$,
$\om(n)/n\to 0$  as $n\to\infty$. 

Then the expansion (\ref{as FH2}) holds as $n\to\infty$
for $\om(n)/n\le t<t_0$ with $t_0$ sufficiently small.
The error term in (\ref{as FH2})
in this case is $o(1)=\bigO(\om(n)^{-1+|||\bt|||})$, uniformly in $t$.

Moreover, there exist positive constants $n_0$, $s_0$, $c_0$, depending only on 
$\al_j$, $\bt_j$, $j=1,2$, and $V(z)$, such that
the expansion (\ref{as FH2}) holds for $n>n_0$ with the error term $o(1)$ replaced by a function
$\upsilon(s)$, $s=nt$, satisfying the estimate $|\upsilon(s)|<c_0 s^{-1+|||\bt|||}$ for $s\ge s_0$, $t<t_0$.
\end{theorem}

To complete the analysis of the nondegenerate (by which we mean that the condition (\ref{gc2}) holds) situation it remains to
consider the case $|||\bt|||=1$. We have

\begin{theorem}\label{theorem: Toeplitz 3}
Let $\Re\alpha_1,\Re\alpha_2,\Re(\al_1+\al_2)>-1/2$,
and assume (\ref{gc2}).
Let $|||\bt|||=0$, and denote $\bt_1^-=\bt_1$, $\bt_2^-=\bt_2-1$,
$f^{-}_t=f_t(z;\al_1,\al_2,\bt^-_1,\bt^-_2)$, $f_t=f_t(z;\al_1,\al_2,\bt_1,\bt_2)$.
Then $|||\bt^-|||=1$.
There exists a sufficiently large $C_0$
such that the following asymptotic expansion holds outside the set $\Omega$ of Theorem
\ref{theorem: Toeplitz 2}:
\begin{multline}\label{asDb1}
D_{n-1}(f^{-}_t)=e^{-i(n-1)t}b_0^{-1}D_n(f_t)\\
\times
\begin{cases}
-r(-2int)\frac{b_-(1)}{b_+(1)}
 t\left(\frac{nt}{\sin t}\right)^{2(\bt_1+\bt_2)}
 e^{i\pi(-\al_1+3\bt_1+\al_2+\bt_2)}(1+\bigO(t)),& 0<t\le C_0/n \cr
 n^{2\bt_1-1}z_1^{-n+1}\frac{b_-(z_1)}{b_+(z_1)}
\frac{\Gamma(1+\al_1-\bt_1)}{\Gamma(\al_1+\bt_1)}e^{i(\pi-2t)\al_2}(2\sin t)^{-2\bt_2}\cr
\qquad\qquad\times(1+\bigO((nt)^{-1}))\cr
\qquad +n^{2\bt_2-1}z_2^{-n+1}\frac{b_-(z_2)}{b_+(z_2)}
\frac{\Gamma(1+\al_2-\bt_2)}{\Gamma(\al_2+\bt_2)}e^{i(-\pi+2t)\al_1}(2\sin t)^{-2\bt_1}\cr
\qquad\qquad\times
(1+\bigO((nt)^{-1})),& C_0/n<t<t_0,
 \end{cases}
 \end{multline}
as $n\to\infty$,
with the error term uniform for $-2int$ bounded away from $\Omega$.
Here the asymptotics for $D_n(f_t)$ are given by Theorem \ref{theorem: Toeplitz 2}, and $r(s)$ is a Painlev\'e V function defined in Section
\ref{section Lax}. In particular, $r(s)$ is related to $\si(s)$ by (\ref{identity r}), and has the
large-$s$ asymptotics (\ref{as r}) and the small-$s$ asymptotics (\ref{r small s}).
\end{theorem}

\begin{remark}
The large-$s$ expansion for $r(s)$ implies, by Remark \ref{2beta1} below,
that the 2 parts of the asymptotics (\ref{asDb1}) coincide in a neighborhood of the boundary
$t=C_0/n$. Thus (\ref{asDb1}) is a complete analogue of (\ref{asT}).
\end{remark}

\begin{remark}\label{remark: degenerate 2}
If $\alpha_1=\alpha_2=\beta_1=\beta_2=\frac{1}{2}$, the function $r(s)$ is elementary,
namely (as we obtain in Section \ref{sec sp case}),
\[
r(s)= -\frac{\sin nt}{(nt)^2},\qquad s=-2int.
\]
In Section \ref{sec sp case} we then show that, in this case, the part of (\ref{asDb1}) for
$C_0/n<t<t_0$ holds uniformly for the whole range $0<t<t_0$. This fact was used
in \cite{DIK3} to analyze the eigenvalues of the Toeplitz matrix $T_n(g)$ with a smooth
real-valued symbol $g$ in a small neighbourhood of the edge of the spectrum.
The rest of the eigenvalues of $T_n(g)$ were analyzed in \cite{DIK3}
using the results of \cite{DIK2} for FH singularities at a nonzero distance from each other.
\end{remark}

Finally, let us verify that (\ref{asDb1}) reduces to (\ref{as FH1}) as $|s|\to 0$, and to (\ref{as FH2b1})
as $|s|\to\infty$. In the former case, we substitute the small $s$ asymptotics (\ref{r small s}) of $r(s)$
and the formula (\ref{as FH1}) for $D_n(f_t)$ into (\ref{asDb1}) and obtain by a straightforward calculation which uses the property $G(z+1)=\Gamma(z)G(z)$
of the Barnes G-function that $D_{n-1}(f_t^{-})$
is given by  (\ref{as FH1}) with $n$ replaced by $n-1$ and with $\bt_2$ replaced by $\bt_2^{-}$.
Consider now $|s|\to\infty$. We have $\bt'=\bt^{-}$, and $\bt_1''=\bt_1^{-}-1=\bt_1-1$, 
$\bt_2''=\bt_2^{-}+1=\bt_2$.
Using the expansion (\ref{as FH2}) for $D_n(f_t)$ and the second part of  (\ref{asDb1}), we obtain
(\ref{as FH2b1}) for $D_{n-1}(f_t^{-})$.
In particular, this extends the validity of 
(\ref{as FH2b1}) (with the appropriately changed estimate for the error term): cf. Theorem \ref{theorem: extension} above.

\subsubsection*{Applications}
In view of the applications we discuss below, consider a special case $V(z)=0$, $\bt_1=\bt_2=0$, $\al_1=\al_2\equiv\alpha\in\mathbb R$. We first prove the following.

\begin{theorem}\label{theorem: integral}
Let
\be\label{f-sp}
f_t(z)=|z-e^{it}|^{2\al}|z-e^{-it}|^{2\al},\qquad \al>-\frac{1}{4},\quad t\in\mathbb R.
\ee
Let $0<t_1<\pi$.
Then, as $n\to\infty$,
\be\label{intall}
\int_0^{t_1}D_n(f_t)dt=
\begin{cases}
C_1(t_1,\al) n^{2\al^2}(1+o(1))& \mbox{if}\quad 2\al^2<1,\cr
C_2 n\ln n(1+o(1))& \mbox{if}\quad 2\al^2=1,\cr
\int_0^{\om(n)/n}D_n(f_t)dt(1+o(1))=
C_3(\al) n^{4\al^2-1}(1+o(1))& \mbox{if}\quad 2\al^2>1.
\end{cases}
\ee
Here $\om(x)$ is any positive, smooth for large $x$ function such that $\om(n)\to\infty$,
$\om(n)/n\to 0$ as $n\to\infty$; $C_1$, $C_2$, $C_3$ are positive constants
\begin{align}
&C_1(t_1,\al)=\frac{G(1+\al)^4}{2^{2\al^2}G(1+2\al)^2}\int_0^{t_1}(\sin t)^{-2\al^2} dt,\\
&C_2=\frac{G(1+\frac{1}{\sqrt{2}})^4}{2G(1+\sqrt{2})^2},\\
&C_3(\al)=\frac{G(1+2\al)^2}{G(1+4\al)}\left[
\int_0^1\exp\left\{\int_0^{-2iu}\frac{\si(s)-2\al^2}{s}ds\right\}du\right.\nonumber\\ 
&\qquad\quad\left.+
\exp\left\{\int_0^{-2i}\frac{\si(s)-2\al^2}{s}ds\right\}\int_1^\infty 
\exp\left\{\int_{-2i}^{-2iu}\frac{\si(s)}{s}ds\right\}u^{-2\al^2}du\right],
\end{align}
where $\si(s)$ (real-valued for $s\in -i\mathbb R_+$) is the solution to the Painlev\'e V equation appearing
in (\ref{asT}) with the parameters $\al_1=\al_2=\al$, $\bt_1=\bt_2=0$.  
\end{theorem}

\begin{remark}
In the case $2\al^2<1$,
the leading order asymptotic term for the integral comes from the expansion (\ref{as FH2}), i.e., from the integration outside of
a contracting neighborhood $[0, \om(n)/n)$, whereas in the case  $2\al^2>1$, the leading order asymptotic term 
comes from the integration over this neighborhood. 
\end{remark}

\begin{proof}
Note first, as follows from (\ref{asT}) and the positivity of $D_n(f_t)$ for real-valued symbols $f_t$,
that $\si(s)$ is real-valued for $s\in -i\mathbb R_+$ with our choice (\ref{f-sp}) of $f_t$.
Moreover, we have by Theorem \ref{theorem: PV},
\begin{align}
&\si(s)=\bigO(|s|^{-1}),&& s\to -i\infty,\\ &\si(s)=2\al^2+\bigO(|s|^{1+4\al})+\bigO(|s\ln|s||),&& s\to-i0.
\end{align}

We divide the interval of integration $t\in (0,t_1)$ into 3 regions, $0<nt\le 1$, $1<nt\le\om(n)$, $\om(n)/n<t\le t_1$.

For $0<nt\le 1$ (and, in general, for $0<t<t_0$), we obtain from Theorem \ref{theorem: Toeplitz} setting
$\al_1=\al_2=\al$, $\bt_1=\bt_2=0$, $V(z)=0$ in (\ref{asT}) and (\ref{as FH1}):
\be\label{begint1}
\ln D_n(f_t)= 4\al^2\ln n +\ln\frac{G(1+2\al)^2}{G(1+4\al)}+
\int_{0}^{-2int}\frac{\sigma(s)-2\al^2}{s}ds
-2\al^2\ln\frac{\sin t}{t}+o(1),
\ee
as $n\to\infty$.
Note that both $\int_{0}^{-2int}(\sigma(s)-2\al^2)\frac{ds}{s}$ and $\ln\frac{\sin t}{t}$ are uniformly bounded for 
$0<nt\le 1$.

For $1<nt\le \om(n)$, we write the above formula in the form:
\begin{multline}\label{begint2}
\ln D_n(f_t)= 2\al^2\ln n +\ln\frac{G(1+2\al)^2}{G(1+4\al)}+\int_{0}^{-2i}\frac{\sigma(s)-2\al^2}{s}ds+\\
\int_{-2i}^{-2int}\frac{\sigma(s)}{s}ds
-2\al^2\ln\sin t+o(1),
\end{multline}
as $n\to\infty$,
and note that $\int_{-2i}^{-2int}\frac{\sigma(s)}{s}ds$ is uniformly bounded for $1<nt\le\om(n)$.

For $\om(n)/n<t\le t_1$, by Theorem \ref{theorem: extension},
we can use the expansion (\ref{as FH2}) for $\ln D_n(f_t)$.

We are now ready to compute the integral. First, using (\ref{begint1}), replacing $(\sin t)/t$ by $1$ 
to the leading order, and changing the integration variable $t=u/n$, we obtain
\be\label{int1}
\int_0^{1/n} D_n(f_t)dt=
n^{4\al^2-1}\frac{G(1+2\al)^2}{G(1+4\al)}
\int_0^1\exp\left\{\int_0^{-2iu}\frac{\si(s)-2\al^2}{s}ds\right\}du(1+o(1)).
\ee
Next, using (\ref{as FH2}) and Theorem \ref{theorem: extension} (uniformity of the error term
in the interval $(t_0,t_1)$ follows from the analysis in \cite{DIK22}), we obtain
\be\label{int3}
\int_{\om(n)/n}^{t_1} D_n(f_t)dt=
n^{2\al^2}\frac{G(1+\al)^4}{2^{2\al^2}G(1+2\al)^2}
\int_{\om(n)/n}^{t_1}(\sin t)^{-2\al^2} dt(1+o(1)).
\ee
Finally, by (\ref{begint2}),
\begin{multline}\label{int2}
\int_{1/n}^{\om(n)/n} D_n(f_t)dt=
n^{2\al^2}\frac{G(1+2\al)^2}{G(1+4\al)}\\
\times \ \exp\left\{\int_0^{-2i}\frac{\si(s)-2\al^2}{s}ds\right\}
\int_{1/n}^{\om(n)/n}\psi(t)t^{-2\al^2}dt\ (1+o(1)),
\end{multline}
where
\[
\psi(t)=\exp\left\{\int_{-2i}^{-2int}\frac{\si(s)}{s}ds\right\}\left(\frac{\sin t}{t}\right)^{-2\al^2}
\]
is bounded and bounded away from zero, uniformly for $1/n<t<t_0$. 
The integration region $(\frac{1}{n},\frac{\om(n)}{n})$ is the most interesting one.
If $2\al^2<1$,  the rightmost integral in (\ref{int2}) converges at zero, and we can write (\ref{int2}) as follows:
\be\label{int22}
\int_{1/n}^{\om(n)/n} D_n(f_t)dt=o(n^{2\al^2}),\qquad 2\al^2<1.
\ee
This formula together with (\ref{int1}) and (\ref{int3}) proves the theorem in the case of $2\al^2<1$.
We see that the contributions of (\ref{int1}) and (\ref{int22}) are only subleading.

If $2\al^2>1$, the rightmost integral in (\ref{int2}) does not converge at zero. We write
\be\label{t-u}
\int_{1/n}^{\om(n)/n}\psi(t)t^{-2\al^2}dt=n^{2\al^2-1}\int_{1}^{\om(n)}\psi(u/n)u^{-2\al^2}du,
\ee
where the integral in the r.h.s. converges at infinity.
We have
\begin{multline}
\int_{1}^{\om(n)}\psi(u/n)u^{-2\al^2}du=\int_{1}^{\om(n)}\exp\left\{\int_{-2i}^{-2iu}\frac{\si(s)}{s}ds\right\}u^{-2\al^2}du
(1+\bigO([\om(n)/n]^{2}))\\=\int_{1}^{\infty}\exp\left\{\int_{-2i}^{-2iu}\frac{\si(s)}{s}ds\right\}u^{-2\al^2}du
(1+o(1)).
\end{multline}
Substituting this into (\ref{t-u}), and that into (\ref{int2}), and adding the contribution of (\ref{int1}), we obtain (\ref{intall}) for $2\al^2>1$:
a simple analysis of (\ref{int3}) shows that it gives only a subleading in $n$ contribution.

If $2\alpha^2=1$, the integral (\ref{int1}) is
\be\label{int1bis}
\int_0^{1/n} D_n(f_t)dt=\bigO(n).
\ee
We rewrite the integral (\ref{int3}) for $2\alpha^2=1$ (by adding and subtracting $1/t$ in the integral on 
the r.h.s.) as follows:
\be\label{int3bis}
\int_{\om(n)/n}^{t_1} D_n(f_t)dt=
n\frac{G(1+\frac{1}{\sqrt{2}})^4}{2G(1+\sqrt{2})^2}
\left(\int_{0}^{t_1}\left(\frac{1}{\sin t}-\frac{1}{t}\right) dt+\ln\frac{nt_1}{\omega(n)}\right)(1+o(1)).
\ee
For $2\alpha^2=1$, the integral in the r.h.s. of (\ref{t-u}) does not converge at infinity. We then add and subtract
from the integrand $u^{-1}\exp\{\int_{-2i}^{-i\infty}\si(s)s^{-1}ds\}$. Substituting the result into (\ref{int2}) and using the identity (\ref{intidentity}), we obtain
\begin{multline}\label{int2bis}
\int_{1/n}^{\om(n)/n} D_n(f_t)dt=
n\frac{G(1+\frac{1}{\sqrt{2}})^4}{2G(1+\sqrt{2})^2}\left[\ln \om(n)\right.\\
\left.+\int_1^{\infty}\left(\exp\left\{-\int_{-2iu}^{-i\infty}\frac{\si(s)}{s}ds\right\}-1\right)\frac{du}{u}\right](1+o(1)).
\end{multline}
Adding (\ref{int1bis}), (\ref{int2bis}), and (\ref{int3bis}) together, we obtain the statement of the theorem for 
the case $2\alpha^2=1$. (Note that the contribution of the terms of order $n\log\om(n)$ cancels in the sum.) 
This completes the proof of (\ref{intall}).
\end{proof}

Theorem \ref{theorem: integral} is relevant for some problems in random matrices, number theory, and statistical physics.

\medskip

Let $p_n(\th)=\det(U-e^{i\th}I)$ be the characteristic polynomial of an $n\times n$ matrix $U$
from the circular unitary ensemble of random matrices, i.e., the distribution of $U$ is given
by the Haar measure on the unitary group.
There exists a large body of conjectural
evidence (see \cite{FyodorovKeating} and references therein, see also \cite{K-hankel}) which relates $p_n(\th)$ for large $n$ to the Riemann $\zeta$-function on the critical line $\zeta(1/2+ix)$. First was an observation of Keating and Snaith \cite{KS} that the averages $\int_0^T |\zeta(1/2+ix)|^{2\al} dx$ behave for large $T$ in a similar way as the expectation 
\be\label{m1}
\mathbb E\left\{ |p_n(\th)|^{2\al}\right\}
\ee 
does for large $n$ (note that it does not depend on $\th$
due to the rotational symmetry). It follows by the Heine-type multi-integral representation
for Toeplitz determinants
\be\label{Heine}
D_n(f)=
\frac{1}{n!} \int^{2\pi}_0 \dots \int^{2\pi}_0 \prod_{0\le j < k \le n-1}
\left|e^{i\th_j} - e^{i\th_k} \right|^2 \prod^{n-1}_{j=0} f \left(e^{i\th_j}\right)
\,\frac{d\th_j}{2\pi}
\ee
that the expectation (\ref{m1}) is the Toeplitz determinant $D_n(f)$
with symbol $f(z)=|z-1|^{2\al}$, and therefore its large $n$ behavior is given by a particular
case of (\ref{as FH1}).  
The observation of  Keating and Snaith enabled them to make a remarkably
detailed conjecture on the large $T$ asymptotics of the averages of the $\zeta$-function,
in particular, to predict the appearance of Barnes' G-function in the formula.

In a similar vein, it is argued in \cite{FyodorovKeating} that maximal values of $|\zeta(1/2+ix)|$ over an interval of the critical line, a classical problem, are related to the distribution of large values of $|p_n(\th)|$.
The characteristic polynomial $p_n(\th)$ also models \cite{FyodorovKeating} extreme properties
of the Gaussian free field and the $1/f$-noise, and is related
to the question of the so-called freezing transition in statistical models.
In this connection, one would like to estimate the moments (cf. (\ref{m1})) \cite[formula (67)]{FyodorovKeating}: 
\[
M_k=\int_0^L d\th_1 \cdots\int_0^L d\th_k \mathbb E \left\{ |p_n(\th_1)|^{2\al}\cdots |p_n(\th_k)|^{2\al}\right\},\qquad k=2,3,\dots.
\]
The expectation inside the integrals is the Toeplitz determinant $D_n(f)$ with symbol
$f(z)=\prod_{j=1}^k |z-e^{i\th_j}|^{2\al}$. Let $k=2$ and fix $L>0$. Then, 
using Theorem \ref{theorem: integral}  and the invariance of the
determinant with respect to rotations of the circle, we immediately obtain that the second moment
$M_2=\bigO(n^{2\al^2})$ for $2\al^2<1$, and  $M_2=\bigO(n^{4\al^2-1})$ for $2\al^2>1$. This 
behavior of $M_2$ was conjectured by Fyodorov and Keating \cite{FyodorovKeating}. 
Note that there are further interesting 
conjectures in \cite{FyodorovKeating} about the asymptotics of a general moment $M_k$ and about the 
distribution of large values of $|p_n(\th)|$, but we do not address them here.  

\medskip

Another application of  Theorem \ref{theorem: integral} is to the problem of Bose-Einstein condensation.
Consider the one-dimensional gas of impenetrable bosons.  This particle system was introduced
by Girardeau \cite{Gir} in 1960. It is one of a very few many-body systems which allow exact analysis without
any approximations. (It is a limiting case of the one-dimensional Bose gas with $\delta$-function interactions,
a model introduced later by Lieb and Liniger \cite{LL}).
Namely, consider the following system of $n\ge 2$ particles in one dimension in a box of size $L$: the wavefunction $\psi(x_1,\dots,x_n)$ obeys the free-particle Schr\"odinger equation, $\psi$ satisfies the periodic boundary conditions with period $L$, 
$\psi$ is symmetric with respect to interchange of particles, $\psi=0$ whenever two particle coordinates coincide. Then the wavefunction of the ground state of the system is the following:
\[
\psi_0(x_1,\dots,x_n)=(n! L^n)^{-1/2}\prod_{1\le j<k\le n}|e^{2\pi ix_j/L}-e^{2\pi ix_k/L}|.
\]
The one-particle reduced density matrix is given by
\be\label{Lrho}
\rho^{(n,L)}(x-y)=n\int_0^Ldx_1\cdots\int_0^Ldx_{n-1}\psi_0(x_1,\dots,x_{n-1},x)\psi_0(x_1,\dots,x_{n-1},y).
\ee

Let $R_n(t)$ be defined by the identity
\be\label{LdefR}
\rho^{(n,L)}(\xi)=\frac{1}{L}R_n\left(\frac{2\pi\xi}{L}\right).
\ee
It follows from (\ref{Lrho})
and the Heine representation (\ref{Heine}) that $R_n(t)$ is the Toeplitz determinant 
\be\label{Ldet}
R_n(t)=D_{n-1}(f_{t/2}),\qquad f_{t/2}(z)=|z-e^{it/2}||z-e^{-it/2}|.
\ee
This fact was first noticed by Lenard \cite{Len1} in 1963. He used it to prove the absence of Bose-Einstein condensation in the ground state (and thus to confirm a result of Schultz \cite{Schultz} who showed the absence of condensation by another method)
as follows.

The Fourier coefficient of $\rho^{(n,L)}(\xi)$
\[
\rho_k=\int_0^L \rho^{(n,L)}(\xi) e^{-2\pi i k\xi/L} d\xi,\qquad k=0,\pm 1,\dots,
\]
is the expectation value of the number of particles with momentum $2\pi k/L$.
According to the criterion of Penrose and Onsager, there is no Bose-Einstein condensation if
the largest eigenvalue of the density matrix is less than of order $n$ as $n=L\to\infty$.
Since $\rho^{(n,n)}(x-y)$ is translationally invariant, its eigenvalues are the Fourier coefficients $\rho_k$.
The largest eigenvalue is $\rho_0$. Thus, by (\ref{LdefR}), (\ref{Ldet}), Lenard had to evaluate the integral
\be\label{Lrho0}
\rho_0=\frac{1}{2\pi}\int_0^{2\pi}R_n(t)dt=
\frac{1}{2\pi}\int_0^{2\pi}D_{n-1}(f_{t/2})dt,\qquad  f_{t/2}(z)=|z-e^{it/2}||z-e^{-it/2}|.
\ee
At that time, in 1963, even (\ref{as FH2}) was not known. 
However, Szeg\H o obtained the bound (see \cite{DIK-hist} for a historical account):
\[
R_n(t)<\left|\frac{en}{\sin(t/2)}\right|^{1/2}.
\]
Substituting this into (\ref{Lrho0}), Lenard observed that $\rho_0=\bigO(n^{1/2})$, which implies, 
in particular, that there is no Bose-Einstein condensation in the ground state.

The question of precise large $n$ asymptotics of the largest (zero-momentum) occupation 
number $\rho_0$ was addressed by Dyson \cite{Dy3}. Using a Coulomb gas interpretation
of $D_{n-1}(f_{t/2})$ and physical arguments (see \cite{DIK-hist} for details), Dyson conjectured that
\be\label{LDy}
\rho_0=C_D n^{1/2}(1+o(1)),\qquad C_D=\left({e\over\pi}\right)^{1/2}2^{-5/6}A^{-6}
\Gamma\left({1\over 4}\right)^2,
\ee
where $A=e^{\frac{1}{12}} e^{-\zeta'(-1)}$ is Glaisher's constant.

We are now in a position to verify this conjecture. Indeed, it follows from (\ref{Lrho0}), from Theorem \ref{theorem: integral} with $\al=1/2$ ($2\al^2=2(1/2)^2=1/2<1$), and from well-known formulae for Barnes' 
G-function that
\begin{multline}
\rho_0=\frac{1}{2\pi}\int_0^{2\pi}D_{n-1}(f_{t/2})dt=
\frac{2}{\pi}\int_0^{\pi/2}D_{n-1}(f_{t})dt\\=
\frac{\sqrt{2n}}{\pi}G(1+1/2)^4\int_0^{\pi/2}(\sin t)^{-1/2} dt (1+o(1))=
\frac{\sqrt{n\pi}}{2}\Gamma(1/4)^2 G(1/2)^4(1+o(1))\\=
\left({en\over\pi}\right)^{1/2}2^{-5/6}A^{-6}
\Gamma\left({1\over 4}\right)^2 (1+o(1)),
\end{multline}
which proves Dyson's conjecture (\ref{LDy}).

For further related conjectures, which can now be approached by similar methods, see 
the work of Forrester, Frankel, Garoni, and Witte \cite{FFGW}.

In \cite{Schultz}, Schultz made his conclusion about the absence of Bose-Einstein condensation by
relating the system to the $XY$ spin-1/2 chain, another exactly solvable model. Manipulating resulting
formulae, Lenard \cite{Len1} obtained an expression for the double-scaling limit
$\rho(\xi)=\lim_{n\to\infty}\rho^{(n,n)}(\xi)$ in terms of a Fredholm minor with a sine-kernel. 
The small and large $\xi$ 
behaviors of $\rho(\xi)$ were analyzed by Vaidya and Tracy in \cite{VaiTr} by identifying $\rho(\xi)$
with a 2-point correlation function of the $XY$ spin chain. Jimbo, Miwa, M\^{o}ri, and Sato \cite{JMMS} then
showed that  $\rho(\xi)$ satisfies Painlev\'e V. Let us compare the formulae of \cite{JMMS} 
with our results. We obtain from (\ref{LdefR}), (\ref{Ldet}), using (\ref{asT}) with $V\equiv 0$,
$\al_1=\al_2=1/2$, $\bt_1=\bt_2=0$, that
\[
 \rho(\xi)=\lim_{n\to\infty}\rho^{(n,n)}(\xi)=\exp\left\{\int_0^{\pi\xi}\widehat\si(x)\frac{dx}{x}\right\},\qquad
 \widehat\si(x)=\si(-2ix)-1/2,
\]
and we have that $\widehat\si(x)=x{d\over dx}\ln\rho(x/\pi)$.
Now it is easy to verify that (\ref{sigmaPV}) for $\al_1=\al_2=1/2$, $\bt_1=\bt_2=0$, reduces
to \cite[Formula (2.22) with changed notation: $\si\to\widehat\si$]{JMMS}, and that the small and large $s$ expansions of Theorem \ref{theorem: PV}
in the case $\al_1=\al_2=1/2$, $\bt_1=\bt_2=0$
are consistent with \cite[Formula (2.23)]{JMMS} and the expansions of \cite{VaiTr}. Note that the latter are much more detailed than Theorem \ref{theorem: PV} in this case.
Theorem \ref{theorem: Toeplitz} completes the analysis of the scaling limit for $\rho^{(n,n)}(\xi)$ 
((\ref{LdefR}), (\ref{Ldet})) in \cite{Len1,Schultz,VaiTr,JMMS} by providing {\it uniform} asymptotics for $0<t<t_0$.

\subsubsection*{Outline}

In Section \ref{secY}, we relate the Toeplitz determinants $D_n(f_t)$ to a system of polynomials orthogonal 
with weight $f_t(z)$ on the unit circle and characterize those polynomials in terms of a Riemann-Hilbert (RH) problem. We obtain an identity for $\frac{d}{dt}\ln D_n(f_t)$ in terms of the solution $Y$ to the RH problem for the orthogonal polynomials. 

In Section \ref{secPsi}, we characterize the Painlev\'e V transcendent $\sigma(s)$ in terms of a RH problem, and we show that this RH problem is solvable for certain values of the parameters. In Section \ref{secaux}, we state an auxiliary RH problem for the confluent hypergeometric function; we use this RH problem in Section \ref{sec large s} and Section \ref{sec s0} to obtain large $s$ and small $s$ asymptotics, respectively, for the Painlev\'e RH problem and for the Painlev\'e function $\sigma$. This gives the proof of Theorem \ref{theorem: PV}. The asymptotics for large $s$ are built out of 3 parametrices: a global one given in terms of elementary functions and 2 local ones around $z_k$, in terms of confluent hypergeometric functions. 
The asymptotics for small $s$ are built out of 2 parametrices: the global one, in terms of a 
confluent hypergeometric function, and the local one, in terms of a hypergeometric function (we have 2 singularities in the same neighborhood in this case). 

In Section \ref{section: RHP Y}, we solve the RH problem for the orthogonal polynomials for large $n$ uniformly for $0<t<t_0$, and in Section \ref{section: as Toep}, we use this solution to prove Theorem \ref{theorem: Toeplitz} and Theorem \ref{theorem: Toeplitz 2}. In Section \ref{section: beta1}, we prove Theorem \ref{theorem: Toeplitz 3}.

The outline of the present paper is similar to that of \cite{CIK2}. However, all the details are different, 
and the analysis here is considerably more involved. The main reasons for this are as follows: 
(a) the transition studied here is between 2 different power-law behaviors of the determinant
whereas in \cite{CIK2} it was between an exponential and a power-law behavior; (b) 
the ``interaction'' of 2 jump-type singularities on the unit circle leads to larger error terms to control,
especially in the case of $|||\bt|||\ge 1/2$. One of the consequences of (a) is that we have to construct different local parametrices near $z=1$ for the orthogonal polynomial RH problem: for $t$ less than of order $1/n$, and for $t$ larger
than that, see Section \ref{sec75}.

\section{Riemann-Hilbert (RH) problem for a system of orthogonal polynomials and a
differential identity.}\label{secY}

\subsection{RH problem for orthogonal polynomials}\label{secY1}
Assume that $D_n(f_t)\neq 0$, $D_{n+1}(f_t)\neq 0$.  Let $\sqrt{z}>0$ for $z>0$.
Define the polynomials $\phi_n$, $\widehat\phi_n$ by the formulae:
\begin{equation}\label{def phi}
\phi_n(z)=\frac{1}{\sqrt{D_n(f_t)D_{n+1}(f_t)}}\left|\begin{matrix}f_{t,0}&f_{t,-1}&\ldots&f_{t,-n}\\f_{t,1}&f_{t,0}&\ldots&f_{t,-n+1}\\\vdots&\vdots&&\vdots\\f_{t,n-1}&f_{t,n-2}&\ldots&f_{t,-1}\\
1&z&\ldots &z^n\end{matrix}\right|=\chi_nz^n+\ldots,
\end{equation}
where the leading coefficient $\chi_n$ is given by \begin{equation}\chi_n=\sqrt{\frac{D_n(f_t)}{D_{n+1}(f_t)}},\end{equation}
and
\begin{equation}\label{def hatphi}
\widehat\phi_n(z)=\frac{1}{\sqrt{D_n(f_t)D_{n+1}(f_t)}}\left|\begin{matrix}f_{t,0}&f_{t,-1}&\ldots&f_{t,-n+1}&1\\f_{t,1}&f_{t,0}&\ldots&f_{t,-n+2}&z\\\vdots&\vdots&&\vdots&\vdots\\
f_{t,n}&f_{t,n-1}&\ldots &f_{t,1}&z^n\end{matrix}\right|=\chi_nz^n+\ldots,
\end{equation}
i.e., up to a constant, $\phi_n(z)$ is the determinant of a Toeplitz matrix with the last row replaced by the monomials $1,\ldots, z^n$, and $\widehat\phi_n(z)$ is the determinant of a Toeplitz matrix with the last column replaced by the monomials $1,\ldots, z^n$. 

If $f_t(e^{i\theta})$ is positive (or $V(e^{i\th})$ is real-valued and $\al_k,i\bt_k\in\mathbb R$, $k=1,2$) 
then as follows, e.g., from the integral representation for a Toeplitz determinant, $D_n(f_t)\neq 0$ for 
all $n\in\mathbb N$, so that $\phi_n(z)$, $\widehat\phi_n(z)$ are defined for all $n$.

The above polynomials satisfy the orthogonality relations
\be
    \frac{1}{2\pi}\int_{C}\phi_n(z)z^{-k}f_t(z)\frac{dz}{iz}=\chi_n^{-1}\delta_{nk},\quad
     \frac{1}{2\pi}\int_{C}\widehat\phi_n(z^{-1})z^k f_t(z)\frac{dz}{iz}=\chi_n^{-1}\delta_{nk},
\ee     
for $k=0,1,\dots, n$,
where $C$ denotes the unit circle oriented counterclockwise.

If $D_n(f_t)$, $D_{n-1}(f_t)$, and $D_{n+1}(f_t)$ are different from zero, then (as first observed by Fokas, Its, and Kitaev \cite{FIK} for orthogonal polynomials on the real line
(see, e.g., \cite{Dbook})), the matrix-valued function $Y(z;n,t)$ given by
\begin{equation}\label{def Y}
    Y(z)=
    \begin{pmatrix}
    \chi_n^{-1}\phi_n(z)&\chi_n^{-1}\int_{C}\frac{\phi_n(\xi)}{\xi-z}\frac{f_t(\xi)d\xi}{2\pi
i\xi^n}\\
    -\chi_{n-1}z^{n-1}\widehat\phi_{n-1}(z^{-1})&-\chi_{n-1}\int_{C}\frac{\widehat\phi_{n-1}(\xi^{-1})}{\xi-z}
    \frac{f_t(\xi)d\xi}{2\pi i\xi}
    \end{pmatrix}
    \end{equation}
is the unique solution of the following Riemann-Hilbert problem:
\subsubsection*{RH problem for $Y$}
    \begin{itemize}
    \item[(a)] $Y:\mathbb C \setminus C \to \mathbb C^{2\times 2}$ is analytic.
    \item[(b)] Let  $z_1=e^{it}$, $z_2=e^{i(2\pi-t)}$.
The continuous boundary values of $Y$ from the inside, $Y_+$, and from the outside, $Y_-$, of the unit circle exist 
on $C\setminus \{z_1,z_2\}$, and are related by the jump condition
\[
Y_+(z)=Y_-(z)
                \begin{pmatrix}
                    1 & z^{-n}f_t(z) \\
                    0 & 1
                \end{pmatrix},
                \qquad  \mbox{for}\; z \in C\setminus\{z_1,z_2\}.
\]
    \item[(c)] $Y(z)=\left(I+\bigO(1/z)\right)z^{n\si_3}$,\qquad $\si_3=
                \begin{pmatrix}
                    1 & 0 \\
                    0 & -1
                \end{pmatrix}$,
                \qquad  as $z\to \infty $.
    \item[(d)] As $z\to z_k$, $z\in\mathbb C\setminus C$, $k=1,2$, we have
   \[Y(z)=\begin{pmatrix}\bigO(1)&\bigO(1)+\bigO(|z-z_k|^{2\alpha_k})\\\bigO(1)&\bigO(1)+\bigO(|z-z_k|^{2\alpha_k})\end{pmatrix},\qquad \mbox{if}\quad \al_k\neq 0,
\]
and
\[Y(z)=\begin{pmatrix}\bigO(1)&\bigO(|\ln|z-z_k||)\\\bigO(1)&\bigO(|\ln|z-z_k||)\end{pmatrix},\qquad \mbox{if}\quad \al_k=0.
\]
    \end{itemize}

The uniqueness of the solution and the identity $\det Y(z)\equiv 1$ are standard facts which easily
follow from the RH problem and Liouville's theorem.

In the next section \ref{diffid}, we will show that $\frac{d}{dt}\ln D_n(f_t)$ can be expressed exactly in terms of the RH solution $Y$ 
for all $n$ (see (\ref{differentialidentity}) below). In Section \ref{section: RHP Y}, we will solve
this RH problem asymptotically for large $n$. In Section \ref{section: as Toep}, we then substitute
these asymptotics into  (\ref{differentialidentity}), and integrate over $t$, which produces (\ref{asT}).

\subsection{Differential identity}\label{diffid}
In this section, we will express $\frac{d}{dt}\ln D_n(f_t)$ in terms of the entries of the solution $Y$ of the above RH
problem. 

In order to be able to derive a differential identity in the case where the symbol $f$ is unbounded, i.e.\ if $\Re\alpha_k < 0$, we need to use the notion of a {\em regularized} integral over the unit circle.
Let $F$ be an analytic function in a neighborhood of the unit circle, and let $f$ be the symbol defined by (\ref{symbol}), with $-\frac{1}{2}<\Re\alpha_k<0$.
Then 
\begin{equation}\label{reg int1}
\int_C\frac{F(z)f(z)}{z-\zeta}dz=c(f,F)(\zeta-z_k)^{2\alpha_k}+\bigO(1),\qquad \zeta\to z_k.
\end{equation}
We define the regularized integral for $\zeta$ near $z_k$ by the expression:
\begin{equation}\label{reg int2}
\int_C^{(r)}\frac{F(z)f(z)}{z-\zeta}dz\equiv \int_C\frac{F(z)f(z)}{z-\zeta}dz-c(f,F)(\zeta-z_k)^{2\alpha_k}.
\end{equation}
This object is bounded (although not analytic) in a complex neighborhood of $z_k$.
From the analysis of similar integrals in \cite{K-hankel}, \cite{DIK22}, it follows that
\begin{multline}\label{regul1}
 \int_C^{(r)} \frac{F(z)f(z)}{z-z_k}dz\\
=\lim_{\ep\to 0}\left[\int_{C\setminus C_\ep} \frac{F(z)f(z)}{z-z_k}dz -
\frac{F(z_k)}{2\al_k}\{f(z_ke^{-i\ep})-f(z_ke^{i\ep})\}\right],
\end{multline}
where
\[
C_{\ep}=\cup_{k=1,2}\{z\in C: |\arg z-\arg z_k|<\ep\}.
\]

We set
$\wt Y(z)=Y(z)$ in a neighborhood of $z_k$ if $\Re\al_k>0$. If $\Re\al_k<0$, the second column 
of $Y$ has an expansion at $z_k$ containing a growing term of order $(z-z_k)^{2\al_k}$; we set $\wt Y_{j1}=Y_{j1}$ for $j=1,2$, and
$\wt Y_{j2}=Y_{j2}-c_j(z-z_k)^{2\alpha_k}$ with $c_j$ such that $\wt Y$ is bounded in a neighborhood of $z_k$. This is the same as replacing the integrals in the definition of $Y$, see (\ref{def Y}), by their regularized versions. With this definition of $\wt Y$, we have the following.

\begin{proposition}\label{prop21}
Let $t>0$ and $n\in\mathbb N$. Suppose that the RH problem for $Y(z;n,t)$ is solvable.
Then $D_n(f_t)\neq 0$, and the following differential identity holds for $\alpha_k\neq 0$, $k=1,2$:
\be\label{differentialidentity}
\begin{aligned}
{\frac{1}{i}}\frac{d}{dt}\ln D_n(f_t)=
\sum_{k=1}^2 (-1)^k\left[
n(\al_k+\bt_k)
-2\al_k z_k \left({dY^{-1}\over dz}\wt Y\right)_{22}(z_k)\right],\\ z_1=\eit,\quad z_2=e^{i(2\pi-t)},
\end{aligned}
\ee
where $\left({dY^{-1}\over dz}\wt Y\right)_{22}(z_k)=\lim_{z\to z_k}\left({dY^{-1}\over dz}\wt Y\right)_{22}(z)$
with $z\to z_k$ non-tangentially to the unit circle. 
\end{proposition}

\begin{proof}
Solvability of the RH problem at $t$ (i.e., the fact that $D_n\neq 0$, $D_{n\pm1}\neq 0$) implies the solvability in a neighborhood of $t$, and hence the existence (and by (\ref{def phi}), 
(\ref{def hatphi}), differentiability in $t$) of the corresponding orthogonal polynomials. 
We start with the identity 
(3.5) of \cite{DIK22}, which, as is easy to see from the arguments in  \cite{DIK22}, holds for any parameter
$t$ of the polynomials with respect to which they are differentiable:\footnote{In \cite{DIK22}, (\ref{mid}) was derived under a stronger assumption that $D_k\neq 0$, $k=1,2,\dots,n+1$. However, a simple continuity argument shows that (\ref{mid}) holds true if only
$D_n,D_{n\pm 1}\neq 0$.}
\be\label{mid}
{\partial \over\partial t}\ln D_n(f(z))=
2n {{\partial\chi_n\over\partial t}\over \chi_n}+
\frac{1}{2\pi}\int_0^{2\pi} 
{\partial \over\partial t}
\left(\phi_n(z){d\wh\phi_n(z^{-1})\over dz}
-\wh\phi_n(z^{-1}){d\phi_n(z)\over dz}\right)z f(z) d\th,
\ee
where $z=e^{i\th}$, $f\equiv f_t$.
We would like to move the differentiation in the integral over to $f$ noting that
\[
I\equiv \frac{1}{2\pi}\int_0^{2\pi} 
\left(\phi_n(z){d\wh\phi_n(z^{-1})\over dz}
-\wh\phi_n(z^{-1}){d\phi_n(z)\over dz}\right)z f(z) d\th=-2n
\]
by orthogonality, and therefore ${\partial I\over\partial t}=0$. However, because of the case 
$-1/2<\Re\al_k\le0$ for which ${\partial f\over\partial t}$ is not integrable, care needs to be taken.

 As before, let
\[
C_{\ep}=\cup_{k=1,2}\{z\in C: |\arg z-\arg z_k|<\ep\},
\]
and assume that $F(z)$ and ${\partial F(z)\over\partial t}$ are analytic
functions in a neighborhood of the unit circle $C$.
Then
\be\label{di1}
\int_C {\partial F(z)\over\partial t}f(z)dz=\int_{C\setminus C_\ep} {\partial F(z)\over\partial t}f(z)dz+
\bigO(\ep^{2\al_1+1})+\bigO(\ep^{2\al_2+1}),\qquad \e\to 0.
\ee
Note that ($z_1=e^{it}$, $z_2=e^{i(2\pi-t)}$)
\be\label{di2}
\begin{aligned}
{\partial\over\partial t}\int_{C\setminus C_\ep} F(z)f(z)dz=
\int_{C\setminus C_\ep} {\partial F(z)\over\partial t}f(z)dz+
\int_{C\setminus C_\ep} F(z) {\partial f(z)\over\partial t}dz\\
+i\sum_{k=1}^{2}(-1)^{k+1}z_kF(z_k)\{f(z_ke^{-i\ep})-f(z_ke^{i\ep})\}+
\bigO(\ep^{2\al_1+1})+\bigO(\ep^{2\al_2+1})
\end{aligned}
\ee
as $\e\to 0$.
On the other hand,
\be\label{di3}
{\partial\over\partial t}\int_{C\setminus C_\ep} F(z)f(z)dz=
{\partial\over\partial t}\int_{C} F(z)f(z)dz+\bigO(\ep^{2\al_1+1})+\bigO(\ep^{2\al_2+1})
\ee
as $\e\to 0$
by estimation of the integral of $Ff$ over $C_\ep$.

Combining (\ref{di1}), (\ref{di2}), and (\ref{di3}), we can write
\begin{align}
&\int_C {\partial F(z)\over\partial t}f(z)dz={\partial\over\partial t}\int_{C} F(z)f(z)dz-G,\label{di4}\\
&G=\lim_{\ep\to 0}\left[
\int_{C\setminus C_\ep} F(z){\partial f(z)\over\partial t}dz+
i\sum_{k=1}^{2}(-1)^{k+1}z_kF(z_k)\{f(z_ke^{-i\ep})-f(z_ke^{i\ep})\}\right].
\end{align}

Let us now compute ${\partial f(z)\over\partial t}$. Since
$|z-z_k|^{2\al_k}=|2\sin\frac{\th+(-1)^k t}{2}|^{2\al_k}$, we have
\[
{\partial\over\partial t}\ln|z-z_k|^{2\al_k}=(-1)^k\al_k\cot\frac{\th+(-1)^k t}{2}=
i(-1)^k\al_k\frac{z+z_k}{z-z_k}.
\]
Therefore,
\be
{\partial f(z)\over\partial t}=\sum_{k=1}^2(-1)^k\left(\al_k\frac{z+z_k}{z-z_k}+\bt_k\right)if(z)=
\sum_{k=1}^2(-1)^k\left(\al_k+\bt_k+\frac{2\al_kz_k}{z-z_k}\right)if.
\ee
So we can write
\be\label{Gdi}
\begin{aligned}
G&= i\sum_{k=1}^2(-1)^k\left( (\al_k+\bt_k) \int_{C} F(z)f(z)dz\right.\\
&\left. +
z_k\lim_{\ep\to 0}\left[ 2\al_k \int_{C\setminus C_\ep} \frac{F(z)f(z)}{z-z_k}dz -
F(z_k)\{f(z_ke^{-i\ep})-f(z_ke^{i\ep})\}\right]\right).
\end{aligned}
\ee
The limit in the last line is exactly $2\alpha_k$ times the regularized integral (\ref{reg int1}) evaluated at $z_k$, by (\ref{regul1}).

By (\ref{di4}), (\ref{Gdi}), and (\ref{regul1}) with $F(z)=(\phi_n(z){d\wh\phi_n(z^{-1})\over dz}
-\wh\phi_n(z^{-1}){d\phi_n(z)\over dz})$, we obtain from (\ref{mid}),
\begin{align}\label{mid2}
{\partial \over\partial t}\ln D_n(f(z))&=
2n {{\partial\chi_n\over\partial t}\over \chi_n}+
i\sum_{k=1}^2(-1)^k(2n (\al_k+\bt_k)-
2\al_k(I_{1,k}-I_{2,k})),\\
I_{1,k}&=\frac{1}{2\pi i}\int_C^{(r)}\frac{\phi_n(z){d\over dz}\wh\phi_n(z^{-1})}{z-z_k}z_k f(z) dz,\notag\\
I_{2,k}&=\frac{1}{2\pi i}\int_C^{(r)}\frac{\wh\phi_n(z^{-1}){d\over dz}\phi_n(z)}{z-z_k}z_k f(z) dz.\notag
\end{align}

Let us simplify $I_{1,k}$. Adding and subtracting ${d\over dz}\wh\phi_n(z^{-1})|_{z=z_k}$ from the numerator
of the integrand, and observing that
\[
\frac{{d\over dz}\wh\phi_n(z^{-1})-{d\over dz}\wh\phi_n(z^{-1})|_{z=z_k}}{z^{-1}-z_k^{-1}}
\]
is a polynomial in $z^{-1}$ of degree $n$ with the leading coefficient $-n\chi_n$,
we obtain by orthogonality that
\begin{multline}\label{In1}
I_{1,k}=n+{d\over dz}\wh\phi_n(z^{-1})|_{z=z_k}\frac{1}{2\pi}\int_C^{(r)}\frac{\phi_n(z)}{z-z_k}(z-z_k+z_k)
z_k f(z)d\th\\
=n+z_k^2{d\over dz}\wh\phi_n(z^{-1})|_{z=z_k}\frac{1}{2\pi}\int_C^{(r)}\frac{\phi_n(z)}{z-z_k}
f(z)(z^{n-1}-z^{n-1}_k+z^{n-1}_k)z^{-(n-1)}d\th\\
=n+ z_k^{n+1}\chi_n {d\over dz}\wh\phi_n(z^{-1})|_{z=z_k}\wt Y_{12}(z_k).
\end{multline}

Before a similar simplification of $I_{2,k}$, it is convenient first to use the following
recurrence relation (see, e.g., (2.4) in \cite{DIK2}):
\be\label{rr}
\chi_n \wh\phi_n(z^{-1})=\chi_{n-1}z^{-1}\wh\phi_{n-1}(z^{-1})+\wh\phi_{n}(0)z^{-n}\phi_n(z).
\ee
Substituting this into $I_{2,k}$, and then arguing in a similar way as for $I_{1,k}$, we obtain:
\be\label{In2}
I_{2,k}=z_k{d\over dz}\phi_n(z)|_{z=z_k}\wh\phi_n(0)\wt Y_{12}(z_k)-z_k{d Y_{11}\over dz}(z_k)\wt Y_{22}(z_k).
\ee
Applying (\ref{rr}) once again to the corresponding term in (\ref{In1}) and subtracting (\ref{In2}),
we obtain:
\begin{multline}\label{In12}
I_{1,k}-I_{2,k}=
n+z_k{d Y_{11}\over dz}(z_k)\wt Y_{22}(z_k)\\+\left(n Y_{21}(z_k) -n\chi_n\wh\phi_n(0)Y_{11}(z_k)
-z_k{d Y_{21}\over dz}(z_k)\right)\wt Y_{12}(z_k).
\end{multline}

Furthermore,
\begin{multline}\label{chichi}
 2{{\partial\chi_n\over\partial t}\over \chi_n}=
 \frac{1}{2\pi i}\int_C f(z){\partial\over\partial t}(\phi_n(z)\wh\phi_n(z^{-1}))\frac{dz}{z}=
 -\frac{1}{2\pi i}\int_C^{(r)} \phi_n(z)\wh\phi_n(z^{-1}){\partial f(z)\over\partial t}\frac{dz}{z}\\
 =
-i\sum_{k=1}^2(-1)^k [\al_k+\bt_k+2\al_k I_{3,k}],
\end{multline}
where
\be
I_{3,k}\equiv  \frac{1}{2\pi i}\int_C^{(r)} \frac{\phi_n(z)\wh\phi_n(z^{-1})}{z-z_k}z_k f(z)\frac{dz}{z}
\ee
can be analyzed as the other such integrals above, in this case first adding and subtracting
$\wh\phi_n(z_k^{-1})$ from the numerator of the integrand. We then obtain
\[
I_{3,k}=-1+\chi_n\wh\phi_n(z_k^{-1})z_k^n Y_{12}(z_k).
\]
Using again (\ref{rr}) gives
\be\label{In3}
I_{3,k}=-1-Y_{21}(z_k)\wt Y_{12}(z_k)+\chi_n\wh\phi_n(0)Y_{11}(z_k)\wt Y_{12}(z_k).
\ee
Substituting (\ref{In3}) into (\ref{chichi}), and the latter with (\ref{In12}) into (\ref{mid2}), we finally obtain
the differential identity (\ref{differentialidentity}) as $\det Y(z)=1$.

\end{proof}

\begin{remark}\label{remark-di} 
A differential identity for ${d \over d t}\ln D_n(f(z))$ in the case when one of (or both) $\al_k$'s 
is zero is now also easy to obtain. Either one can derive it directly using (\ref{Gdi}), or one can observe that both the left and the right-hand side of (\ref{differentialidentity}) are continuous in $\alpha_k$, so that the differential identity for $\alpha_k=0$ is obtained from (\ref{differentialidentity}) by letting $\alpha_k\to 0$.
%
\end{remark}

\section{Model RH problem}\label{secPsi}

In this section, we state a RH problem which will be used afterwards
to construct the local parametrix $P$ near $1$ in the asymptotic
analysis of the RH problem for the orthogonal polynomials. We will prove the solvability of this model RH problem for certain values of the parameters, obtain asymptotics for it for large and small values of a parameter $s$ in the problem, and relate the problem to the $\sigma$-form of the fifth Painlev\'e equation.
We assume here that $\Re\alpha_1,\Re\alpha_2>-1/2$.

\subsection{Formulation of the problem}

\begin{figure}[t]
\begin{center}
    \setlength{\unitlength}{0.8truemm}
    \begin{picture}(110,95)(-10,-7.5)
    \put(50,45){\thicklines\circle*{.8}}
    \put(50,60){\thicklines\circle*{.8}}
    \put(50,30){\thicklines\circle*{.8}}
    \put(51,56){\small $+i$}
    \put(51,31){\small $-i$}
    \put(50,60){\thicklines\circle*{.8}}
    \put(50,30){\thicklines\circle*{.8}}
    \put(69,60){\thicklines\vector(1,0){.0001}}
    \put(69,30){\thicklines\vector(1,0){.0001}}
    \put(50,60){\line(1,1){25}}
    \put(50,30){\line(1,-1){25}}
    \put(50,60){\line(-1,1){25}}
    \put(50,30){\line(-1,-1){25}}
    \put(50,30){\line(1,0){35}}
    \put(50,60){\line(1,0){35}}
    \put(50,30){\line(0,1){30}}
    \put(65,75){\thicklines\vector(1,1){.0001}}
    \put(65,15){\thicklines\vector(-1,1){.0001}}
    \put(50,39){\thicklines\vector(0,1){.0001}}
    \put(50,54){\thicklines\vector(0,1){.0001}}
    \put(35,75){\thicklines\vector(-1,1){.0001}}
    \put(35,15){\thicklines\vector(1,1){.0001}}

    \put(74,88){\small $\begin{pmatrix}1&e^{2\pi i(\alpha_1-\beta_1)}\\0&1\end{pmatrix}$}
    \put(-14,88){\small $\begin{pmatrix}1&0\\-e^{-2\pi i(\alpha_1-\beta_1)}&1\end{pmatrix}$}
\put(86,29){\small $e^{2\pi i\beta_2\sigma_3}$} \put(86,59){\small
$e^{2\pi i\beta_1\sigma_3}$}
\put(31,45){\small $\begin{pmatrix}0&1\\-1&1\end{pmatrix}$}
    \put(-11,0){\small $\begin{pmatrix}1&0\\-e^{2\pi i(\alpha_2-\beta_2)}&1\end{pmatrix}$}
    \put(74,0){\small $\begin{pmatrix}1&e^{-2\pi i(\alpha_2-\beta_2)}\\0&1\end{pmatrix}$}
\put(100,43){VI}
    \put(8,43){III}
    \put(88,71){I}
    \put(47,7){IV}
    \put(47,75){II}
    \put(88,16){V}
    \end{picture}
    \caption{The jump contour and jump matrices for $\Psi$.}
    \label{figure: Gamma}
\end{center}
\end{figure}
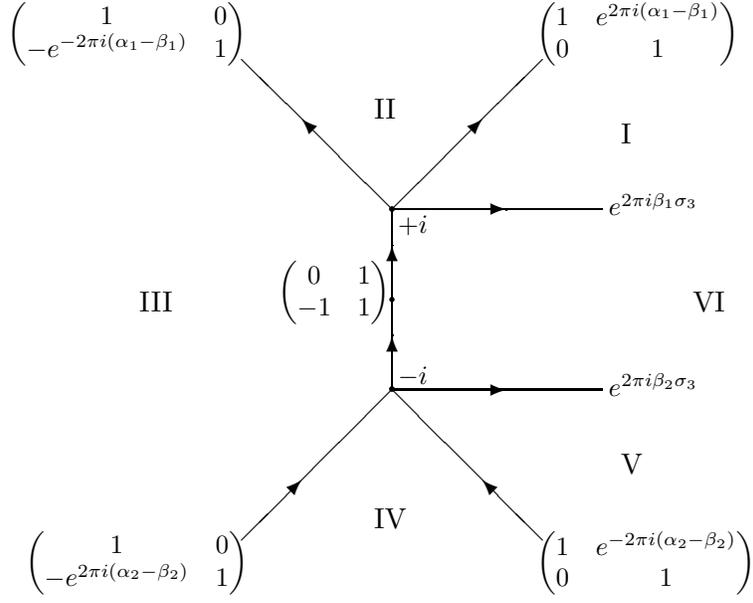

\subsubsection*{RH problem for $\Psi$}
\begin{itemize}
    \item[(a)] $\Psi:\mathbb C\setminus \Gamma \to \mathbb C^{2\times 2}$ is analytic, where
    \begin{align*}&\Gamma=\cup_{k=1}^7\Gamma_k,&& \Gamma_1=i+e^{\frac{i\pi}{4}}\mathbb R^+,
    &&&\Gamma_2=i+e^{\frac{3i\pi}{4}}\mathbb R^+,\\&\Gamma_3=-i+e^{\frac{5i\pi}{4}}\mathbb R^+,
    &&\Gamma_4=-i+e^{\frac{7i\pi}{4}}\mathbb R^+,&&&\Gamma_5=-i+\mathbb R^+,\\&\Gamma_6=i+\mathbb R^+,
    &&\Gamma_7=[-i,i],\end{align*}
    with the orientation chosen as in Figure \ref{figure: Gamma} (the ``-'' side of a contour line is the right-hand side of it).
    \item[(b)] $\Psi$ satisfies the jump conditions
    \begin{equation}\label{jump Psi}\Psi_+(\zeta)=\Psi_-(\zeta)J_k,\qquad
    \zeta\in\Gamma_k,\end{equation}  where
                \begin{align}
                &J_1=\begin{pmatrix}1&e^{2\pi i(\alpha_1-\beta_1)}\\0&1\end{pmatrix},
                &&J_2=\begin{pmatrix}1&0\\-e^{-2\pi i(\alpha_1-\beta_1)}&1\end{pmatrix},\\
                &J_3=\begin{pmatrix}1&0\\-e^{2\pi i(\alpha_2-\beta_2)}&1\end{pmatrix},
                &&J_4=\begin{pmatrix}1&e^{-2\pi i(\alpha_2-\beta_2)}\\0&1\end{pmatrix},\\
                &J_5=e^{2\pi i\beta_2\sigma_3},
                &&J_6=e^{2\pi i\beta_1\sigma_3},\\
                &J_7=\begin{pmatrix}0&1\\-1&1\end{pmatrix}.
                \end{align}
    \item[(c)] We have in all regions:
    \be\label{Psi as}
    \Psi(\zeta)=\left(I+\frac{\Psi_1}{\zeta}+\frac{\Psi_2}{\zeta^2}+\bigO(\zeta^{-3})\right)
    P^{(\infty)}(\zeta)e^{-\frac{is}{4}\zeta\sigma_3} \qquad  \mbox{ as $\zeta\to \infty$,}
    \ee
    where
    \be\label{Pinfty}
    P^{(\infty)}(\zeta)=
    \left(\frac{is}{2}\right)^{-(\beta_1+\beta_2)\sigma_3}
    (\zeta-i)^{-\beta_1\sigma_3}(\zeta+i)^{-\beta_2\sigma_3}
     \ee
                with the branches corresponding to the
                arguments between $0$ and $2\pi$, and where $s\in -i\mathbb R^+$.
    \end{itemize}
    The RH solution $\Psi=\Psi(\zeta;s)$ depends on the complex variable $\zeta$ but also on the complex parameter $s$. We will be concerned with the case where $s\in -i\mathbb R^+$ or $s$ in a small neighborhood of the negative imaginary axis.
    Without additional conditions on the behavior of $\Psi$
    near the points $\pm i$, the RH problem will not
    have a unique solution.
    If $2\alpha_1\notin\mathbb N\cup\{0\}$, $\Re\al_1>-1/2$,
     define $F_1(\zeta)$ by the equations
    \begin{equation}\label{Psi i}\Psi(\zeta;s)=F_1(\zeta;s)(\zeta-i)^{\alpha_1\sigma_3}G_j,\qquad \zeta\in\mbox{ region $j$},\end{equation}
where $j$ takes the values $j=I, II, III, VI$, and where $(\zeta-i)^{\alpha_1\sigma_3}$ is taken with the branch cut on $i+e^{\frac{3\pi i}{4}}\mathbb R^+$, with the argument of $\zeta-i$ between $-5\pi/4$ and $3\pi/4$.  The matrices $G_j$ are
piecewise constant matrices consistent with the jump relations; they
are given by
\begin{align}
&\label{G}G_{III}=\begin{pmatrix}1&g\\0&1\end{pmatrix},\qquad g=-\frac{1}{2
i\sin(2\pi\alpha_1)}(e^{2\pi i\alpha_1}-e^{-2\pi
i\beta_1}), \\&\label{G2}
G_{VI}=G_{III}J_7^{-1}=\begin{pmatrix}1+g&-1\\1&0\end{pmatrix},\\
&\label{G3} G_{I}=G_{VI}J_6,\qquad\qquad  G_{II}=G_{I}J_1.
\end{align}
It is straightforward to verify that $F_1$ has no jumps in a vicinity of $i$, and it is thus meromorphic in a neighborhood of $i$, with possibly an isolated singularity at $i$.

\medskip

Similarly, for $\zeta$ near $-i$, if $2\alpha_2\notin\mathbb N\cup\{0\}$, $\Re\al_2>-1/2$, we define $F_2$ by the equations
    \begin{equation}\label{Psi -i}\Psi(\zeta;s)=F_2(\zeta;s)(\zeta+i)^{\alpha_2\sigma_3}H_j,\qquad \zeta\in\mbox{ region $j$,}\end{equation}
where $(\zeta+i)^{\alpha_2\sigma_3}$ is defined with the branch cut on $-i+e^{\frac{5\pi i}{4}}\mathbb R^+$, with the argument of $\zeta+i$ between $-3\pi/4$ and $5\pi/4$,
and where $H_j$, $j=III, IV, V, VI$ is a
piecewise constant matrix:
\begin{align}
&\label{H}H_{III}=\begin{pmatrix}1&h\\0&1\end{pmatrix},\qquad h=-\frac{1}{2
i\sin(2\pi\alpha_2)}(e^{2\pi i\beta_2}-e^{-2\pi
i\alpha_2})
 \\&\label{H2}
H_{IV}=H_{III}J_3^{-1},
\qquad\qquad H_{V}=H_{IV}J_4^{-1},\qquad\qquad H_{VI}=H_{V}J_5.
\end{align}
Similarly as at $i$, one shows using the jump conditions for $\Psi$ that $F_2$ is meromorphic near $-i$ with a possible singularity at $-i$.

\medskip

If $2\alpha_1\in\mathbb N\cup\{0\}$, the constant $g$ and the matrices $G_j$ are ill-defined, and we need a different definition of $F_1$:
\begin{equation}\label{Psi i 2}\Psi(\zeta;s)=F_1(\zeta;s)(\zeta-i)^{\alpha_1\si_3}
\begin{pmatrix}1&g_{int}\ln(\zeta-i)\\0& 1
\end{pmatrix}G_j,\qquad \zeta\in\mbox{ region $j$,}\end{equation}
where
\be\label{gint}
g_{int}=\frac{e^{-2\pi i\beta_1}-e^{2\pi i\alpha_1}}{2\pi ie^{2\pi i\alpha_1}},
\ee
and $G_{III}=I$,
and the other $G_j$'s are defined as above by applying the appropriate jump conditions. 
Thus defined, $F_1$ has no jumps in a neighborhood of $i$.
Similarly,
if $2\alpha_2\in\mathbb N\cup\{0\}$, we define $F_2$ by the expression:
\begin{equation}\label{Psi -i 2}\Psi(\zeta;s)=F_2(\zeta;s)(\zeta+i)^{\alpha_2\sigma_3}\begin{pmatrix}1&\frac{e^{-2\pi i\alpha_2}-e^{2\pi i\beta_2}}{2\pi ie^{-2\pi i\alpha_2}}\ln(\zeta+i)\\0&1\end{pmatrix}H_j,\qquad \zeta\in\mbox{ region $j$,}\end{equation}
with $H_{III}=I$,
and the other $H_j$'s expressed via $H_{III}$ as in (\ref{H2}). Then $F_2$ has no jumps near $-i$.

\medskip
We are now ready to set an additional RH condition for $\Psi$ in order to ensure uniqueness of the solution.
We complement the RH conditions (a)-(c) with:
\subsubsection*{RH problem for $\Psi$ - extra condition}
\begin{itemize}
\item[(d)] The functions $F_1$ and $F_2$ given in (\ref{Psi i}), (\ref{Psi i 2}) and (\ref{Psi -i}), (\ref{Psi -i 2}) are analytic functions of $\zeta$ at $i$ and $-i$, respectively.
\end{itemize}

Given complex parameters $s,\alpha_1,\alpha_2,\beta_1,\beta_2$, the uniqueness of the function $\Psi$ which satisfies RH conditions (a)-(d) can be proved using standard arguments in the following way.
If $\Psi$ satisfies the RH conditions (a) and (b), it is straightforward to show that $\det\Psi$ is a meromorphic function in $\zeta$, with possibly isolated singularities at $\pm i$. By condition (d), the singularities of $\det\Psi$ are removable, and $\det\Psi$ is an entire function, which tends to $1$ at infinity by condition (c). Thus $\det\Psi$ is identically equal to $1$ by the Liouville theorem, and $\Psi(\zeta)$ is invertible for every $\zeta$. Now, assuming that there are two solutions $\Psi$ and $\widetilde\Psi$ satisfying (a)-(d), one shows in a similar way that $\widetilde\Psi \Psi^{-1}=I$.

Existence of a RH solution $\Psi$ is a much more subtle issue.
If $\alpha_1,\alpha_2,\alpha_1+\alpha_2>-1/2$ are real and $\beta_1,\beta_2\in i\mathbb R$, we will prove later on that the RH problem is solvable for any $s\in -i\mathbb R^+$.
In the more general case where $\Re\alpha_1, \Re\alpha_2,\Re(\alpha_1+\alpha_2)>-1/2$ and
$|||\beta|||<1$, we will analyze the RH problem asymptotically as $s\to -i\infty$ and as $s\to 0$. Our analysis will imply that the RH problem is solvable for $s\in -i\mathbb R^+$ and $|s|$ sufficiently small or $|s|$ sufficiently large, but it is possible that, given $\alpha_1,\alpha_2,\beta_1,\beta_2$, there is a finite number of values of $s\in -i\mathbb R^+$ for which the RH problem is not solvable.

\subsection{Special case $\alpha_1=\alpha_2=\beta_1=\beta_2=\frac{1}{2}$}\label{section: degenerate}
If $\alpha_1=\alpha_2=\beta_1=\beta_2=\frac{1}{2}$, the RH problem for $\Psi$ can be solved explicitly.
Let
\begin{equation}
L(\zeta)\equiv \left(\frac{is}{2}\right)^{-\sigma_3}(\zeta-i)^{-\frac{1}{2}\sigma_3}(\zeta+i)^{-\frac{1}{2}\sigma_3}e^{-\frac{is}{4}\zeta\sigma_3},
\end{equation}
with the branch cuts of the square roots $(\zeta\mp i)^{-\frac{1}{2}\sigma_3}$ along $\pm i+\mathbb R^+$ as in (\ref{Pinfty}), and let
\begin{equation}\label{def Psi deg}
\Psi(\zeta)=
\begin{pmatrix}1&-\frac{2i}{s^2}\left(\frac{e^{-\frac{s}{2}}}{\zeta +i}   - \frac{e^{\frac{s}{2}}}{\zeta -i}\right)\\0&1\end{pmatrix}L(\zeta)\times
\begin{cases}
I,&\mbox{ in regions II and IV,}\\
\begin{pmatrix}1&0\\-1&1\end{pmatrix}
,&\mbox{ in region III,}\\
\begin{pmatrix}1&-1\\0&1\end{pmatrix},&\mbox{ in regions I, V, and VI.}\\
\end{cases}
\end{equation}
It is straightforward to verify that $\Psi$ satisfies the RH conditions (a)-(c). To verify the extra condition (d), note that the terms with logarithms in (\ref{Psi i 2}) and (\ref{Psi -i 2}) vanish. Then it is easily verified that $F_1$ and $F_2$, given by (\ref{Psi i 2}) and (\ref{Psi -i 2}), are analytic near $\pm i$ by substituting (\ref{def Psi deg}).
We also see that the $1,1$ entry of the matrix $\Psi_1$ in (\ref{Psi as}) vanishes.
Using the formulae (\ref{C1}) and (\ref{sigma4b}) below, we obtain that $\sigma(s)=0$ in this case, as announced in Remark \ref{remark: degenerate}.

If $2\alpha_1, 2\alpha_2, 2\beta_1, 2\beta_2\in\mathbb N$ and $\alpha_1=\beta_1$, $\alpha_2=\beta_2$, the RH solution can also be constructed explicitly, but the function $L$ has to be modified in a straightforward way to satisfy (\ref{Psi as}). Furthermore, the upper-triangular matrix in (\ref{def Psi deg}) has to be modified in order to preserve the conditions (\ref{Psi i 2}) and (\ref{Psi -i 2}); it can have higher order poles at $\pm i$.

\subsection{Lax pair}\label{section Lax}
In this section, we assume that $s$ is such that the RH problem for $\Psi$ is solvable.
Let
\begin{equation}\label{def A B}
A=\left(\frac{d}{d\zeta}\Psi\right) \, \Psi^{-1},\qquad B=\left(\frac{d}{ds}\Psi\right) \, \Psi^{-1}.
\end{equation}
It follows from the RH conditions that $A$ is a rational function
with simple poles at $\pm i$ and bounded at infinity,
\begin{equation}\label{A}
A(\zeta;s)=A_\infty(s)
+\frac{A_1(s)}{\zeta-i}+\frac{A_2(s)}{\zeta+i}.
\end{equation}
and that
$B$ is a polynomial of degree $1$,
\begin{equation}\label{B}
B(\zeta;s)=B_1 \zeta + B_0(s).
\end{equation}
Note that $\Psi_1$ is traceless by (\ref{Psi as}) since $\det\Psi\equiv 1$, and
write the matrix $\Psi_1=\Psi_1(s)$ in (\ref{Psi as}) as
\begin{equation}\label{C1}
\Psi_1(s)=\begin{pmatrix}
q(s)&r(s)\\p(s)&-q(s)\end{pmatrix}.
\end{equation}
Substituting (\ref{Psi as}) into (\ref{def A B}) and (\ref{A}), one derives that
\begin{equation}A_\infty=-\frac{is}{4}\sigma_3,\end{equation} and that
\begin{equation}\label{sum}
A_1+A_2=\begin{pmatrix} -(\beta_1+\beta_2) &\frac{irs}{2}\\-\frac{ips}{2}&\beta_1+\beta_2
\end{pmatrix}.
\end{equation}
Expanding $A$ as $\zeta$ tends to infinity, we obtain that the coefficient of the
$\zeta^{-2}$-term is equal to $i(A_1-A_2)$ by (\ref{A}). Since this must be equal to the $\zeta^{-2}$-term of $\Psi_\zeta \Psi^{-1}$, we obtain by (\ref{Psi as}) the identity
\[A_1-A_2=i\left(\Psi_1+i(\beta_1-\beta_2)\sigma_3+(\beta_1+\beta_2)[\Psi_1,\sigma_3]
+\frac{is}{4}[\Psi_2,\sigma_3]-\frac{is}{4}[\Psi_1,\sigma_3]\Psi_1\right),\]
which gives
\begin{equation}\label{diff}
A_1-A_2=\begin{pmatrix}iq-(\beta_1-\beta_2)-\frac{srp}{2}&ir-2i(\beta_1+\beta_2)r+\frac{sh}{2}+\frac{sqr}{2}\\
ip+2i(\beta_1+\beta_2)p-\frac{sj}{2}+\frac{sqp}{2}&-iq+(\beta_1-\beta_2)+\frac{srp}{2}\end{pmatrix},
\end{equation}
where $h=\Psi_{2,12}, j=\Psi_{2,21}$. By (\ref{sum}) and (\ref{diff}), we
obtain
\begin{align}&\label{A10}
A_1=\frac{1}{2}\begin{pmatrix}-2v-2\alpha_1&\frac{irs}{2}+ir-2i(\beta_1+\beta_2)r+\frac{sh}{2}+\frac{sqr}{2}\\
-\frac{is}{2}p+ip+2i(\beta_1+\beta_2)p-\frac{sj}{2}+\frac{sqp}{2}&2v+2\alpha_1,\end{pmatrix},\\
&\label{A20}A_2=-\frac{1}{2}\begin{pmatrix}-2v-2\alpha_1+2\beta_1+2\beta_2&-\frac{isr}{2}+ir-2i(\beta_1+\beta_2)r+\frac{sh}{2}+\frac{sqr}{2}\\
\frac{isp}{2}+ip+2i(\beta_1+\beta_2)p-\frac{sj}{2}+\frac{sqp}{2}&2v+2\alpha_1-2\beta_1-2\beta_2,\end{pmatrix},
\end{align}
where $v$ is given by
\begin{equation}\label{v}
v=\frac{-i}{2}q +\frac{s}{4}rp-\alpha_1+\beta_1.
\end{equation}
Now we can use (\ref{Psi i}) and (\ref{Psi -i}) to derive that
\begin{equation}
\label{det A12}\det A_1(s)=-\alpha_1^2,\qquad \det
A_2(s)=-\alpha_2^2.
\end{equation}
It follows that we can write $A_1$ and $A_2$ in the form
\begin{align}&\label{A1}
A_1=\begin{pmatrix}-v-\alpha_1&-uyv\\
\frac{v+2\alpha_1}{uy}&v+\alpha_1\end{pmatrix},\\
&\label{A2a}A_2=\begin{pmatrix}v+\alpha_1-\beta_1-\beta_2&y(v+\alpha_1-\alpha_2-\beta_1-\beta_2)\\
-\frac{v+\alpha_1+\alpha_2-\beta_1-\beta_2}{y}&-v-\alpha_1+\beta_1+\beta_2\end{pmatrix},
\end{align}
for some functions $u,y$ depending on $s$.

\medskip

For $B_1$ and $B_0$, we can again use (\ref{Psi as}) to derive
\begin{equation}
B_1=-\frac{i\sigma_3}{4},\qquad
B_0=\begin{pmatrix}-\frac{\beta_1+\beta_2}{s}&\frac{ir}{2}\\-\frac{ip}{2}&\frac{\beta_1+\beta_2}{s}\end{pmatrix}.\label{B0}
\end{equation}
The $1/\zeta$ term in the asymptotic expansion of $\frac{d}{ds}\Psi\, \Psi^{-1}$ at infinity
must vanish, and this implies
\begin{equation}\label{identities}
q_s=\frac{i}{2}rp,\qquad h=2ir_s-rq+4ir\frac{\beta_1+\beta_2}{s},\qquad j=-2ip_s+qp+4ip\frac{\beta_1+\beta_2}{s},
\end{equation}
so that by (\ref{v}),
\begin{equation}\label{v2}
v=-\frac{i}{2}(sq_s+q)-\alpha_1+\beta_1.
\end{equation}

 The compatibility condition of the linear
system $\Psi_\zeta=A\Psi$ with $\Psi_s=B\Psi$ gives
\begin{equation}\label{compAB}
A_s-B_\zeta+[A,B]=0,
\end{equation}
the vanishing of the term of $\bigO(1)$ as $\zeta\to\infty$ in (\ref{compAB}) gives an expression for the off-diagonal elements of $B_0$ in terms of $u$, $v$, and $y$: we have
\begin{equation}
B_0=\frac{1}{s}\begin{pmatrix}-\beta_1-\beta_2&-uvy+y(v+\alpha_1-\alpha_2-\beta_1-\beta_2)\\\frac{v+2\alpha_1}{uy}-\frac{1}{y}(v+\alpha_1+\alpha_2-\beta_1-\beta_2)&\beta_1+\beta_2\end{pmatrix}.\label{B0-2}
\end{equation}
Writing down the residues at $\pm i$ in (\ref{compAB}) and using (\ref{A1}), (\ref{A2a}), and (\ref{B0-2}), we obtain
\begin{align}
&su_s=su-2v(u-1)^2+(u-1)\left[u(-\alpha_1-\alpha_2+\beta_1+\beta_2)+3\alpha_1-\alpha_2-\beta_1-\beta_2\right]\label{systemu}\\
&\label{systemv}
sv_s=-\frac{1}{u}(v+2\alpha_1)(v+\alpha_1-\alpha_2-\beta_1-\beta_2)+uv(v+\alpha_1+\alpha_2-\beta_1-\beta_2)\\
&\label{systemy}
sy_s=y\left(\frac{1}{u}(v+2\alpha_1)-2v-2\alpha_1-\frac{s}{2}+uv\right).
\end{align}

The system (\ref{systemu})-(\ref{systemy}) is an alternative form of the Painlev\'e V equation:
in particular (\ref{systemu})-(\ref{systemv}) implies that $u$ solves the Painlev\'e V equation
\begin{equation}\label{PV}
u_{ss}=\left(\frac{1}{2u}+\frac{1}{u-1}\right)u_{s}^2-\frac{1}{s}u_s+\frac{(u-1)^2}{s^2}\left(a_1
u+\frac{a_2}{u}\right)+\frac{a_3 u}{s} +a_4\frac{u(u+1)}{u-1},
\end{equation}
with the parameters $a_j$ given by
\begin{align}&a_1=\frac{1}{2}(\alpha_1+\alpha_2-\beta_1-\beta_2)^2,&
a_2=-\frac{1}{2}(\alpha_1+\alpha_2+\beta_1+\beta_2)^2,\\
&a_3=1-2\alpha_1+2\alpha_2, &a_4=-\frac{1}{2}.\end{align}
For us, it is more convenient to relate (\ref{systemu})--(\ref{systemv}) to the so-called $\sigma$-form
of the fifth Painlev\'e equation. Let
\begin{equation}\label{sigma}
w(s)=\frac{i}{2}sq(s)+(\alpha_1- \beta_1)s.
\end{equation}
By (\ref{v2}), it follows that
\begin{equation}\label{dsigma}
w_s=-v.
\end{equation}
By (\ref{systemv}),
\begin{equation}\label{sigma1}
sw_{ss}=\frac{1}{u}(v+2\alpha_1)(v+\alpha_1-\alpha_2-\beta_1-\beta_2)-uv(v+\alpha_1+\alpha_2-\beta_1-\beta_2).
\end{equation}
We also have by (\ref{sigma}), (\ref{dsigma}), (\ref{v}), and
(\ref{B0}) that $w-sw_s=\frac{s^2}{4}rp$. Furthermore, $\frac{s^2}{4}rp$ can be expressed by (\ref{B0}) and (\ref{B0-2}) in terms of $u,v,y$. We obtain
\begin{eqnarray}
w-sw_s&=&\frac{s^2}{4}rp\nonumber\\
&=&(-uyv+y(v+\alpha_1-\alpha_2-\beta_1-\beta_2))(\frac{v+2\alpha_1}{uy}-\frac{1}{y}(v+\alpha_1+\alpha_2-\beta_1-\beta_2))\nonumber
\\&=&\left(uv(v+\alpha_1+\alpha_2-\beta_1-\beta_2)+\frac{v+2\alpha_1}{u}(v+\alpha_1-\alpha_2-\beta_1-\beta_2)\right)\nonumber\\&&\qquad\qquad
+\left(-2v^2
+2(\beta_1+\beta_2-2\alpha_1)v+\alpha_2^2-(\alpha_1-\beta_1-\beta_2)^2\right).\label{sigma2}
\end{eqnarray}
By (\ref{sigma1}) and (\ref{sigma2}),
\begin{multline}
s^2w_{ss}^2-\left(w-sw_s+2v^2
-2(\beta_1+\beta_2-2\alpha_1)v-\alpha_2^2+(\alpha_1-\beta_1-\beta_2)^2\right)^2\\
=-4v(v+2\alpha_1)(v+\alpha_1+\alpha_2-\beta_1-\beta_2)(v+\alpha_1-\alpha_2-\beta_1-\beta_2).
\end{multline}
Now we substitute (\ref{dsigma}), which gives
\begin{multline}\label{sitemp}
s^2w_{ss}^2=\left(w-sw_s+2w_s^2
+2(\beta_1+\beta_2-2\alpha_1)w_s-\alpha_2^2+(\alpha_1-\beta_1-\beta_2)^2\right)^2\\
-4w_s(w_s-2\alpha_1)(w_s-\alpha_1-\alpha_2+\beta_1+\beta_2)(w_s-\alpha_1+\alpha_2+\beta_1+\beta_2).
\end{multline}
We set \begin{eqnarray}\label{sigma4}
\sigma(s)&=&w(s)+\frac{\beta_1+\beta_2-2\alpha_1}{2}s  -\alpha_2^2-\alpha_1^2+\frac{1}{2}(\beta_1+\beta_2)^2\\
&=&\frac{i}{2}sq(s)-\frac{\beta_1-\beta_2}{2}s-\alpha_1^2-\alpha_2^2+\frac{1}{2}(\beta_1+\beta_2)^2.\label{sigma4b}
\end{eqnarray}
Then the equation (\ref{sitemp}) becomes
\begin{equation}
s^2\sigma_{ss}^2=\left(\sigma-s\sigma_s+2\sigma_s^2\right)^2
-4(\sigma_s-\theta_1)(\sigma_s-\theta_2)
(\sigma_s-\theta_3)(\sigma_s-\theta_4),\label{sigmaform}
\end{equation}
where
\begin{align}
&\theta_1=-\alpha_1+\frac{\beta_1+\beta_2}{2},&&\theta_2=\alpha_1+\frac{\beta_1+\beta_2}{2},\\
&\theta_3=\alpha_2-\frac{\beta_1+\beta_2}{2},
&&\theta_4=-\alpha_2-\frac{\beta_1+\beta_2}{2}.
\end{align}
Equation (\ref{sigmaform}) is the $\sigma$-form of the Painlev\'e V equation as given in \cite[Formula (2.8)]{ForresterWitte}.
The function $r$ defined in (\ref{C1}) can be expressed in terms of $\sigma$. Substituting (\ref{identities}) and (\ref{A10}) in the first equation of (\ref{det A12}), we obtain an identity for the logarithmic derivative of $r$ in terms of $\sigma$:
\begin{equation}\label{identity r}
\frac{r_s}{r}=-2\frac{\beta_1+\beta_2}{s}+\frac{-4 \alpha_1^2 + 4 \alpha_2^2 - 8 (\beta_1 + \beta_2) \sigma_s +
 s^3 \sigma_{ss}}{s^2 (-2 \alpha_1^2 - 2 \alpha_2^2 + (\beta_1 + \beta_2)^2 -
   2 \sigma + 2 s \sigma_s)}.
\end{equation}


The following two results relate the Painlev\'e transcendent $\sigma$ to the functions $F_1$ and $F_2$, defined in (\ref{Psi i}), (\ref{Psi -i}), (\ref{Psi i 2}), and (\ref{Psi -i 2}), evaluated at $\pm i$.
\begin{proposition}\label{prop F1}We have the identities
\begin{align}
&\label{id a1}\alpha_1\left(F_1(i;s)^{-1}\sigma_3F_1(i;s)\right)_{22}=-\sigma_s(s)+\frac{\beta_1+\beta_2}{2},\\
&\label{id a2}\alpha_2\left(F_2(-i;s)^{-1}\sigma_3F_2(-i;s)\right)_{22}=\sigma_s(s) +\frac{\beta_1+\beta_2}{2}.
\end{align}
\end{proposition}
\begin{proof}
If we use (\ref{Psi i}) and (\ref{Psi -i}) to compute $A=\Psi_\zeta\Psi^{-1}$ (as defined in (\ref{def A B})-(\ref{A})), we obtain that $A_1=\alpha_1F_1(i)\sigma_3F_1^{-1}(i)$, and that $A_2=\alpha_2F_2(-i)\sigma_3F_2^{-1}(-i)$. By (\ref{A1}), (\ref{A2a}), (\ref{dsigma}), and (\ref{sigma4}),
the $2,2$-entries of $A_1$ and $A_2$ give us the identities
\begin{align}
&\alpha_1\left(F_1(i)\sigma_3F_1^{-1}(i)\right)_{22}=-\sigma_s+\frac{\beta_1+\beta_2}{2}\\
&\alpha_2\left(F_2(-i)\sigma_3F_2^{-1}(-i)\right)_{22}=\sigma_s+\frac{\beta_1+\beta_2}{2}.\end{align}
Note further that $\left(F_j^{-1}\sigma_3F_j\right)_{22}=\left(F_j\sigma_3F_j^{-1}\right)_{22}$, which implies (\ref{id a1})-(\ref{id a2}).
\end{proof}

\begin{proposition}\label{prop F2}There exist complex constants $c_1, c_2$, which may depend on $\alpha_1,\alpha_2,\beta_1,\beta_2$, but not on $s$, such that
\begin{align}
&\label{id F1}\alpha_1\left(F_1(i;s)^{-1}F_{1,\zeta}(i;s)\right)_{22}=\frac{i}{4}\sigma(s)-\frac{i}{8}(\beta_1+\beta_2)s+c_1,\\
&\label{id F2}\alpha_2\left(F_2(-i;s)^{-1}F_{2,\zeta}(-i;s)\right)_{22}=-\frac{i}{4}\sigma(s) -\frac{i}{8}(\beta_1+\beta_2)s+c_2.
\end{align}
\end{proposition}

\begin{proof}
By (\ref{def A B}) and (\ref{Psi i})-(\ref{Psi -i}), we have $F_{j,s}=BF_j$, $j=1,2$.
Let us expand $F_j(\zeta;s)$ as $\zeta\to \pm i$ as
\begin{equation}\label{F Taylor}
F_j(\zeta;s)=F_j^{(0)}(s)\left(I+F_j^{(1)}(s)(\zeta\mp i)+\bigO(\zeta\mp i)^2\right).
\end{equation}
Substituting this into $F_{j,s}=BF_j$, we obtain by (\ref{B}),
\begin{align}
&\label{Fx}F_{j,s}^{(0)}=(B_0\pm iB_1)F_j^{(0)}, \\
&\label{Fx1}F_{j,s}^{(1)}=\left(F_j^{(0)}\right)^{-1}B_1F_j^{(0)}=-\frac{i}{4}\left(F_j(\pm i;s)^{-1}\sigma_3F_j(\pm i;s)\right).
\end{align}
Here $j=1$ corresponds to the upper symbol in $\pm$ or $\mp$, and $j=2$ to the lower one.
Also by (\ref{F Taylor}), we have $F_j^{-1}(\pm i)F_{j,\zeta}(\pm i)=F_j^{(1)}$, which shows that
\begin{equation}
\alpha_j \left(F_j^{-1}(\pm i)F_{j,\zeta}(\pm i)\right)_{22,s}=-\frac{i}{4}(\mp\sigma_s+\frac{\beta_1+\beta_2}{2})
\end{equation}
by Proposition \ref{prop F1}. Integrating, we obtain (\ref{id F1})-(\ref{id F2}).
\end{proof}

\begin{remark}
To find the explicit expressions for the constants $c_1, c_2$, one can use either the
large $s$ or the small $s$ asymptotic solution of the $\Psi$-RH problem presented in the following sections.
In this way, we obtain
\be\label{c1c2}
c_1=-c_2=\frac{i}{8}(\bt_1+\bt_2)^2.
\ee
\end{remark}

\subsection{Solvability of the RH problem for $\alpha_{1,2}\in\mathbb R$, $\beta_{1,2}\in i\mathbb R$}\label{secsolvability}

From the general theory of the RH problems related to the Painlev\'e equations, it follows that the RH problem for $\Psi(\zeta;s)$, $s\neq 0$, is solvable except for certain isolated values of $s$. At those values, $\sigma(s)$ can have a pole. It will follow from our asymptotic analysis below that the RH problem for $\Psi$ is solvable for sufficiently large $|s|$, $s$ in a neighborhood of $-i\mathbb R^+$,
and therefore, there are at most a finite number of values for $s$ in a neighborhood of $-i\mathbb R^+$ where the RH problem is not solvable.

If $\alpha_1,\alpha_2>-1/2$ are real and $\beta_1,\beta_2$ are purely imaginary, we can prove that there are no poles on $-i\mathbb R^+$ and that the RH problem for $\Psi$ is solvable for all $s\in -i\mathbb R^+$. For this, we use the technique of a so-called vanishing lemma, which consists of proving that the homogeneous version of the RH problem for $\Psi$ has only the trivial zero solution. By \cite{Zhou} and \cite{FIKN, FokasMuganZhou, FokasZhou}, such a vanishing lemma is equivalent to the existence of a solution to the original, non-homogeneous, form of the RH problem. 

\begin{lemma}{\bf (Vanishing lemma)}\ Let $s\in -i\mathbb R^+$,
$i\beta_1,i\beta_2\in \mathbb R$ and
$\alpha_1,\alpha_2>-1/2$, and suppose that $\Psi_0$ satisfies the
conditions (a), (b), and (d) of the RH problem for $\Psi$ in Section \ref{secPsi}, with condition
(c) replaced by the homogeneous asymptotic condition
    \begin{equation}
    \Psi_0(\zeta)e^{\frac{|s|}{4}\zeta\sigma_3}=\bigO(\zeta^{-1})
    , \qquad \mbox{ as $\zeta\to\infty$}.
    \end{equation} Then $\Psi_0 \equiv 0$.\label{vanishing lemma}
        \end{lemma}
\begin{proof}
        Suppose that $\Psi_0$ satisfies the above homogeneous RH conditions.
        Let
        \begin{align*}
        &N(\zeta)=\Psi_0(\zeta)J_1^{-1}J_6^{-1}e^{\frac{|s|}{4}\zeta\sigma_3},&&\mbox{
        in region II, $\Re \zeta>0$,}\\
                &N(\zeta)=\Psi_0(\zeta)J_2e^{\frac{|s|}{4}\zeta\sigma_3},&&\mbox{
        in region II, $\Re \zeta<0$,}\\
                &N(\zeta)=\Psi_0(\zeta)J_3e^{\frac{|s|}{4}\zeta\sigma_3},&&\mbox{
        in region IV, $\Re \zeta<0$,}\\
                &N(\zeta)=\Psi_0(\zeta)J_4^{-1}J_5e^{\frac{|s|}{4}\zeta\sigma_3},&&\mbox{
        in region IV, $\Re \zeta>0$,}\\
        &N(\zeta)=\Psi_0(\zeta)J_5e^{\frac{|s|}{4}\zeta\sigma_3},&&\mbox{
        in region V,}\\
                &N(\zeta)=\Psi_0(\zeta)J_6^{-1}e^{\frac{|s|}{4}\zeta\sigma_3},&&\mbox{
        in region I,}\\
                &N(\zeta)=\Psi_0(\zeta)e^{\frac{|s|}{4}\zeta\sigma_3},&&\mbox{
                in regions III and VI.}
        \end{align*}
Then $N$ satisfies the following RH problem
\subsubsection*{RH problem for $N$}
        \begin{itemize}
        \item[(a)] $N$ is analytic in $\mathbb C\setminus i\mathbb R$.
        \item[(b)] On $i\mathbb R$, $N$ satisfies the jump conditions
        \begin{align}\label{RHP N: b1}
        &N_+(\zeta)=N_-(\zeta)e^{-\frac{|s|}{4}\zeta\sigma_3}J_7e^{\frac{|s|}{4}\zeta\sigma_3},
        &&\mbox{  $\zeta\in(-i,i)$},\\
        &\label{RHP N: b2}N_+(\zeta)=N_-(\zeta)e^{-\frac{|s|}{4}\zeta\sigma_3}J_5^{-1}J_4J_3e^{\frac{|s|}{4}\zeta\sigma_3},
        &&\mbox{  $\zeta\in(-i\infty,-i)$},\\
        &\label{RHP N: b3}N_+(\zeta)=N_-(\zeta)e^{-\frac{|s|}{4}\zeta\sigma_3}J_6J_1J_2e^{\frac{|s|}{4}\zeta\sigma_3},
        &&\mbox{  $\zeta\in(+i, +i\infty)$}.
        \end{align}
        \item[(c)] For a fixed $s\in -i\mathbb R^+$,
        \begin{equation}N(\zeta)=\bigO(\zeta^{-1}),\qquad\mbox{ as
$\zeta\to\infty$}.\end{equation} \end{itemize} Set
        \begin{equation}
        H(\zeta)=N(\zeta)N^*(-\overline{\zeta}).
        \end{equation}
        From the asymptotics for $N$, it follows that
$H(\zeta)=\bigO(\zeta^{-2})$ as $\zeta\to\infty$.
       $H$ is analytic in the left half plane $\Re\zeta<0$, and it has singularities at $\pm i$.
       Because of (\ref{Psi i}) and (\ref{Psi -i}), it follows that
       the singularities are weak enough for $H$ to be integrable along the imaginary line, as $\alpha_1,\alpha_2>-1/2$.
        Using Cauchy's theorem, we then have
        \begin{equation}\label{HHC}
        \int_{-i\infty}^{+i\infty}H_+(\zeta)d\zeta=0.
        \end{equation}
        Because of the jump conditions for $N$, this implies that
        \begin{multline}
        \int_{-i\infty}^{-i}N_-(\zeta)e^{-\frac{|s|}{4}\zeta\sigma_3}J_5^{-1}J_4J_3e^{\frac{|s|}{4}\zeta\sigma_3}
        N_-^*(\zeta)d\zeta
        +\int_{-i}^{i}N_-(\zeta)e^{-\frac{|s|}{4}\zeta\sigma_3}J_7e^{\frac{|s|}{4}\zeta\sigma_3}N_-^*(\zeta)d\zeta\\+
        \int_{i}^{+i\infty}N_-(\zeta)e^{-\frac{|s|}{4}\zeta\sigma_3}J_6J_1J_2e^{\frac{|s|}{4}\zeta\sigma_3}N_-^*(\zeta)d\zeta=0.
        \end{multline}
        Summing up this expression with its Hermitian conjugate and using the form of the jump matrices $J_k$, we
        obtain for $s\in -i\mathbb R^+$, for real $\alpha_1,\alpha_2$, and purely imaginary $\beta_1,\beta_2$,
\begin{multline}
        \int_{-i\infty}^{-i}N_-(\zeta)\begin{pmatrix}0&0\\0&2e^{2i\pi\beta_2}\end{pmatrix}
        N_-^*(\zeta)d\zeta
        +\int_{-i}^{i}N_-(\zeta)\begin{pmatrix}0&0\\0&2\end{pmatrix}N_-^*(\zeta)d\zeta\\+
        \int_{i}^{+i\infty}N_-(\zeta)\begin{pmatrix}0&0\\0&2e^{2i\pi\beta_1}\end{pmatrix}N_-^*(\zeta)d\zeta=0.
        \end{multline}
As $\beta_{1,2}\in i\mathbb R$, it follows that
 the second column of $N_-$
is identically zero on $i\mathbb R$. From the jump conditions
(\ref{RHP N: b1})--(\ref{RHP N: b3}), it then follows that the first
column of $N_+$ is zero on $i\mathbb R$ as well. From the identity
theorem, it follows that $N_{j2}(\zeta)=0$ for $\Re\zeta>0$, and
$N_{j1}(\zeta)=0$ for $\Re\zeta<0$.
        Let us now define
        \begin{equation}
        g_j(\zeta)=\begin{cases}
        N_{j2}(\zeta), &\mbox{ for  $\Re\zeta<0$},\\
        N_{j1}(\zeta), &\mbox{ for  $\Re\zeta>0$}.
        \end{cases}
        \end{equation}
       Then,
        $g_j$ is analytic in $\mathbb C\setminus i\mathbb R$.
        On $i\mathbb R$, $g$ has the following jump relations:
        \begin{equation}
        g_{j,+}(\zeta)=g_{j,-}(\zeta)\times
        \begin{cases}
        e^{-2\pi i\alpha_2}e^{-\frac{|s|}{2}\zeta}, &\mbox{  $\zeta\in(-i\infty,-i)$,}\\
        e^{2\pi i\alpha_1}e^{-\frac{|s|}{2}\zeta}, &\mbox{  $\zeta\in(+i,+i\infty)$,}\\
        e^{-\frac{|s|}{2}\zeta}, &\mbox{  $\zeta\in(-i,i)$.}
        \end{cases}
        \end{equation}
        Now we write $\widehat g_j$ for the analytic continuation of $g_j$ from the left
half plane to $\mathbb
C\setminus\{\zeta: \Re\zeta\geq 0, -1\leq \Im\zeta\leq 1\}$,
        \begin{equation}
        \widehat g_j(\zeta)=\begin{cases}
        g_j(\zeta), &\mbox{ $\Re\zeta<0$,}\\
        g_j(\zeta)e^{2\pi i\alpha_1}e^{-\frac{|s|}{2}\zeta}, &\mbox{  $\Re\zeta\geq 0$, $\Im\zeta>1$,}\\
        g_j(\zeta)e^{-2\pi i\alpha_2}e^{-\frac{|s|}{2}\zeta}, &\mbox{  $\Re\zeta\geq 0$, $\Im\zeta<-1$.}
        \end{cases}
        \end{equation}
        Set
        \begin{equation}
        h_j(\zeta)=\widehat g_j(-(\zeta+2)^{3/2}),
        \end{equation}
        where we choose $(\zeta+2)^{3/2}$ with the branch cut on $(-\infty,-2]$ and corresponding to arguments between $-\pi$ and $\pi$.
        It is easy to verify that $h_j$ is analytic and bounded for $\Re\zeta\geq
0$, and that
        $h_j(\zeta)=\bigO(e^{-\frac{|s|}{2}|\zeta|})$ for $\zeta\to\pm i\infty$. By Carlson's
theorem (see, e.g., \cite{Titch}), this implies that $h_j\equiv 0$ if $|s|>0$. It follows that
$g_j\equiv 0$ and consequently $N\equiv 0$ and $\Psi_0\equiv 0$, which
proves the lemma.
\end{proof}

\section{Auxiliary RH problem}\label{secaux}

We assume in this section that $\alpha\pm\beta\neq -1, -2, \ldots $, and $\Re\al>-1/2$.
In \cite[Section 4.2.1]{CIK2}, see also \cite{ItsKrasovsky, DIK2, MMS}, a
function $M=M^{(\alpha,\beta)}$ was constructed explicitly in terms
of the confluent hypergeometric function, which solves the following
RH problem.

\subsubsection*{RH problem for $M$}
    \begin{itemize}
    \item[(a)] $M:\mathbb C\setminus \left(e^{\pm\frac{\pi i}{4}}\mathbb R\cup
\mathbb R^+\right) \to \mathbb C^{2\times 2}$ is analytic,
    \item[(b)] $M$ has continuous boundary values on $e^{\pm\frac{\pi i}{4}}\mathbb
R\cup \mathbb R^+\setminus\left\{0\right\}$ related by the
conditions:
    \begin{align}
    &\label{RHP M0: b1}M_{+}(\lambda)=M_{-}(\lambda)\begin{pmatrix}1&e^{\pi
i(\alpha-\beta)}\\0&1\end{pmatrix}, &\mbox{ as $\lambda\in
e^{\frac{i\pi}{4}}\mathbb R^+$,}\\
    &M_{+}(\lambda)=M_{-}(\lambda)\begin{pmatrix}1&0\\-e^{-\pi
i(\alpha-\beta)}&1\end{pmatrix}, &\mbox{ as $\lambda\in
e^{\frac{3i\pi}{4}}\mathbb R^+$,}\\
    &M_{+}(\lambda)=M_{-}(\lambda)\begin{pmatrix}1&0\\e^{\pi
i(\alpha-\beta)}&1\end{pmatrix}, &\mbox{ as $\lambda\in
e^{\frac{5i\pi}{4}}\mathbb R^+$,}\\
    &M_{+}(\lambda)=M_{-}(\lambda)\begin{pmatrix}1&-e^{-\pi
i(\alpha-\beta)}\\0&1\end{pmatrix}, &\mbox{ as $\lambda\in
e^{\frac{7i\pi}{4}}\mathbb R^+$,}\\
    &\label{RHP M0: b5}M_{+}(\lambda)=M_{-}(\lambda)e^{2\pi i\beta\sigma_3}, &\mbox{
as $\lambda\in \mathbb R^+$},
    \end{align}
    where all the rays of the jump contour are oriented away from the origin, see Figure \ref{Figure: M}.
    \item[(c)] Furthermore, in all sectors,
\begin{equation}\label{RHP Ph: c}
M(\lambda)=\left(I+M_1\lambda^{-1}+\bigO(\lambda^{-2})\right)\lambda^{-\beta\sigma_3}e^{-\frac{1}{2}\lambda\sigma_3},\qquad
\mbox{ as $\lambda\to\infty$},
\end{equation}
where $0<\arg\lb<2\pi$, and
\begin{equation}\label{M1}
M_1=M_1^{(\al,\bt)}=
\begin{pmatrix}\alpha^2-\beta^2&-e^{-2\pi i\beta}\frac{\Gamma(1+\alpha-\beta)}{\Gamma(\alpha+\beta)}\\e^{2\pi i\beta}\frac{\Gamma(1+\alpha+\beta)}{\Gamma(\alpha-\beta)}&-\alpha^2+\beta^2\end{pmatrix}.
\end{equation}
    \end{itemize}
\begin{figure}[t]
\begin{center}
    \setlength{\unitlength}{0.8truemm}
    \begin{picture}(100,48.5)(0,2.5)

    \put(48,27){\small $0$}
    \put(50,25){\thicklines\circle*{.8}}

    \put(50,25){\line(1,1){25}}
    \put(50,25){\line(1,-1){25}}
    \put(50,25){\line(-1,1){25}}
    \put(50,25){\line(-1,-1){25}}
    \put(50,25){\line(1,0){45}}
    \put(69,25){\thicklines\vector(1,0){.0001}}
    \put(65,40){\thicklines\vector(1,1){.0001}}
    \put(65,10){\thicklines\vector(1,-1){.0001}}
    \put(35,40){\thicklines\vector(-1,1){.0001}}
    \put(35,10){\thicklines\vector(-1,-1){.0001}}

    \put(67,37){\small $\begin{pmatrix}1&e^{\pi i(\alpha-\beta)}\\0&1\end{pmatrix}$}
    \put(-4,40){\small $\begin{pmatrix}1&0\\-e^{-\pi i(\alpha-\beta)}&1\end{pmatrix}$}
    \put(85,27){\small $e^{2\pi i\beta\sigma_3}$}
    \put(2,8){\small $\begin{pmatrix}1&0\\e^{\pi i(\alpha-\beta)}&1\end{pmatrix}$}
    \put(67,11){\small $\begin{pmatrix}1&-e^{-\pi i(\alpha-\beta)}\\0&1\end{pmatrix}$}

    \put(110,35){1}
    \put(50,50){2}
    \put(0,25){3}
    \put(110,7){5}
    \put(50,-3){4}
    \end{picture}
    \caption{The jump contour and jump matrices for $M$.}
    \label{Figure: M}
\end{center}
\end{figure}
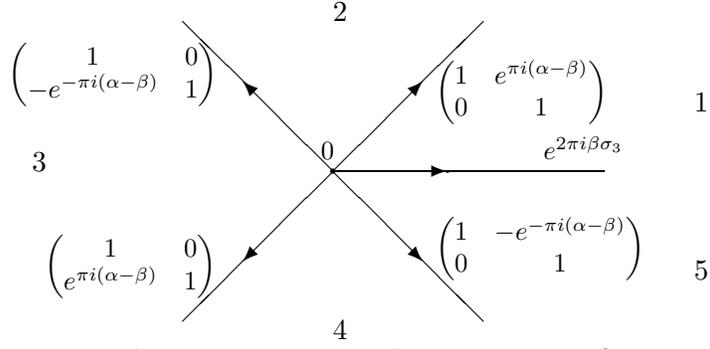
The function $M$ was used in \cite{CIK2} to construct the global parametrix for an analogue
of the $\Psi$-RH problem for small values of a parameter in the problem.
In the present paper, we will make use of $M$ twice: in the construction of the local parametrices at $\pm i$
in the study of the large $|s|$ asymptotics, and in the construction of the global parametrix for the small $|s|$ asymptotics.
In the latter case, we will need, in addition to the RH conditions, more precise information on the local behavior of $M$ at zero in the region between the lines
$e^{\frac{3i\pi}{4}}\mathbb R^{+}$ and $e^{\frac{5i\pi}{4}}\mathbb R^{+}$, which we call region 3.
Write $M\equiv M^{(3)}$
in this region.
It is known (see (\cite[Section 4.2.1]{CIK2}) that $M^{(3)}$ can be written in the form
\be\label{M3}
M^{(3)}(\lb)=L(\lb)\lb^{\al\si_3}\wt G_3,\qquad 2\al\neq 0,1,\dots,\quad \al\pm\bt\neq
-1,-2,\dots,
\ee
with the branch of $\lb^{\pm\al}$ chosen with $0<\arg\lb<2\pi$.
Here
\begin{multline}\label{ME}
L(\lb)=e^{-\lb/2}\left(\begin{matrix}
e^{-i\pi(\al+\bt)}\frac{\Gamma(1+\al-\bt)}{\Gamma(1+2\al)}\varphi(\al+\bt,1+2\al,\lb)
\cr
-e^{-i\pi(\al-\bt)}\frac{\Gamma(1+\al+\bt)}{\Gamma(1+2\al)}\varphi(1+\al+\bt,1+2\al,\lb)
\end{matrix}\right.\\
\left.\begin{matrix}
e^{i\pi(\al-\bt)}\frac{\Gamma(2\al)}{\Gamma(\al+\bt)}\varphi(-\al+\bt,1-2\al,\lb)\cr
e^{i\pi(\al+\bt)}\frac{\Gamma(2\al)}{\Gamma(\al-\bt)}\varphi(1-\al+\bt,1-2\al,\lb)
\end{matrix}\right)
\end{multline}
is an entire function, with
\begin{equation}\label{phidef}
\varphi(a,c;z) =
1+\sum_{n=1}^{\infty}\frac{a(a+1)\cdots(a+n-1)}{c(c+1)\cdots(c+n-1)}\frac{z^n}{n!},
\qquad c\neq 0,-1,-2,\dots ,
\end{equation} and $\wt G_3$ is the constant matrix
\be\label{MG3}
\wt G_3=\begin{pmatrix}
1 & \wt g \cr
0 & 1
\end{pmatrix},\qquad \wt g=\wt g(\al,\bt)=-\frac{\sin\pi(\alpha+\beta)}{\sin 2\pi\alpha}.
\ee
If $2\al$ is an integer, we have
\begin{align}
&M^{(3)}(\lambda)=\wt L(\lambda)\lb^{\al\si_3}
\begin{pmatrix}1& m(\lb)\\0& 1\end{pmatrix},\label{M3half}\\
&m(\lb)=\frac{(-1)^{2\al+1}}{\pi}\sin\pi(\alpha+\beta)\ln(\lb e^{-i\pi}),
&\mbox{ if }\quad 2\alpha=0,1,\dots,\label{mhalf}
\end{align}
where $\wt L(\lb)$ is analytic at zero, and the branch of the logarithm corresponds to the argument of $\lb$ between $0$ and $2\pi$.

\section{Asymptotics for $\Psi$ as $s\to -i\infty$}\label{sec large s}

We will perform a series of transformations of the RH problem for $\Psi(\zeta)$ in order to obtain a small norm RH problem for which we can derive asymptotics as $s\to -i\infty$.
The asymptotic solution will be built out of 3 parametrices: 2 local ones (at $z_1$ and $z_2$)
given in terms of the function $M$ of the previous section, and a global one at infinity in terms
of elementary functions. The asymptotic solution will be valid uniformly in $\zeta$.

We assume in this section that $s\in -i\mathbb R^+$, $|s|$ is large, $\Re\alpha_1,\Re\alpha_2>-1/2$, 
$\al_k\pm\bt_k\neq -1,-2,\dots$, $k=1,2$, 
and that $|||\beta|||<1$. The results of this section can be easily extended to $s$ in a neighborhood of
$-i\mathbb R^+$.

\subsection{Normalization of the problem and opening of lens}

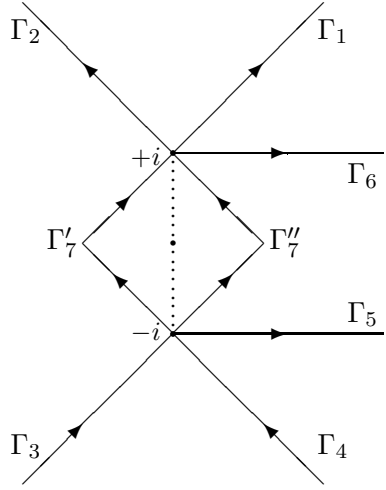
\begin{figure}[t]
\begin{center}
    \setlength{\unitlength}{0.8truemm}
    \begin{picture}(110,95)(-10,-7.5)
    \put(50,45){\thicklines\circle*{.8}}
    \put(50,60){\thicklines\circle*{.8}}
    \put(50,30){\thicklines\circle*{.8}}
    \put(43,58){\small $+i$}
    \put(43,29){\small $-i$}
    \put(50,60){\thicklines\circle*{.8}}
    \put(50,30){\thicklines\circle*{.8}}
    \put(69,60){\thicklines\vector(1,0){.0001}}
    \put(69,30){\thicklines\vector(1,0){.0001}}
    \put(50,60){\line(1,1){25}}
    \put(50,30){\line(1,-1){25}}
    \put(50,60){\line(-1,1){25}}
    \put(50,30){\line(-1,-1){25}}
    \put(50,30){\line(1,0){35}}
    \put(50,60){\line(1,0){35}}

\multiput(50,30)(0,1.5){20}{\circle*{0.1}}

\put(50,30){\line(1,1){15}}\put(50,30){\line(-1,1){15}}
\put(50,60){\line(1,-1){15}}\put(50,60){\line(-1,-1){15}}

%
\put(60,40){\thicklines\vector(1,1){.0001}}
    \put(57,53){\thicklines\vector(-1,1){.0001}}
\put(40,40){\thicklines\vector(-1,1){.0001}}
    \put(43,53){\thicklines\vector(1,1){.0001}}

\put(23,79){$\Gamma_2$} \put(74,79){$\Gamma_1$}
\put(23,10){$\Gamma_3$} \put(74,10){$\Gamma_4$}
\put(79,32){$\Gamma_5$} \put(79,55){$\Gamma_6$}
\put(29,44){$\Gamma_7'$} \put(66,44){$\Gamma_7''$}

    \put(65,75){\thicklines\vector(1,1){.0001}}
    \put(65,15){\thicklines\vector(-1,1){.0001}}

    \put(35,75){\thicklines\vector(-1,1){.0001}}
    \put(35,15){\thicklines\vector(1,1){.0001}}

 \end{picture}
    \caption{The jump contour for $U$.}
    \label{figure: lens}
\end{center}
\end{figure}
Consider the contour shown in Figure \ref{figure: lens}, and set
\begin{equation}\label{def U}
U(\zeta)=\begin{cases}\Psi(\zeta)e^{\frac{|s|}{4}\zeta \sigma_3},&\mbox{ outside the region delimited by $\Gamma_7'$ and $\Gamma_7''$,}\\
\Psi(\zeta)\begin{pmatrix}1&1\\0&1\end{pmatrix}e^{\frac{|s|}{4}\zeta \sigma_3},&\mbox{ in the right part of this region,}\\
\Psi(\zeta)\begin{pmatrix}1&0\\1&1\end{pmatrix}e^{\frac{|s|}{4}\zeta
\sigma_3},&\mbox{ in the left part of this region.} \end{cases}
\end{equation}
Then we have the following RH problem for $U$:

\subsubsection*{RH problem for $U$}
\begin{itemize}
    \item[(a)] $U:\mathbb C\setminus (\cup_{k=1}^6\Gamma_k\cup\Gamma_7'\cup\Gamma_7'') \to \mathbb C^{2\times 2}$ is
    analytic.
    \item[(b)] $U$ satisfies the jump conditions
    \begin{align}\label{jump U1}&U_+(\zeta)=U_-(\zeta)e^{-\frac{|s|}{4}\zeta \sigma_3}J_ke^{\frac{|s|}{4}\zeta \sigma_3},&
    \zeta\in\Gamma_k, k=1,\ldots, 6,\\
    &\label{jump U2}U_+(\zeta)=U_-(\zeta)e^{-\frac{|s|}{4}\zeta \sigma_3}\begin{pmatrix}1&0\\-1&1\end{pmatrix}
    e^{\frac{|s|}{4}\zeta \sigma_3},&
    \zeta\in\Gamma_7',\\
    &\label{jump U3}U_+(\zeta)=U_-(\zeta)e^{-\frac{|s|}{4}\zeta \sigma_3}\begin{pmatrix}1&1\\0&1\end{pmatrix}
    e^{\frac{|s|}{4}\zeta \sigma_3},&
    \zeta\in\Gamma_7''.
    \end{align}
    \item[(c)] We have
    \begin{multline}\label{U as}
    U(\zeta)=\left(I+\frac{\Psi_1}{\zeta}+\frac{\Psi_2}{\zeta^2}+\bigO(\zeta^{-3})\right)
    P^{(\infty)}(\zeta)
    \qquad  \mbox{ as $\zeta\to \infty$,}
    \end{multline}
    where $P^{(\infty)}(\zeta)$ is defined in (\ref{Pinfty}).
    \end{itemize}
The local behavior of $U$ near $\pm i$ can be deduced from the local behavior for $\Psi$ (see the condition (d) of the RH problem for $\Psi$) and (\ref{def U}).

\subsection{Global parametrix} As $s\to -i\infty$, the jump matrices
for $U$ tend to $I$ exponentially fast on $\Gamma_1, \Gamma_2,
\Gamma_3,\Gamma_4, \Gamma_7', \Gamma_7''$, except in the vicinity of
$\pm i$. Ignoring those parts of the contour, we are left with a RH
problem with jumps only on $\Gamma_5\cup\Gamma_6$, which can be
solved explicitly. It is indeed easily verified that $P^{(\infty)}(\zeta)$ given by (\ref{Pinfty})
(the choice of the branches of $(\zeta-i)^{-\beta_1\sigma_3}$ and $(\zeta+i)^{-\beta_2\sigma_3}$
is described following (\ref{Pinfty}))
satisfies the same jump condition as $U$ on $\Gamma_5\cup\Gamma_6$.
In addition, $U(\zeta)P^{(\infty)}(\zeta)^{-1}$ tends to $I$ as $\zeta\to\infty$.

\subsection{Local parametrices}\label{section: local infty}

Let $\mathcal U_1$ and  $\mathcal U_2$ be fixed nonintersecting open discs centered at $i$ and $-i$, respectively.

\subsubsection{Construction of a local parametrix near $i$}
Set
\begin{equation}\label{def tilde M}
\widetilde M(\lambda)=M^{(\alpha_1,\beta_1)}(\lambda)e^{-\frac{\pi
i}{2}(\alpha_1-\beta_1)\sigma_3},
\end{equation}
in terms of the solution of the auxiliary RH-problem of Section \ref{secaux} with parameters
$\alpha=\alpha_1$, $\beta=\beta_1$,
and let $P_1$ be of the form
\begin{equation}\label{def P+}
P_1(\zeta)=E_1(\zeta)\widetilde M\left(\frac{|s|}{2}(\zeta-i)\right)e^{\frac{|s|}{4}\zeta
\sigma_3},\qquad \zeta\in \mathcal U_1.
\end{equation}
If $E_1$ is an analytic function in $\mathcal U_1$,
one verifies directly that $P_1$ satisfies the same jump
conditions as $U$ in $\mathcal U_1$, and the singularity of $U(\zeta)P_1^{-1}(\zeta)$ at $i$ is removable. Moreover, as $s\to -i\infty$, we have by (\ref{RHP Ph: c}), (\ref{def tilde M}), and (\ref{def P+}) that
\begin{equation}
P_1(\zeta)P^{(\infty)}(\zeta)^{-1}=
E_1(\zeta)(I+\bigO(|s|^{-1}))e^{\frac{i|s|}{4}\sigma_3}\left(\frac{|s|}{2}\right)^{\beta_2\sigma_3}e^{-\frac{\pi
i}{2}(\alpha_1-\beta_1)\sigma_3}(\zeta+i)^{\beta_2\sigma_3},\label{matching}
\end{equation}
for $\zeta\in \partial \mathcal U_1$.
Set
\begin{equation}\label{def E1}
E_1(\zeta)=(\zeta+i)^{-\beta_2\sigma_3}\left(\frac{|s|}{2}\right)^{-\beta_2\sigma_3}e^{-\frac{i|s|}{4}\sigma_3}e^{\frac{\pi
i}{2}(\alpha_1-\beta_1)\sigma_3}.
\end{equation}
This function is analytic in $\mathcal U_1$,
and by (\ref{RHP Ph: c})--(\ref{M1}),
\begin{multline}\label{matching 2}
P_1(\zeta)P^{(\infty)}(\zeta)^{-1}=\left(\frac{|s|}{2}\right)^{-\beta_2\sigma_3}\left(I+\frac{2}{|s|}Q_1(\zeta)+\bigO(|s|^{-2})\right)\left(\frac{|s|}{2}\right)^{\beta_2\sigma_3},\\\qquad s\to -i\infty,\qquad
\zeta\in\partial \mathcal U_1,
\end{multline}
with
\begin{multline}\label{Q1}
Q_1(\zeta)=\frac{1}{\zeta -i}(\zeta+i)^{-\beta_2\sigma_3}e^{-\frac{i|s|}{4}\sigma_3}\\
\times \ \begin{pmatrix}\alpha_1^2-\beta_1^2&-e^{i\pi(\alpha_1-3\beta_1)}\frac{\Gamma(1+\alpha_1-\beta_1)}{\Gamma(\alpha_1+\beta_1)}\\
e^{-i\pi(\alpha_1-3\beta_1)}\frac{\Gamma(1+\alpha_1+\beta_1)}{\Gamma(\alpha_1-\beta_1)}&-\alpha_1^2+\beta_1^2\end{pmatrix}
e^{\frac{i|s|}{4}\sigma_3}(\zeta+i)^{\beta_2\sigma_3}.
\end{multline}

\subsubsection{Construction of a local parametrix near $-i$}
Similarly, we set
\begin{equation}
\widehat M(\lambda)=M^{(\alpha_2,\beta_2)}(\lambda)e^{\frac{\pi
i}{2}(\alpha_2-\beta_2)\sigma_3},
\end{equation} and
\begin{equation}
P_2(\zeta)=E_2(\zeta)\widehat M\left(\frac{|s|}{2}(\zeta+i)\right)e^{\frac{|s|}{4}\zeta
\sigma_3},\qquad \zeta\in \mathcal U_2.
\end{equation}
If $E_2$ is analytic in $\mathcal U_2$, then it is straightforward to verify that $P_2$ has the same jump conditions as $U$ has near $-i$, and the singularity of $U(\zeta)P_1^{-1}(\zeta)$ at $i$ is removable.
Set
\begin{equation}\label{def E3}
E_2(\zeta)=(\zeta-i)^{-\beta_1\sigma_3}\left(\frac{|s|}{2}\right)^{-\beta_1\sigma_3}e^{\frac{i|s|}{4}\sigma_3}e^{-\frac{\pi i}{2}(\alpha_2-\beta_2)\sigma_3}.
\end{equation}
We have that
\begin{multline}\label{matching 3}
P_2(\zeta)P^{(\infty)}(\zeta)^{-1}=\left(\frac{|s|}{2}\right)^{-\beta_1\sigma_3}\left(I+\frac{2}{|s|}Q_2(\zeta)+\bigO(|s|^{-2})\right)\left(\frac{|s|}{2}\right)^{\beta_1\sigma_3},\\\qquad s\to -i\infty,\qquad
\zeta\in\partial \mathcal U_2,
\end{multline}
with
\begin{multline}\label{Q2}
Q_2(\zeta)=\frac{1}{\zeta +i}(\zeta-i)^{-\beta_1\sigma_3}e^{\frac{i|s|}{4}\sigma_3}\\
\times \ \begin{pmatrix}\alpha_2^2-\beta_2^2&-e^{-i\pi(\alpha_2+\beta_2)}\frac{\Gamma(1+\alpha_2-\beta_2)}{\Gamma(\alpha_2+\beta_2)}\\
e^{i\pi(\alpha_2+\beta_2)}\frac{\Gamma(1+\alpha_2+\beta_2)}{\Gamma(\alpha_2-\beta_2)}&-\alpha_2^2+\beta_2^2\end{pmatrix}
e^{-\frac{i|s|}{4}\sigma_3}(\zeta-i)^{\beta_1\sigma_3}.
\end{multline}

\subsection{Solution of the  RH problem and the asymptotics of $\si(s)$ for large $s$}

Using the global parametrix $P^{(\infty)}$ and the local parametrices $P_1$ and $P_2$, we define a function $\widetilde R$ as follows:
\begin{equation}\label{def tilde R}
\widetilde R(\zeta)=\begin{cases}U(\zeta)P_1(\zeta)^{-1},&\zeta\in \mathcal U_1,\\
U(\zeta)P_2(\zeta)^{-1},&\zeta\in \mathcal U_2,\\
U(\zeta)P^{(\infty)}(\zeta)^{-1},&\zeta\in \mathbb C
\setminus\left(\overline{\mathcal U_1}\cup \overline{\mathcal U_2}\right).\end{cases}
\end{equation}
Next, define $R$ by the expression:
\begin{equation}\label{def R}
R(\zeta)=\left(\frac{|s|}{2}\right)^{\frac{\beta_1+\beta_2}{2}\sigma_3}\widetilde R(\zeta)\left(\frac{|s|}{2}\right)^{-\frac{\beta_1+\beta_2}{2}\sigma_3}.
\end{equation}

Then $R$ has no jumps inside $\mathcal U_1$ and $\mathcal U_2$, and $R$ is analytic in $\mathbb C\setminus\Gamma_R$, where $\Gamma_R=\partial \mathcal U_1\cup\partial \mathcal U_2\cup [(\Gamma_1\cup\Gamma_2\cup\Gamma_3\cup\Gamma_4\cup\Gamma_7'\cup\Gamma_7'')\setminus \left(\overline{\mathcal U_1}\cup \overline{\mathcal U_2}\right)]$.
 Except on $\partial \mathcal U_1\cup\partial \mathcal U_2$, $R$ has exponentially small jumps in $s$ on the contour as $|s|\to\infty$. We choose the {\it clockwise} orientation for $\partial \mathcal U_1$ and $\partial \mathcal U_2$. We have the following:
\subsubsection*{RH problem for $R$}
\begin{itemize}
    \item[(a)] $R:\mathbb C\setminus \Gamma_R\to \mathbb C^{2\times 2}$ is
    analytic.
    \item[(b)] $R$ satisfies the jump conditions
    \begin{align}&\label{jump R1}R_+(\zeta)=R_-(\zeta)(I+\bigO(e^{-c|s|})),\qquad
    \zeta\in\Gamma\setminus (\partial \mathcal U_1\cup \partial \mathcal U_2),\\
&\label{jump R2}R_+(\zeta)=R_-(\zeta)(I+
\frac{2}{|s|}
\Delta_1(\zeta;s)+
\bigO(|s|^{-2+ |||\beta|||})),\qquad
    \zeta\in\partial \mathcal U_1,\\
    &\label{jump R3}R_+(\zeta)=R_-(\zeta)(I+\frac{2}{|s|}\Delta_2(\zeta;s)+\bigO(|s|^{-2+|||\beta|||})),\qquad
    \zeta\in\partial \mathcal U_2,
    \end{align}
    where
\begin{align}
&\label{Delta1}\Delta_1(\zeta;s)\equiv\Delta_1(\zeta)=\left(\frac{|s|}{2}\right)^{\frac{\beta_1-\beta_2}{2}\sigma_3}Q_1(\zeta)\left(\frac{|s|}{2}\right)^{-\frac{\beta_1-\beta_2}{2}\sigma_3},\\
&\label{Delta2}\Delta_2(\zeta;s)\equiv\Delta_2(\zeta)=\left(\frac{|s|}{2}\right)^{-\frac{\beta_1-\beta_2}{2}\sigma_3}Q_2(\zeta)\left(\frac{|s|}{2}\right)^{\frac{\beta_1-\beta_2}{2}\sigma_3}.
\end{align}

    \item[(c)] As $\zeta\to\infty$, we have
    \begin{equation}\label{R as+}
    R(\zeta)=I+\bigO(\zeta^{-1}).
    \end{equation}
    \end{itemize}

If $|||\beta|||<1$, $\frac{1}{|s|}\Delta_j(\zeta)=\bigO(|s|^{-1+|||\beta|||})=o(1)$, all the jumps are close to the identity matrix, and this is a small-norm RH problem. Following the general theory for small-norm RH problems, we can conclude that the RH problem for $R$ is solvable for $|s|$ sufficiently large, and that
\begin{equation}\label{as R-1}R(\zeta)=I+\bigO(|s|^{-1+|||\beta|||}/(|\zeta|+1)),
\end{equation} uniformly for $\zeta\in\mathbb C\setminus\Gamma_R$ as $s\to -i\infty$.
Formula (\ref{as R-1}) is sufficient to obtain the leading (linear) term in the large-$s$ expansion of the Painlev\'e solution $\sigma$ using (\ref{C1}) and (\ref{sigma4b}).
However, for our analysis of the determinant $D_n(f_t)$, we need to compute the asymptotics of $\si(s)$ up to the terms decreasing with $s$. This requires a more detailed analysis, because the standard expansion of $R$
contains terms with powers of $|s|^{-1+|||\bt|||}$, and to obtain $\si$, we would need to multiply these series by $s$, thus obtaining
the asymptotic expansion of $\si(s)$ with terms of order $|s|^{-k+1+k|||\bt|||}$, $k=1,2\dots$. So the closer $|||\bt|||$ is to 1,
the more terms in this expansion we would have to compute before we encounter the terms decreasing with $s$.

\medskip

For definiteness, let us assume $\Re(\beta_1-\beta_2)\geq 0$. The case $\Re(\beta_1-\beta_2)< 0$ can be treated similarly.
Following \cite{DIK22}, where a similar situation arose, we proceed as follows.
We have
\begin{align}\label{Delta12}
&\frac{2}{|s|}\Delta_1(\zeta;s)=h_1(\zeta;s)\sigma_+ + \bigO(|s|^{-1}),
&\frac{2}{|s|}\Delta_2(\zeta;s)=h_2(\zeta;s)\sigma_- + \bigO(|s|^{-1}),
\end{align}
as $s\to -i\infty$, where
\be
\sigma_+=\begin{pmatrix}0&1\\0&0\end{pmatrix},\qquad \sigma_-=\begin{pmatrix}0&0\\1&0\end{pmatrix},
\ee
and
\begin{align}
&h_1(\zeta;s)=-\frac{1}{\zeta -i}(\zeta+i)^{-2\beta_2}e^{-\frac{i|s|}{2}}\left(\frac{|s|}{2}\right)^{-1+\beta_1-\beta_2}
e^{i\pi(\alpha_1-3\beta_1)}\frac{\Gamma(1+\alpha_1-\beta_1)}{\Gamma(\alpha_1+\beta_1)},\\
&h_2(\zeta;s)=\frac{1}{\zeta +i}(\zeta-i)^{2\beta_1}e^{-\frac{i|s|}{2}}\left(\frac{|s|}{2}\right)^{-1+\beta_1-\beta_2}
e^{i\pi(\alpha_2+\beta_2)}\frac{\Gamma(1+\alpha_2+\beta_2)}{\Gamma(\alpha_2-\beta_2)}.
\end{align}
Write $R$ in the form
\begin{equation}\label{XR}
R(\zeta)=\widehat R(\zeta)X(\zeta),
\end{equation}
where $X$ is a solution to the following RH problem, which is the RH problem for $R$, but
all the jump matrix elements of order less than the highest in $|s|$, except $1$ on the diagonal, are set to zero.

\subsubsection*{RH problem for $X$}
\begin{itemize}
    \item[(a)] $X:\mathbb C\setminus \left(\partial \mathcal U_1\cup \partial\mathcal U_2\right)\to \mathbb C^{2\times 2}$ is
    analytic.
    \item[(b)] $X$ satisfies the jump conditions
    \begin{align}
&\label{jump hatX2}X_+(\zeta)=X_-(\zeta)(I+h_1(\zeta)\sigma_+),\qquad
    \zeta\in\partial \mathcal U_1,\\
    &\label{jump hatX3}X_+(\zeta)=X_-(\zeta)(I+h_2(\zeta)\sigma_-),\qquad
    \zeta\in\partial \mathcal U_2.
    \end{align}
    \item[(c)] As $\zeta\to\infty$, we have
    \begin{equation}\label{hatX as+}
    X(\zeta)=I+\bigO(\zeta^{-1}).
    \end{equation}
    \end{itemize}
This RH problem can be solved explicitly as follows.
We look for the solution $X$ in the form
\begin{align}
&\label{hat X1}X(\zeta)=I+\frac{1}{\zeta-i}W_1 + \frac{1}{\zeta+i}W_2, &\zeta\in \mathbb C\setminus \overline{\left(\mathcal U_1\cup \mathcal U_2\right)},\\
&\label{hat X2}X(\zeta)=\left(I+\frac{1}{\zeta-i}W_1 + \frac{1}{\zeta+i}W_2\right)(I-h_1(\zeta)\sigma_+), &\zeta\in \mathcal U_1,\\
&\label{hat X3}X(\zeta)=\left(I+\frac{1}{\zeta-i}W_1 + \frac{1}{\zeta+i}W_2\right)(I-h_2(\zeta)\sigma_-), &\zeta\in \mathcal U_2,
\end{align}
where the $2\times 2$ constant in $\zeta$ matrices $W_1$ and $W_2$ will now be determined.
By the analyticity of $X$ at $\pm i$, the singular terms on the right hand side of (\ref{hat X2}) and (\ref{hat X3}) must vanish. The vanishing of the term with $(\zeta- i)^{-2}$ in $\mathcal U_1$ and with $(\zeta+i)^{-2}$ in $\mathcal U_2$ gives
\begin{align}
&\label{W111}W_{1,11}=W_{1,21}=0,\\
&\label{W212}W_{2,12}=W_{2,22}=0.
\end{align}
Similarly, the vanishing of the terms with $(\zeta\mp i)^{-1}$ gives
\begin{align}
&\label{W112}W_{1,12}=\widehat h_1(i)\left(1+\frac{W_{2,11}}{2i}\right),\\
&\label{W122}W_{1,22}=\frac{1}{2i}\widehat h_1(i)W_{2,21},\\
&\label{W221}W_{2,21}=\widehat h_2(-i)\left(1-\frac{W_{1,22}}{2i}\right),\\
&\label{W211}W_{2,11}=-\frac{1}{2i}\widehat h_2(-i)W_{1,12},
\end{align}
where
\begin{equation}\widehat h_1(\zeta)=(\zeta-i)h_1(\zeta),\qquad \widehat h_2(\zeta)=(\zeta+i)h_2(\zeta).
\end{equation}
This system of equations is solved explicitly, in particular,
\begin{align}
&W_{2,11}=-2i\frac{\gamma(s)}{1+\gamma(s)},\label{W2112}\\ &\gamma(s)=-\frac{1}{4}\widehat h_1(i)\widehat h_2(-i)=
\left|s\right|^{2(-1+\beta_1-\beta_2)}e^{-i|s|}e^{i\pi(\alpha_1+\alpha_2)}\label{Wextra}\\
& \label{gamma}\hspace{6cm} \times \quad\frac{\Gamma(1+\alpha_1-\beta_1)\Gamma(1+\alpha_2+\beta_2)}{\Gamma(\alpha_1+\beta_1)\Gamma(\alpha_2-\beta_2)}.
\end{align}

Let us now derive the RH conditions for $\widehat R$, defined in (\ref{XR}). Using the jump conditions for $R$ and $X$, the form (\ref{hat X1}) of $X(z)$, and the fact that $\sigma_\pm$ is nilpotent, we obtain
\begin{multline}
\widehat R_+= R_+ X_+^{-1}=R_-\left(I+h_1(\zeta)\sigma_++\frac{2}{|s|}
\widetilde\Delta_1(\zeta)+\bigO(|s|^{-2+|||\beta|||})\right)X_+^{-1}\\
=\widehat R_-\left(I+\frac{2}{|s|}\widetilde\Delta_1(\zeta)+\bigO(|s|^{-2+|||\beta|||})\right),
\end{multline}
on $\partial \mathcal U_1$, where
\begin{equation}\label{tilde Delta1}
\frac{2}{|s|}\widetilde\Delta_1(\zeta)=\frac{2}{|s|}\Delta_1(\zeta)-h_1(\zeta)\sigma_+,
\end{equation}
and similarly
on $\partial \mathcal U_2$, with $h_2$ instead of $h_1$ and $\widetilde\Delta_2$ instead of $\widetilde\Delta_1$, where
\begin{equation}\label{tilde Delta2}
\frac{2}{|s|}\widetilde\Delta_2(\zeta)=\frac{2}{|s|}\Delta_2(\zeta)-h_2(\zeta)\sigma_-.
\end{equation}

Thus, we have the following RH problem for $\widehat R$:
\subsubsection*{RH problem for $\widehat R$}
\begin{itemize}
    \item[(a)] $\widehat R:\mathbb C\setminus \Gamma_R$ is
    analytic.
    \item[(b)] $\widehat R$ satisfies the jump conditions
    \begin{equation}
    \widehat R_+(\zeta)=\widehat R_-(\zeta)\left(I+\frac{2}{|s|}\widetilde\Delta_j(\zeta)+\bigO(|s|^{-2+|||\beta|||})\right),\qquad s\to -i\infty,
    \end{equation}
    for $\zeta\in\partial\mathcal U_j$, $j=1,2$.
    \item[(c)] As $\zeta\to\infty$, we have
    \begin{equation}\label{X as}
    \widehat R(\zeta)=I+\widehat R_1\zeta^{-1}+\bigO(\zeta^{-2}).
    \end{equation}
    \end{itemize}
This is a small-norm RH problem, and as before, we can conclude that it is solvable for $|s|$ sufficiently large, and since $\wt\Delta_j/|s|=\bigO(|s|^{-1})$,
\be\label{whRas}
\widehat R(\zeta)=I+\bigO\left(\frac{1}{|s|(|\zeta|+1)}\right)
\ee
uniformly for $\zeta\in\mathbb C\setminus\Gamma_R$ as $|s|\to\infty$.
More precisely,
\begin{equation}\label{X as2}
\widehat R(\zeta)=I+\frac{1}{\pi i|s|}\sum_{j=1}^2\int_{\partial\mathcal U_j}\frac{\widetilde\Delta_j(\xi)}{\xi-\zeta}d\xi+\bigO(|s|^{-2+|||\beta|||}/(|\zeta|+1)),
\end{equation}
and $\widehat R_1$ in (\ref{X as}) is given by
\begin{equation}\label{X1}
\widehat R_1=\frac{2}{|s|}\left(\Res(\widetilde\Delta_1;i)+\Res(\widetilde\Delta_2;-i)\right)+\bigO(|s|^{-2+|||\beta|||}).
\end{equation}
By inverting the transformation $U\mapsto R$ and using (\ref{XR}),
we obtain
\begin{equation}
U(\zeta)=\left(\frac{|s|}{2}\right)^{-\frac{\beta_1+\beta_2}{2}\sigma_3}\widehat R(\zeta)X(\zeta)\left(\frac{|s|}{2}\right)^{\frac{\beta_1+\beta_2}{2}\sigma_3}P^{(\infty)}(\zeta),\qquad \zeta\in\mathbb C\setminus (\overline{\mathcal U_1}\cup\overline{\mathcal U_2}).
\end{equation}
Substituting the explicit formula (\ref{hat X1}) for $X$ and the asymptotic expansion (\ref{X as}), we obtain
\begin{multline}
U(\zeta)=\left(\frac{|s|}{2}\right)^{-\frac{\beta_1+\beta_2}{2}\sigma_3}\left(I+\frac{\widehat R_1}{\zeta}+\bigO(\zeta^{-2})\right)\left(I+\frac{W_1+W_2}{\zeta}+\bigO(\zeta^{-2})\right)\\ \times \left(\frac{|s|}{2}\right)^{\frac{\beta_1+\beta_2}{2}\sigma_3}P^{(\infty)}(\zeta),\qquad \zeta\in\mathbb C\setminus (\overline{\mathcal U_1}\cup\overline{\mathcal U_2}),
\end{multline}
as $\zeta\to\infty$.
Comparing this with (\ref{U as}), we conclude that
\begin{equation}\label{555}
\Psi_{1}=\left(\frac{|s|}{2}\right)^{-\frac{\beta_1+\beta_2}{2}\sigma_3}\left(\widehat R_1+W_1+W_2\right)\left(\frac{|s|}{2}\right)^{\frac{\beta_1+\beta_2}{2}\sigma_3},
\end{equation}
and therefore by (\ref{W111}), (\ref{Wextra}), and (\ref{X1}),
\begin{equation}
q(s)=\Psi_{1,11}=\frac{2}{|s|}\left(\Res(\widetilde\Delta_{1,11};i)+\Res(\widetilde\Delta_{2,11};-i)\right)+\frac{2}{i}\frac{\gamma(s)}{1+\gamma(s)}+\bigO(|s|^{-2+|||\beta|||}),
\end{equation}
as $|s|\to\infty$.
By (\ref{tilde Delta1})--(\ref{tilde Delta2}), (\ref{Delta1})--(\ref{Delta2}), we obtain
\begin{equation}
q(s)=\frac{2}{|s|}\left(\alpha_1^2+\alpha_2^2-\beta_1^2-\beta_2^2
\right)+\frac{2}{i}\frac{\gamma(s)}{1+\gamma(s)}+\bigO(|s|^{-2+|||\beta|||}).
\end{equation}
By (\ref{sigma4b}), we find
\begin{equation}\label{sigma at infty}
\sigma(s)=\frac{\beta_2-\beta_1}{2}s-\frac{1}{2}(\beta_1-\beta_2)^2
+s\frac{\gamma(s)}{1+\gamma(s)}+\bigO(|s|^{-1+|||\beta|||}),\qquad s\to -i\infty.
\end{equation}
Recall that $\gamma(s)$ is given by (\ref{gamma}), and thus $s\frac{\gamma}{1+\gamma}$ is of order $|s|^{-1+2|||\beta|||}$.
This proves (\ref{asinfty}) if $\Re(\bt_1-\bt_2)\ge 0$. The case $\Re(\bt_1-\bt_2)< 0$
can be studied similarly.

\section{Asymptotics for $\Psi$ as $s\to -i 0_+$}\label{sec s0}
In this section we will construct an asymptotic solution for $\Psi(\zeta)$ as $s\to -i 0_+$
out of 2 parametrices: the local one (at $\zeta=0$) will be given in terms of the hypergeometric function,
and the global one at infinity, in terms of the confluent hypergeometric function. 
The asymptotic solution will be valid uniformly in $\zeta$. The results of this section can be easily extended to $s$ in a neighborhood of
$-i\mathbb R^+$.

We assume here that
$\Re\alpha_1,\Re\alpha_2,\Re(\al_1+\al_2)>-1/2$, 
$\al_k\pm\bt_k\neq -1,-2,\dots$, $k=1,2$, $(\al_1+\al_2)\pm(\bt_1+\bt_2)\neq -1,-2,\dots$.

\subsection{Modified RH-problem}
It is convenient to consider the following transformation of $\Psi$:
let \be\label{whPsi}
\widehat \Psi (\lb)=e^{-\frac{s}{4}\sigma_3}e^{-\frac{i\pi}{2}(\al_1-\bt_1-\al_2+\bt_2)\si_3}
\Psi\left(\frac{2}{|s|}\lb+i\right)e^{\frac{i\pi}{2}(\al_1-\bt_1-\al_2+\bt_2)\si_3}e^{-2\pi i\beta_2\sigma_3}
 \ee
  for $\Re\lambda>0$, $-|s|<\Im\lambda<0$, and
  \be\label{whPsi2}
\widehat\Psi(\lambda)=e^{-\frac{s}{4}\sigma_3}e^{-\frac{i\pi}{2}(\al_1-\bt_1-\al_2+\bt_2)\si_3}
\Psi\left(\frac{2}{|s|}\lb+i\right)e^{\frac{i\pi}{2}(\al_1-\bt_1-\al_2+\bt_2)\si_3}\ee  elsewhere.
Recall that $s\in -i\mathbb R^+$, and
note that the interval of the imaginary axis $[i,-i]$ in the $\zeta=\frac{2}{|s|}\lb+i$-variable
is mapped onto $[0,s]$ in the $\lb$-variable. The function $\widehat\Psi$ has jumps for $\lambda$ on $e^{\frac{i\pi}{4}}\mathbb R^+$, $e^{\frac{3i\pi}{4}}\mathbb R^+$, $(s,0)$, $s+e^{-\frac{3i\pi}{4}}\mathbb R^+$, and $s+e^{-\frac{i\pi}{4}}\mathbb R^+$. The $4$ semi-infinite jump rays can be deformed freely by analytic continuation of the RH solution $\widehat\Psi$ from one region to another, as long as they do not intersect and as long as they tend to infinity in a sufficiently narrow sector containing the original ray. It is convenient to deform the lines $s+e^{-\frac{3i\pi}{4}}\mathbb R^+$ and $s+e^{-\frac{i\pi}{4}}\mathbb R^+$ in such a way that they coincide with $e^{-\frac{3i\pi}{4}}\mathbb R^+$ and $e^{-\frac{i\pi}{4}}\mathbb R^+$, except in a fixed neighborhood $\mathcal U_0$ of $0$. We choose $\mathcal U_0$ so that it contains the interval $[0,s]$,
and a jump contour as indicated in Figure \ref{figure: whPsi}.

\begin{figure}[t]
\begin{center}
    \setlength{\unitlength}{0.8truemm}
    \begin{picture}(110,95)(-10,-7.5)
    \put(50,60){\thicklines\circle*{.8}}
    \put(50,40){\thicklines\circle*{.8}}
    \put(51,56){\small $0$}
    \put(51,41){\small $s$}
    \put(50,60){\thicklines\circle*{.8}}
    \put(69,60){\thicklines\vector(1,0){.0001}}
    \put(50,60){\line(1,1){25}}

    \put(70,35){\line(1,-1){25}}

    \put(50,60){\line(-1,1){25}}

    \put(30,35){\line(-1,-1){25}}
    \put(30,35){\line(4,1){20}}
    \put(70,35){\line(-4,1){20}}
    \put(50,60){\line(1,0){35}}
    \put(50,40){\line(0,1){20}}
    \put(65,75){\thicklines\vector(1,1){.0001}}

    \put(85,20){\thicklines\vector(-1,1){.0001}}

    \put(50,51){\thicklines\vector(0,1){.0001}}
    \put(35,75){\thicklines\vector(-1,1){.0001}}

    \put(15,20){\thicklines\vector(1,1){.0001}}

    \put(74,88){\small $\widehat J_1=\begin{pmatrix}1&e^{\pi i(\alpha_1+\alpha_2-\beta_1-\beta_2)}\\0&1\end{pmatrix}$}
    \put(-25,88){\small $\widehat J_2=\begin{pmatrix}1&0\\-e^{-\pi i(\alpha_1+\alpha_2-\beta_1-\beta_2)}&1\end{pmatrix}$}
 \put(86,59){\small
$\widehat J_5=e^{2\pi i(\beta_1+\beta_2)\sigma_3}$}
\put(51,48){\small $\widehat J_6=\begin{pmatrix}0&e^{-i\pi(\al_1-\bt_1-\al_2-\bt_2)}\\-e^{i\pi(\al_1-\bt_1-\al_2-\bt_2)}&e^{-2\pi i\beta_2}\end{pmatrix}$}
    \put(-22,3){\small $\widehat J_3=\begin{pmatrix}1&0\\-e^{\pi i(\alpha_1+\alpha_2-\beta_1-\beta_2)}&1\end{pmatrix}$}
    \put(74,3){\small $\widehat J_4=\begin{pmatrix}1&e^{-\pi i(\alpha_1+\alpha_2-\beta_1-\beta_2)}\\0&1\end{pmatrix}$}

    \put(-2,48){III}
    \put(88,71){I}
    \put(47,7){IV}
    \put(47,75){II}
    \put(92,31){V}
    \end{picture}
    \caption{The jump contour and jump matrices for $\widehat\Psi$.}
    \label{figure: whPsi}
\end{center}
\end{figure}
Then  $\widehat\Psi$ satisfies the following RH conditions.
\subsubsection*{RH problem for $\widehat\Psi$}
\begin{itemize}
    \item[(a)] $\widehat\Psi:\mathbb C\setminus \widehat\Gamma \to \mathbb C^{2\times 2}$ is analytic; $\widehat\Gamma$ is the contour depicted in Figure \ref{figure: whPsi}.
    \item[(b)] $\widehat\Psi$ satisfies the jump conditions
    \begin{equation}\label{jump whPsi}\widehat\Psi_+(\lambda)=\widehat\Psi_-(\lambda)\widehat J_k,\qquad
    \lambda\in\widehat\Gamma_k,\end{equation}  where $\widehat J_k$, $k=1,\ldots, 6$ are the matrices given in Figure \ref{figure: whPsi}, and $\widehat \Gamma_k$ are the corresponding parts of the contour.
    \item[(c)] As $\lambda\to\infty$, we have in all regions:
    \begin{equation}\label{whPsi as}
    \widehat\Psi(\lb)=\left(I+\frac{\widehat \Psi_1}{\lambda}+\bigO(\lambda^{-2})\right)
    \lambda^{-(\beta_1+\beta_2)\sigma_3}
        e^{-\frac{\lambda}{2}\sigma_3},    \end{equation}
                \end{itemize}
where $0<\arg\lb<2\pi$ and
\begin{equation}
\label{whC1}
\widehat \Psi_1=\frac{|s|}{2}\left(e^{-\frac{i\pi}{2}(\al_1-\bt_1-\al_2+\bt_2)\si_3}e^{-\frac{s}{4}\sigma_3}\Psi_1e^{\frac{s}{4}\sigma_3}e^{\frac{i\pi}{2}(\al_1-\bt_1-\al_2+\bt_2)\si_3}-2i\beta_2\sigma_3\right).
\end{equation}

\subsection{Global parametrix}
We note that $\widehat \Psi$ defined in (\ref{whPsi})--(\ref{whPsi2}) has the same condition
at infinity and the same jumps in $\mathbb C\setminus\mathcal U_0$ as the auxiliary
function $M^{(\al,\bt)}$ of Section \ref{secaux} with parameters $\al=\al_1+\al_2$, $\bt=\bt_1+\bt_2$.
The function $M(\lambda)\equiv M^{(\al_1+\al_2,\bt_1+\bt_2)}(\lb)$ will serve as a parametrix for $\widehat\Psi$ in $\mathbb C\setminus\mathcal U_0$.

\subsection{Local parametrix}
We look for a function $P_0$ satisfying the following conditions
\subsubsection*{RH problem for $P_0$}
\begin{itemize}
\item[(a)] $P_0:\mathcal U_0\setminus \widehat\Gamma\to \mathbb C^{2\times 2}$ is analytic,
\item[(b)] $P_0$ satisfies the same jump conditions as $\widehat\Psi$ for $\lambda\in\mathcal U_0\cap\widehat\Gamma$,
\item[(c)] for $\lambda\in\partial\mathcal U_0$, $P_0(\lambda)=M(\lambda)\left(I+o(1)\right)$ as $s\to -i 0_+$,
\item[(d)] $\widehat\Psi(\lambda)P_0(\lambda)^{-1}$ is analytic at $0$ and $s$.
\end{itemize}

Since $\widehat\Psi(\lambda)$ has 2 singular points inside $\mathcal U_0$, it makes sense to try to construct $P_0$ in terms of the hypergeometric functions, similarly to \cite{CIK2}.

Recall that for $c\neq 0,-1,-2,\dots$, the hypergeometric function is
represented by the standard series
\[
F(a,b,c,z)= 1+\sum_{n=1}^{\infty}\frac{a(a+1)\cdots(a+n-1)b(b+1)\cdots(b+n-1)}
{c(c+1)\cdots(c+n-1)}\frac{z^n}{n!},
\]
converging in the disk $|z|\le r<1$ of any radius $r<1$,
and is extended to an analytic function in the plane with the cut $[1,+\infty)$.
We choose the argument of $z$ between $0$ and $2\pi$.
The hypergeometric function we need here is the following:
\be\label{hyp1}
\widehat F(\lb;s)\equiv F\left(1,1+2\al_1,2+2(\al_1+\al_2),\frac{s}{\lb}\right).
\ee
This function of $\lb$ is thus defined 
on the plane with the cut on the interval of the imaginary axis $[0,s]=[0,e^{-i\pi/2}|s|]$.
(The jump of $\widehat F(\lb)$ on $[0,s]$ can be obtained using the transformation
of the hypergeometric function between the arguments $z$ and $1/z$.) 

We now let
\be\label{J0}
J(\lb;s)=-\frac{1}{\pi}
\frac{|s|^{1+2(\al_1+\al_2)}}{\lb}\frac{\Ga(1+2\al_1)\Ga(1+2\al_2)}{\Ga(2+2(\al_1+\al_2))}
\widehat F(\lb;s).
\ee
By a standard integral representation of the hypergeometric function, $J(\lb;s)$ can also be written as follows:
\be\label{J}
J(\lb;s)=\frac{1}{\pi i}\int_s^0\frac{|\xi|^{2\alpha_1}|\xi-s|^{2\alpha_2}}{\xi-\lambda}d\xi.
\ee
This representation implies that on the cut $[0,s]$ oriented upwards,
\be\label{Jjump}
J(\lb)_+=J(\lb)_- + 2|\lb|^{2\alpha_1}|\lb-s|^{2\alpha_2},\qquad \lb\in(0,s).
\ee

We are now ready to construct the parametrix $P_0$.
First, consider the case where $2(\alpha_1+\alpha_2)\notin\mathbb N\cup\{0\}$.
Then set 
\be\label{P}
P_0^{(3)}(\lb)=L(\lb)\begin{pmatrix}1&
c_0 J(\lb;s)\\0&1\end{pmatrix}
\lb^{\al_1\sigma_3}(\lb-s)^{\al_2\si_3}e^{2\pi i\al_1\si_3}\wt G_3,\qquad \lb\in \mathcal U_0,
\ee
with the cut of $\lb^{\al_1}$ along the line $\widehat\Gamma_2$,
the cut of $(\lb-s)^{\al_2}$  along the line $\widehat\Gamma_3$,
and their arguments $\frac{\pi}{2}\al_j$ on $i\mathbb R^{+}$ (i.e., with the same choice of branches in $\mathcal U_0$
as for $(\zeta-i)^{\al_1}$ in (\ref{Psi i}), and $(\zeta+i)^{\al_2}$ in (\ref{Psi -i})).
Here $L$ is given by (\ref{ME}), the matrix $\wt G_3$ by (\ref{MG3})
with $\al=\al_1+\al_2$, $\bt=\bt_1+\bt_2$, and
\begin{multline}\label{c0}
c_0={1\over 2 }e^{2\pi i(\al_1+\al_2)}\left[
e^{i\pi(\al_1-\al_2)}\frac{\sin\pi(\al_1+\al_2+\bt_1+\bt_2)}{\sin 2\pi(\al_1+\al_2)}\right.\\
\left.
+e^{-i\pi(\al_1-\al_2)}\frac{\sin\pi(\al_1+\al_2-\bt_1-\bt_2)}{\sin 2\pi(\al_1+\al_2)}
-e^{i\pi(\bt_1-\bt_2)}\right].
\end{multline}
With $P_0^{(3)}(\lb)$ given by (\ref{P}) set
\begin{equation}\label{def P00}
\begin{aligned}
P_0(\lambda)&=P_0^{(3)}(\lambda),\qquad\mbox{ in region III,}\\
P_0(\lambda)&=
P_0^{(3)}(\lambda)\wt G_3^{-1}e^{-2\pi i\al_1\si_3}\wt G_3\widehat J_2^{-1}\times
\begin{cases}
I,&\mbox{ in region II,}\\
\widehat J_1^{-1},&\mbox{ in region I,}\\
\widehat J_1^{-1}\widehat J_5^{-1},&\mbox{ in region V,}\\
\widehat J_1^{-1}\widehat J_5^{-1}\widehat J_4,&\mbox{ in region IV.}
\end{cases}
\end{aligned}
\end{equation}

We then have the following.

\begin{proposition}\label{prop local}
Let $2(\alpha_1+\alpha_2)\notin \mathbb N\cup\{0\}$. Then the function (\ref{def P00}) solves the RH problem for $P_0$.
\end{proposition}
\begin{proof}
Condition (a) of the RH problem for $P$ is satisfied by construction, as well as the jump relations on $\widehat \Gamma_1,\widehat \Gamma_4, \widehat \Gamma_5$.
The jump condition on $\widehat\Gamma_2$ follows from the definition of $P_0$ in regions II and III, and from the fact that $P_0^{(3)}$ has a jump on $\widehat\Gamma_2$ because of the branch cut of $\lb^{\al_1\sigma_3}$:
for $\lambda\in\widehat\Gamma_2$, \[
P_{0,-}(\lambda)^{-1}P_{0,+}(\lambda)=\widehat J_2 \left(
\wt G_3^{-1}e^{-2\pi i\al_1\si_3}\wt G_3\right)^{-1} \widetilde G_3^{-1}e^{-2\pi i\alpha_1\sigma_3}\widetilde G_3=\widehat J_2.
\]
On $\widehat\Gamma_3$, we have after a similar calculation using (\ref{P}), taking into account the branch cut of $(\lambda-s)^{\alpha_2\sigma_3}$, and using the value of $\wt g(\al_1+\al_2,\bt_1+\bt_2)$ in (\ref{MG3}),
\begin{multline}\notag
P_{0,-}(\lambda)^{-1}P_{0,+}(\lambda)=\widehat J_4^{-1}\widehat J_5\widehat J_1\widehat J_2 
\wt G_3^{-1}e^{2\pi i\alpha_1\sigma_3}\wt G_3
 P_{0,-}^{(3)}(\lambda)^{-1}P_{0,+}^{(3)}(\lambda)\\
=\widehat J_4^{-1}\widehat J_5\widehat J_1\widehat J_2 
\wt G_3^{-1}e^{2\pi i(\al_1+\al_2)\sigma_3}\wt G_3=
\widehat J_3.
\end{multline}
On the interval $\widehat\Gamma_6=[s,0]$, by (\ref{P}), (\ref{J}), and (\ref{c0}), we obtain
\begin{align*}
&P_{0,-}(\lambda)^{-1}P_{0,+}(\lambda)=\widehat J_5\widehat J_1\widehat J_2 
\wt G_3^{-1}e^{2\pi i\alpha_1\sigma_3}\wt G_3
 P_{0,-}^{(3)}(\lambda)^{-1}P_{0,+}^{(3)}(\lambda)\\
&\qquad =\begin{pmatrix}0&e^{\pi i(\alpha_1+\alpha_2+\beta_1+\beta_2)}\\-e^{-\pi i(\alpha_1+\alpha_2+\beta_1+\beta_2)}&e^{-2\pi i(\beta_1+\beta_2)}\end{pmatrix}
\wt G_3^{-1}e^{2\pi i\alpha_1\sigma_3}\wt G_3\\
&\hspace{9cm}\times \ \begin{pmatrix}1&2c_0e^{-\pi i(3\alpha_1+\alpha_2)}\\0&1\end{pmatrix}\\
&\qquad=\widehat J_6.
\end{align*}
In fact, it is the requirement that here $P_{0,-}(\lambda)^{-1}P_{0,+}(\lambda)=\widehat J_6$ which
fixes the value (\ref{c0}) of $c_0$.

Now we prove the matching condition $P_0(\lb)M(\lb)^{-1}=I+\bigO(|s|)$
on the boundary $(\partial\mathcal U_0)\cap\mbox{region III}$  as $s\to -i0_+$. Here we use
(\ref{M3}) with $\al=\al_1+\al_2$, $\bt=\bt_1+\bt_2$. The branch of $\lb^{\al_1+\al_2}\equiv
\lb_M^{\al_1+\al_2}$ in (\ref{M3})
was chosen with arguments between $0$ and $2\pi$. On the other hand, the branch of 
$\lb^{\al_1}\equiv \lb_P^{\al_1}$ in (\ref{P}) was chosen with arguments between $-5\pi/4$ and $3\pi/4$.
Therefore $\lb_M^{\al_1}= \lb_P^{\al_1} e^{2\pi i\al_1}$ for $\lb$ in region III.
Taking this into account, we obtain
\[
P(\lb)M(\lb)^{-1}=L(\lb)\begin{pmatrix}1&
c_0 J(\lb;s)\\0&1\end{pmatrix}
\lb^{-\al_2\sigma_3}(\lb-s)^{\al_2\si_3}L(\lb)^{-1}
\]
for $\lb\in (\partial\mathcal U_0)\cap\mbox{region III}$. Here we can (and do) choose 
the branch of $\lb^{-\al_2}$ with arguments between $-\pi/2$ and $-5\pi/2$, and that of $(\lb-s)^{\al_2}$ with arguments between $-\pi/2$ and $3\pi/2$.
Note that this expression can be extended to the whole plane as an analytic function outside the cut $[0,s]$.
As $s\to -i0_+$, we obtain uniformly on the boundary $\partial\mathcal U_0$:
\begin{equation}\label{asPN1}
P(\lb)M(\lb)^{-1}=I+\Delta_1(\lb)+\bigO(|s|^2)+\bigO(|s|^{2+4(\al_1+\al_2)})=I+o(1),\qquad\lb\in\partial\mathcal U_0,
\end{equation}
where
\begin{multline}\label{delta}\Delta_1(\lb)=
-\al_2\frac{s}{\lb}L(\lb)\si_3 L(\lb)^{-1}\\
 -\frac{c_0}{\pi\lb}|s|^{1+2(\al_1+\al_2)}\frac{\Ga(1+2\al_1)\Ga(1+2\al_2)}
{\Ga(2+2(\al_1+\al_2))}L(\lb)\begin{pmatrix}
0&1\\
0&0\end{pmatrix}L(\lb)^{-1}.
\end{multline}
Here we used (\ref{J}) and the connection between Beta and Gamma functions. Alternatively, we can
use (\ref{J0}). 

Finally, we need to prove condition (d) of the RH problem. Since $\widehat\Psi$ and $P_0$ have the same jump relations, the possible singularities of $\widehat\Psi(\lb)P_0(\lb)^{-1}$ at $0$ and $s$ are isolated, and it is easily verified by the construction of $P_0$ and the behavior of $\widehat\Psi$ at $0$ and $s$ that the singularities cannot be essential. Therefore, to conclude that the singularities are removable, it suffices to check that $\widehat\Psi(\lb)P_0(\lb)^{-1}$ is $o(|\lb|^{-1})$ as $\lambda \to 0$ in region III, and that it is $o(|\lb-s|^{-1})$ as $\lambda\to s$ in region III.
Let us consider the behavior of $\widehat\Psi(\lb)P_0(\lb)^{-1}$ as $\lb\to 0$ in region III.
From the conditions (\ref{Psi i}), (\ref{Psi i 2}), we obtain as $\lb\to 0$ in region III:
\be
\begin{aligned}\label{whPsi0}
\widehat\Psi(\lb)=&\widehat F_1(\lb) \lb^{\al_1\si_3}\begin{pmatrix}
1& \ell(\lb) e^{-i\pi(\al_1-\bt_1-\al_2+\bt_2)}\\
0& 1\end{pmatrix},\\ \widehat F_1(\lb)=&
e^{-\frac{s}{4}\si_3}e^{-i{\pi\over 2}(\al_1-\bt_1-\al_2+\bt_2)\si_3}F_1\left(\frac{2}{|s|}\lb+i\right)
e^{i{\pi\over 2}(\al_1-\bt_1-\al_2+\bt_2)\si_3}\left(\frac{2}{|s|}\right)^{\al_1\si_3},
\end{aligned}
\ee
where $\ell=g$, with $g$ given by (\ref{G}) for $2\al_1\neq 0,1,\dots$; and
$\ell=g_{int}\ln(2\lb/|s|)$, with $g_{int}$ given by (\ref{gint}) for $2\al_1= 0,1,\dots$.

Multiplying (\ref{whPsi0}) on the right by $P_0(\lb)^{-1}$ and substituting (\ref{P}), we obtain after a straightforward analysis of (\ref{J}) that
\begin{equation}
\widehat\Psi(\lb)P_0(\lb)^{-1}=\widehat F_1(\lb)\begin{pmatrix}1&\bigO(|\lambda|^{2\alpha_1})+\bigO(|\ln\lb |)+\bigO(1)&\\0&1
\end{pmatrix}
L(\lambda)^{-1},
\end{equation}
as $\lambda\to 0$ in region III.
Since $\widehat F_1$ and $L$ are analytic at $0$ and $\Re\alpha_1>-1/2$, this implies that $\widehat\Psi(\lb)P_0(\lb)^{-1}$ is analytic at $0$.
In a similar way we obtain
\begin{equation}
\widehat\Psi(\lb)P_0(\lb)^{-1}=\widehat F_2(\lb)\begin{pmatrix}1&\bigO(|\lambda|^{2\alpha_2})+\bigO(|\ln\lb |)+\bigO(1)&\\0&1
\end{pmatrix}
L(\lambda)^{-1},
\end{equation}
as $\lambda\to s$ in region III.
Hence $\widehat\Psi(\lb)P_0(\lb)^{-1}$ is analytic at $s$.

\end{proof}

\begin{remark}
The representation (\ref{J0}) allows us to obtain the full expansion of $P_0(\lb)$ near its singularities
without much effort (cf. \cite{CIK2}). 
The exact cancellation of singular parts of $P_0(\lb)$ and $\widehat\Psi(\lb)$
at $0$, $s$, in the expression $\widehat\Psi(\lb)P_0(\lb)^{-1}$ gives an alternative way to
fix the value (\ref{c0}) of the constant $c_0$. 
Consider, for example, the case of $\lb$ near $0$ and $2\al_1\neq 0,1,2\dots$. 
Then we can use the following standard transformation of the hypergeometric function:
\be\label{z to 1 over z}
\eqalign{
F(1,1+2\al_1,2+2(\al_1+\al_2),z)=
-\frac{\pi}{\sin 2\pi\al_1}\frac{\Ga(2+2(\al_1+\al_2))}{\Ga(1+2\al_1)\Ga(1+2\al_2)}\\
\times\left(e^{i\pi} z^{-1}\right)^{1+2\al_1}\left(1-\frac{1}{z}\right)^{2\al_2}\\
+\frac{1+2(\al_1+\al_2)}{2\al_1}e^{i\pi} z^{-1}F(1,-2(\al_1+\al_2),1-2\al_1,1/z),}
\ee
to write $J(\lb;s)$ in the form
\begin{align}\label{Jsa}
J(\lb;s)&=J_{sing}(\lb)+J_{an}(\lb),\qquad
J_{sing}(\lb)=
\frac{e^{i\pi(3\al_1-\al_2)}}{i\sin 2\pi\al_1}\lb^{2\al_1}(\lb-s)^{2\al_2},\\
J_{an}(\lb)&=
-\frac{1}{i\pi}|s|^{2(\al_1+\al_2)}\frac{\Ga(2\al_1)\Ga(1+2\al_2)}{\Ga(1+2(\al_1+\al_2))}
F\left(1,-2(\al_1+\al_2),1-2\al_1,\frac{\lb}{s}\right),
\end{align}
with the branch of $\lb^{\al_1}$ corresponding to the arguments between $-\pi/2$ and $-5\pi/2$,
and the branch of $(\lb-s)^{\al_2}$, to the arguments between $3\pi/2$ and $-\pi/2$.
Since $J_{an}(\lb)$ is analytic at $\lb=0$, this representation shows the form of the singularity
of $J(\lb)$ at zero. Note that in region III, the branches in (\ref{Jsa}) coincide with those
in (\ref{P}).
We now write $P_0$ as $\lb\to 0$ in region III in the form
\begin{multline}\label{P0}
P_0(\lb)=L(\lb)\begin{pmatrix}1&
c_0 J_{an}(\lb)\\0&1\end{pmatrix}
(\lb-s)^{\al_2\si_3}
\begin{pmatrix}1&
c_0 J_{sing}(\lb)(\lb-s)^{-2\al_2}\\0&1\end{pmatrix}\\
\times\lb^{\al_1\sigma_3}e^{2\pi i\al_1\si_3}
\begin{pmatrix}1&
\wt g(\al_1+\al_2,\bt_1+\bt_2)\\0&1\end{pmatrix},
\end{multline}
where $\wt g(\al,\bt)$ is given by (\ref{MG3}).

Comparing this expression with (\ref{whPsi0}), we see that the condition of analyticity
of  $\widehat\Psi(\lb)P_0(\lb)^{-1}$ at zero is the condition of vanishing of the term 
with $\lb^{2\al_1}$ in $\widehat\Psi(\lb)P_0(\lb)^{-1}$, which is
\begin{multline}\notag
g e^{-i\pi(\al_1-\bt_1-\al_2+\bt_2)}-
\wt g(\al_1+\al_2,\bt_1+\bt_2)\\
=c_0 J_{sing}(\lb)(\lb-s)^{-2\al_2}\lb^{-2\al_1}e^{-4\pi i\al_1}=
c_0\frac{e^{-i\pi(\al_1+\al_2)}}{i\sin 2\pi\al_1}.
\end{multline}
Solving this condition for $c_0$, we again obtain (\ref{c0}).
\end{remark}

We now construct the parametrix in the remaining case $2(\alpha_1+\alpha_2)\in \mathbb N\cup\{0\}$. 
(Note, in particular, that the constant $c_0$ in (\ref{c0}) is not defined in this case.)
Set
\be\label{Jhalf}
\wt J(\lb;s)=\frac{1}{2}\frac{\partial}{\partial\al_1}J(\lb;s)=
\frac{1}{\pi i}\int_s^0\frac{|\xi|^{2\alpha_1}|\xi-s|^{2\alpha_2}\ln |\xi|}{\xi-\lambda} d\xi,
\ee
and then set
\begin{equation}\label{Phalf}
\wt P_0^{(3)}(\lb)=\widetilde L(\lb)\begin{pmatrix}1&
e_1\wt J(\lb;s)+ e_2 J(\lb;s)\\0&1\end{pmatrix}
\lb^{\al_1\sigma_3}(\lb-s)^{\al_2\si_3}e^{2\pi i\al_1\si_3}\begin{pmatrix}1&
m(\lambda)\\0&1\end{pmatrix},
\end{equation}
where
\begin{align}
e_1&=\frac{i}{\pi}e^{-i\pi(\al_1+\al_2)}\sin\pi(\al_1+\al_2+\bt_1+\bt_2)\sin2\pi\al_1,\label{c1}\\
e_2&=\frac{1}{2}\left(i\pi c_1+ e^{i\pi(-2\al_1+\bt_1+\bt_2)}-(-1)^{2(\al_1+\al_2)}e^{i\pi(\bt_1-\bt_2)}\right),
\label{c2}
\end{align}
the matrix $\widetilde L$ and $m(\lambda)$ are as in (\ref{M3half}) and (\ref{mhalf}), respectively,
with $\al=\al_1+\al_2$, $\bt=\bt_1+\bt_2$.

With $\wt P_0^{(3)}(\lb)$ given by (\ref{Phalf}) set
\begin{equation}\label{def P00half}
\begin{aligned}
P_0(\lambda)&=\wt P_0^{(3)}(\lambda),\qquad\mbox{ in region III,}\\
P_0(\lambda)&=\wt P_0^{(3)}(\lambda)\begin{pmatrix}1&-m(\lb)\\ 0&1\end{pmatrix}
e^{-2\pi i\alpha_1\si_3}\begin{pmatrix}1&m(\lb)\\ 0&1\end{pmatrix}\widehat J_2^{-1}\\
&\times
\begin{cases}
\wh J_2^{-1},&\mbox{ in region II,}\\
\wh J_2^{-1}\wh J_1^{-1},&\mbox{ in region I,}\\
(-1)^{2(\al_1+\al_2)} \wh J_3^{-1},&\mbox{ in region IV,}\\
(-1)^{2(\al_1+\al_2)} \wh J_3^{-1}\wh J_4^{-1},&\mbox{ in region V.}
\end{cases}
\end{aligned}
\end{equation}
Then the jump conditions hold (the jump condition on $\widehat \Gamma_6$ fixes the values (\ref{c1}),
(\ref{c2})
of the constants $c_1$ and $c_2$). One verifies conditions (c) and (d) in a similar way as above.
In particular, we obtain
\begin{equation}\label{asPN1half}
P_0(\lb)M(\lb)^{-1}=I+\bigO(|s\ln |s||),\qquad\lb\in\partial\mathcal U_0,\qquad s\to -i0_+.
\end{equation}
Thus, we have
\begin{proposition}\label{prop local half}
Let $2(\alpha_1+\alpha_2)\in \mathbb N\cup\{0\}$. Then the function (\ref{def P00half}) solves the RH problem for $P_0$.
\end{proposition}

\subsection{Solution of the RH problem and the asymptotics of $\si(s)$ for small $s$}
We define as usual
\be \label{def R0}
H(\lb)=\begin{cases}
\widehat\Psi(\lb) P_0^{-1}(\lb),& \lb\in\mathcal U_0,\\
\widehat\Psi(\lb) M^{-1}(\lb),& \lb\in\mathbb C\setminus\overline{\mathcal U_0}.
\end{cases}
\ee

We then have that $H$ is analytic in $\mathbb C\setminus\partial\mathcal U_0$  
with the jump (see (\ref{asPN1}), (\ref{asPN1half})) 
$P_0(\lb)M^{-1}(\lb)=1+o(1)$ uniformly for $\lb\in\partial\mathcal U_0$ as $s\to -i0_+$. As $\lb\to\infty$, we have $H(\lb)=I+\bigO(\lb^{-1})$ for any fixed $s$.
Therefore, this RH problem for $H$ is a small-norm problem solvable in the standard way by a Neumann series.
Consider the case $2(\alpha_1+\alpha_2)\notin\mathbb N\cup\{0\}$. We have
\begin{align}&\label{Rsol}
H(\lb)=I+H^{(1)}(\lb)+\bigO(|s|^2)+\bigO(|s|^{2+4(\al_1+\al_2)}),\\
& H^{(1)}(\lb)=\frac{1}{2\pi i}\int_{\partial\mathcal U_0}\frac{\Delta_1(\mu)}{\mu-\lb}d\mu,
\end{align}
uniformly in $\mathbb C\setminus\partial\mathcal U_0$ as $s\to -i0_+$, where $\Delta_1$ is given by (\ref{delta}).
Here $\partial\mathcal U_0$ is oriented clockwise.

From the asymptotics for $H$, we can obtain the asymptotics for $\widehat\Psi(\lb)$ as $s\to -i0_+$, and
hence we can compute the asymptotics for the Painlev\'e function $\si(s)$ in this limit.
By (\ref{sigma4b}) and (\ref{C1}), we need to compute
$\Psi_{1,11}(s)$, which is the coefficient of $1/\lb$ in the large $\lb$ expansion (\ref{Psi as}) of
$\Psi(\lb)$.
First, we observe by (\ref{whC1}) that
\be\label{whPsiPsi}
\Psi_{1,11}=\frac{2}{|s|}\widehat \Psi_{1,11}+2i\beta_2.
\ee
To compute $\widehat \Psi_{1,11}$ for small $|s|$, we use the asymptotic solution of the $\widehat\Psi$ problem.
We have as $\lb\to \infty$,
\be\label{whPsi as2}
\widehat \Psi(\lb)=H(\lb)M(\lb)=
\left(I+\frac{H_1}{\lb}+\bigO(\lb^{-2})\right)
\left(I+\frac{M_1}{\lb}+\bigO(\lb^{-2})\right)
  \lb^{-(\bt_1+\bt_2)\si_3} e^{-\frac{1}{2}\lb\si_3},
\ee
where $M_1$ is given by (\ref{M1}) with $\al=\al_1+\al_2$, $\bt=\bt_1+\bt_2$,
and where, by (\ref{Rsol}),
\be\label{C1R}
H_1=-\frac{1}{2\pi i}\int_{\partial\mathcal U_0}\Delta_1(\mu)d\mu+
\bigO(|s|^2)+\bigO(|s|^{2+4(\al_1+\al_2)}).
\ee
Comparing (\ref{whPsi as}) and (\ref{whPsi as2}), we obtain
\be\label{sigma-med}
\widehat \Psi_{1,11}=
(H_1+M_1)_{11}=H_{1,11}+(\al_1+\al_2)^2-(\bt_1+\bt_2)^2.
\ee
Computing the residue of $\Delta_1(\mu)$ at zero, we obtain by
(\ref{C1R}),  (\ref{delta}), and (\ref{ME}) with $\al=\al_1+\al_2$, $\bt=\bt_1+\bt_2$,
\be\begin{aligned}
H_{1,11}=
\al_2 \frac{\bt_1+\bt_2}{\al_1+\al_2}s-
\frac{c_0}{\pi}\frac{\Ga(1+\al_1+\al_2+\bt_1+\bt_2)\Ga(1+\al_1+\al_2-\bt_1-\bt_2)}
{\Ga(1+2(\al_1+\al_2))^2}\\
\times
\frac{\Ga(1+2\al_1)\Ga(1+2\al_2)}
{\Ga(2+2(\al_1+\al_2))}e^{-2\pi i(\al_1+\al_2)}
|s|^{1+2(\al_1+\al_2)}+\bigO(|s|^2)+\bigO(|s|^{2+4(\al_1+\al_2)}).
\end{aligned}
\ee
Therefore, by (\ref{sigma-med}), (\ref{whPsiPsi}), (\ref{C1}), and (\ref{sigma4b}), we finally obtain the small $s$ expansion
(\ref{as0}) of $\si(s)$. In the case where $2(\alpha_1+\alpha_2)\in\mathbb N\cup\{0\}$, the estimate (\ref{as0half}) is obtained similarly by (\ref{asPN1half}).

Using the small $s$ asymptotic expansion (\ref{Rsol}) for $H$ and inverting the transformations (\ref{def R0}) and (\ref{whPsi}), one obtains small $s$ asymptotics
for $F_1$ and $F_2$ defined in (\ref{Psi i}), (\ref{Psi -i}), (\ref{Psi i 2}), and (\ref{Psi -i 2}). These lead to a proof of (\ref{c1c2}).

\section{Asymptotics of the orthogonal polynomials}\label{section: RHP Y}

We now use the model problem for $\Psi$ of Section \ref{secPsi} to obtain asymptotics for the solution
of the $Y$-RH problem for the orthogonal polynomials of Section \ref{secY}, for large $n$
uniformly in $0<t<t_0$ with a fixed sufficiently small $t_0$.
We assume that
$\Re\alpha_1,\Re\alpha_2>-1/2$, 
$\al_k\pm\bt_k\neq -1,-2,\dots$, $k=1,2$, 
and that $|||\beta|||<1$. For the small $s=2int$ asymptotics, we require furthermore that
$\Re(\al_1+\al_2)>-1/2$ and
$(\al_1+\al_2)\pm(\bt_1+\bt_2)\neq -1,-2,\dots$.

\subsection{Normalization of the RH problem}
    Set
    \begin{equation}\label{def T}
    T(z)=\begin{cases}Y(z)z^{-n\sigma_3},&\mbox{ for $|z|>1$},\\
    Y(z),&\mbox{ for $|z|<1$.}
    \end{cases}
    \end{equation}
    Then, by the RH conditions for $Y$, we obtain (recall that $z_1=e^{it}$, $z_2=e^{i(2\pi-t)}$):
    \subsubsection*{RH problem for $T$}
    \begin{itemize}
    \item[(a)] $T:\mathbb C \setminus C \to \mathbb C^{2\times 2}$ is analytic.
    \item[(b)] $T_+(z)=T_-(z)
                \begin{pmatrix}
                    z^n & f_t(z) \\
                    0 & z^{-n}
                \end{pmatrix},$
                \qquad  for $z \in C\setminus\{z_1,z_2\}$.
    \item[(c)] $T(z)=I+\bigO(1/z)$
                \qquad  as $z\to \infty $.
                    \item[(d)] As $z\to z_k$, $T(z)$ has the same singular behavior as $Y(z)$, see Section \ref{secY1}.
    \end{itemize}

\subsection{Opening of lens}
The function $f_t(z)$ admits the following factorization
on the unit circle:
\be\label{ffactoriz}
f_t(e^{i\th})=\mathcal{D}_+(e^{i\th})\mathcal{D}_-^{-1}(e^{i\th}),\qquad \th\neq \pm t,
\ee
where $\mathcal{D}_+$, $\mathcal{D}_-$ are the boundary values from the inside and the outside of the unit circle, respectively, of the Szeg\H o function $\mathcal{D}(z)=\exp {1\over 2\pi i}\int_C
{\ln f_t(s)\over s-z} ds$, which is analytic inside and outside of $C$. We have
(see (4.9)--(4.10) in \cite{DIK2}):
\be\label{Dl1}
\mathcal{D}(z)
=e^{\sum_0^\infty V_j z^j}\prod_{k=1}^2\left({z-z_k\over z_k e^{i\pi}}\right)^{\al_k+\bt_k}\equiv\mathcal D_{in,t}(z)
,\qquad
|z|<1,
\ee
and
\be\label{Dg1}
\mathcal{D}(z)^{-1}
=e^{\sum_{-\infty}^{-1}V_j z^j}\prod_{k=1}^2\left({z-z_k\over z}\right)^{\al_k-\bt_k}\equiv \mathcal D_{out,t}(z)^{-1}
,\qquad
|z|>1.
\ee
The branch of $(z-z_k)^{\al_k\pm\bt_k}$ in (\ref{Dl1}), (\ref{Dg1})
is fixed by the condition that $\arg(z-z_k)=2\pi$ on the line going from $z_k$ to the right parallel
to the real axis, and the branch cut is the line $\th=\th_k$ from $z=z_k=e^{i\th_k}$ to infinity.
In (\ref{Dg1}) for any $k$, the branch cut
of the root $z^{\al_k-\bt_k}$ is the line $\th=\th_k$ from
$z=0$ to infinity, and $\th_k<\arg z<2\pi+\th_k$.

By (\ref{ffactoriz}),
\begin{equation}\label{fDD}
f_t(e^{i\theta})=\mathcal D_{in,t}(e^{i\theta})\mathcal D_{out,t}(e^{i\theta})^{-1},
\end{equation}
and this function extends analytically to the
complex plane with two branch cuts $e^{it}\mathbb R^+$ and
$e^{-it}\mathbb R^+$. Orienting the cuts away from zero, we obtain for the jumps of $f_t$:
\begin{align*}
&f_{t+}=f_{t-}e^{2\pi i(\al_j-\bt_j)},& \mbox{ on $z_j(0,1)$},\\
&f_{t+}=f_{t-}e^{-2\pi i(\al_j+\bt_j)},& \mbox{ on $z_j(1,\infty)$}.
\end{align*}
 We can factorize
the jump matrix for $T$ if $|z|=1$, $t<\arg z<2\pi -t$ as follows:
\begin{equation}
\begin{pmatrix}
                    z^n & f_t(z) \\
                    0 & z^{-n}
                \end{pmatrix}= \begin{pmatrix}
                    1 & 0 \\
                    z^{-n}f_t(z)^{-1} & 1
                \end{pmatrix}\begin{pmatrix}
                    0 & f_t(z) \\
                    -f_t(z)^{-1} & 0
                \end{pmatrix}\begin{pmatrix}
                    1 & 0 \\
                    z^{n} f_t(z)^{-1} & 1
                \end{pmatrix}.
\end{equation}
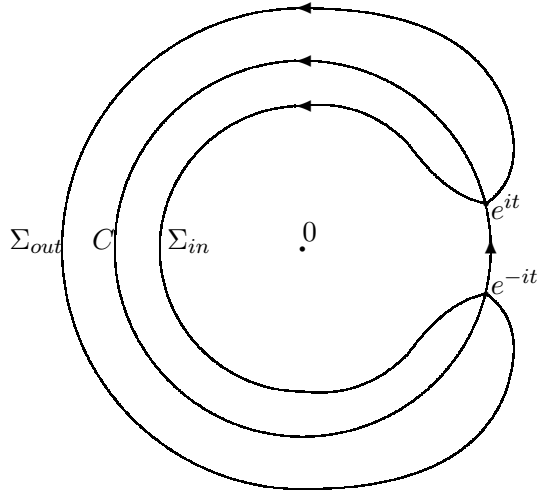
\begin{figure}[t]
    \begin{center}
    \setlength{\unitlength}{1truemm}
    \begin{picture}(100,55)(-5,10)
        \cCircle(50,40){25}[f]
        \put(50,40){\thicklines\circle*{.8}}
        \put(50,41){$0$}
        \put(75,41.5){\thicklines\vector(0,1){.0001}}
        \put(74.4,46){\thicklines\circle*{.8}}
        \put(75,44){$e^{it}$}
        \put(74.4,34){\thicklines\circle*{.8}}
        \put(75,34){$e^{-it}$}
        \put(49,65){\thicklines\vector(-1,0){.0001}}
        \put(49,72){\thicklines\vector(-1,0){.0001}}
        \put(49,59){\thicklines\vector(-1,0){.0001}}
        \put(22,40){$C$}
        \put(11,40){$\Sigma_{out}$}
        \put(32,40){$\Sigma_{in}$}

        \qbezier(74.4,46)(80,50)(77,59)
        \qbezier(77,59)(72,72)(50,72)
        \cCircle(50,40){32}[l]
\qbezier(74.4,34)(80,30)(77,21)
        \qbezier(77,21)(72,8)(50,8)

        \qbezier(74.4,46)(69,47)(64,54)
        \qbezier(64,54)(58,60)(50,59)
        \cCircle(50,40){19}[l]
        \qbezier(74.4,34)(69,33)(64,26)
        \qbezier(64,26)(58,20)(50,21)

%
    \end{picture}
    \caption{The contour $\Sigma_S=\Sigma_{in}\cup C\cup\Sigma_{out}$.}
    \label{fig2}
\end{center}
\end{figure}
We fix a lens-shaped region as shown in Figure \ref{fig2},
and define
\begin{equation}\label{def S}
S(z)=\begin{cases} T(z),&\mbox{outside the lens,}\\
T(z)\begin{pmatrix}
                    1 & 0 \\
                    z^{-n}f_t(z)^{-1} & 1
                \end{pmatrix},&\mbox{in the part of the lens outside the unit circle,}\\
T(z)\begin{pmatrix}
                    1 & 0 \\
                    -z^{n}f_t(z)^{-1} & 1
                \end{pmatrix},&\mbox{in the part of the lens inside the unit circle.}
\end{cases}
\end{equation}
The following RH conditions for $S$ can be verified directly.

\subsubsection*{RH problem for $S$}
\begin{itemize}
    \item[(a)] $S:\mathbb C \setminus \Sigma_S \to \mathbb C^{2\times 2}$ is analytic.
    \item[(b)] $S_+(z)=S_-(z)J_S(z),$
                \qquad  for $z \in \Sigma_S$, where $J_S$ is given
                by
\begin{equation}\label{jump S}
J_S(z)=\begin{cases}\begin{pmatrix}
                    1 & 0 \\
                    z^{-n}f_t(z)^{-1} & 1
                \end{pmatrix},&\mbox{on $\Sigma_{out}$,}\\
\begin{pmatrix}
                    0 & f_t(z) \\
                    - f_t(z)^{-1} & 0
                \end{pmatrix},&\mbox{ on the arc
                $(e^{it},e^{i(2\pi-t)})$,}\\\begin{pmatrix}
                    1 & 0 \\
                    z^{n}f_t(z)^{-1} & 1
                \end{pmatrix},&\mbox{ on $\Sigma_{in}$,}\\
                \begin{pmatrix}
                    z^n & f_t(z) \\
                    0 & z^{-n}
                \end{pmatrix},&\mbox{ on the arc $(e^{i(2\pi-t)},e^{it})$.}
\end{cases}
\end{equation}
    \item[(c)] $S(z)=I+\bigO(1/z)$
                \qquad  as $z\to \infty $.
    \item[(d)] As $z\to z_k$, $k=1,2$, and $z$ in the region outside the lens, we have
   \[S(z)=\begin{pmatrix}\bigO(1)&\bigO(1)+\bigO(|z-z_k|^{2\alpha_k})\\\bigO(1)&\bigO(1)+\bigO(|z-z_k|^{2\alpha_k})\end{pmatrix},\qquad \mbox{if}\quad \al_k\neq 0,
\]
and
\[S(z)=\begin{pmatrix}\bigO(1)&\bigO(|\ln|z-z_k||)\\\bigO(1)&\bigO(|\ln|z-z_k||)\end{pmatrix},\qquad \mbox{if}\quad \al_k=0.
\]
The behavior of $S(z)$, $z\to z_k$ in the other sectors is obtained from these expressions by application of the appropriate jump conditions.
    \end{itemize}
Let us now fix a small complex neighborhood $\mathcal U$ of $1$, for example a small disk such that for $t<t_0$, the singularities $e^{it}$ and $e^{-it}$
are contained in $\mathcal U$. Noting that $z^{n}$, resp.\ $z^{-n}$, is exponentially decaying as $n\to\infty$ for $|z|<1$, resp.\ $|z|>1$, one observes by (\ref{jump S}) that the jump matrix $J_S(z)$ converges to the identity matrix as
$n\to\infty$ for $z\in(\Sigma_{in}\cup\Sigma_{out})\setminus \mathcal U$, uniformly in $z$ and $t<t_0$.

\subsection{Global parametrix}\label{secN}
Define the function
\begin{equation}\label{def N}
N(z)=\begin{cases}\mathcal D_{in,t}(z)^{\sigma_3}\begin{pmatrix}0&1\\-1&0\end{pmatrix},&\mbox{
for $|z|<1$,}\\
\mathcal D_{out,t}(z)^{\sigma_3},&\mbox{ for $|z|>1$.}
\end{cases}
\end{equation}
It is straightforward to verify that $N$ satisfies the following RH conditions:
\subsubsection*{RH problem for $N$}
\begin{itemize}
    \item[(a)] $N:\mathbb C \setminus C \to \mathbb C^{2\times 2}$ is analytic.
    \item[(b)] $N_+(z)=N_-(z)
                \begin{pmatrix}
                    0 & f_t(z) \\
                    -f_t(z)^{-1} & 0
                \end{pmatrix},$
                \qquad  for $z \in C\setminus\{e^{\pm it}\}$.
    \item[(c)] $N(z)=I+\bigO(1/z)$
                \qquad  as $z\to \infty $.
    \end{itemize}

\subsection{$0<t\le \om(n)/n$. Local parametrix near $1$}\label{sec74}
Let $\om(x)$ be a positive, smooth for large $x$ function such that $\om(n)\to\infty$,
$\om(n)=o(n)$ as $n\to\infty$. 
In order to obtain the asymptotic solution to the $Y$-RH problem with ``good'' uniformity properties
in $0<t<t_0$, we will have to construct different local parametrices for the cases $0<t\le 1/n$,
$1/n<t\le\om(n)/n$, and $\om(n)/n<t<t_0$.
First, we consider $0<t\le 1/n$ and $1/n<t\le\om(n)/n$. For general $\bt$'s we have not excluded the possibility that
there is a finite set $\Omega$ of the points $s=-2int$ away from zero where the RH problem for $\Psi$ is not solvable. In order to integrate the differential
identity for $D_n(f)$ later on, we will need uniform asymptotics for the polynomials in a complex
neighborhood of the interval $0<t\le \om(n)/n$ away from $\Omega$. For simplicity of notation,
we consider only the case of real $t$, $0<t\le \om(n)/n$, in this section, assuming that this interval is disjoint from $\Omega$. The extension to a neighborhood of $0<t\le \om(n)/n$ with small
neighborhoods of the poles removed can be carried out easily by the reader.

For  $0<t\le 1/n$ and $1/n<t\le\om(n)/n$,  we will now construct a local parametrix in $\mathcal U$ which satisfies
the same jump conditions as $S$ inside $\mathcal U$, and which
matches with the global parametrix on the boundary of $\mathcal U$ for large $n$. More precisely, we will construct $P$ satisfying the following conditions.
\subsubsection*{RH problem for $P$}
\begin{itemize}
    \item[(a)] $P:{\mathcal U}\setminus \Sigma_S \to \mathbb C^{2\times 2}$ is analytic.
    \item[(b)] $P_+(z)=P_-(z)J_S(z),$
                \qquad  for $z \in \mathcal U\cap\Sigma_S$.
    \item[(c)] As $n\to\infty$, we have
    \begin{equation}\label{matching PN0}P(z)=\widetilde n^{-(\beta_1+\beta_2)\sigma_3}(I+o(1))\widetilde n^{(\beta_1+\beta_2)\sigma_3}N(z)
                \qquad  \mbox{ for $z\in\partial\mathcal U$,}
                \end{equation}
                uniformly for $0<t<t_0$, where \begin{equation}\label{wtn}\widetilde n=\min\{n,\sqrt{n/t}\}.\end{equation}
                \item[(d)] As $z\to z_k$, $k=1,2$, $S(z)P(z)^{-1}=\bigO(1)$.
    \end{itemize}

 \subsubsection{Modified model RH problem}
The RH problem for $\Psi$ was convenient to prove solvability, to
derive the Lax pair, and to obtain the asymptotics for $\Psi$ as $s\to
-i\infty$ and $s\to -i0_+$. In order to construct the local parametrix
near $1$ for the RH problem for the orthogonal polynomials,
we use an equivalent, but a more convenient form of the model RH
problem for $\Psi$. Set
\begin{equation}\label{def Phi}
\Phi(\zeta;s)\equiv\Phi(\zeta)=\begin{cases}
\Psi(\zeta),& -1<\Im\zeta<1,\\
\Psi(\zeta)e^{\pi i(\alpha_1-\beta_1)\sigma_3},& \Im\zeta>1,\\
\Psi(\zeta)e^{-\pi i(\alpha_2-\beta_2)\sigma_3},& \Im\zeta<-1.
\end{cases}
\end{equation}

\begin{figure}[t]
\begin{center}
    \setlength{\unitlength}{0.8truemm}
    \begin{picture}(100,65)(0,2.5)

    \put(50,45){\thicklines\circle*{.8}}
    \put(50,60){\thicklines\circle*{.8}}
    \put(50,30){\thicklines\circle*{.8}}
    \put(51,56){\small $i$}
    \put(51,32){\small $-i$}
    \put(50,60){\thicklines\circle*{.8}}
    \put(50,30){\thicklines\circle*{.8}}
    \put(35,60){\thicklines\vector(1,0){.0001}}
    \put(35,30){\thicklines\vector(1,0){.0001}}
    \put(69,60){\thicklines\vector(1,0){.0001}}
    \put(69,30){\thicklines\vector(1,0){.0001}}
    \put(50,60){\line(1,1){25}}
    \put(50,30){\line(1,-1){25}}
    \put(50,60){\line(-1,1){25}}
    \put(50,30){\line(-1,-1){25}}
    \put(50,30){\line(1,0){35}}
    \put(50,60){\line(1,0){35}}
    \put(50,30){\line(-1,0){35}}
    \put(50,30){\line(0,1){30}}
    \put(50,60){\line(-1,0){35}}
    \put(65,75){\thicklines\vector(1,1){.0001}}
    \put(65,15){\thicklines\vector(-1,1){.0001}}
    \put(50,39){\thicklines\vector(0,1){.0001}}
    \put(50,54){\thicklines\vector(0,1){.0001}}
    \put(35,75){\thicklines\vector(-1,1){.0001}}
    \put(35,15){\thicklines\vector(1,1){.0001}}
    \put(74,79){\small $\begin{pmatrix}1&1\\0&1\end{pmatrix}$}
    \put(8,79){\small $\begin{pmatrix}1&0\\-1&1\end{pmatrix}$}
\put(86,29){\small $e^{\pi i(\alpha_2+ \beta_2)\sigma_3}$}
\put(86,59){\small $e^{\pi i(\alpha_1+ \beta_1)\sigma_3}$}
\put(-6,29){\small $e^{\pi i(\alpha_2- \beta_2)\sigma_3}$}
\put(31,43){\small $\begin{pmatrix}0&1\\-1&1\end{pmatrix}$}
\put(-5,59){\small $e^{\pi i(\alpha_1- \beta_1)\sigma_3}$}
    \put(8,9){\small $\begin{pmatrix}1&0\\-1&1\end{pmatrix}$}
    \put(74,9){\small $\begin{pmatrix}1&1\\0&1\end{pmatrix}$}
    \end{picture}
    \caption{The jump contour $\Sigma$ and the jump matrices for $\Phi$.}
    \label{figure: Sigma}
\end{center}
\end{figure}
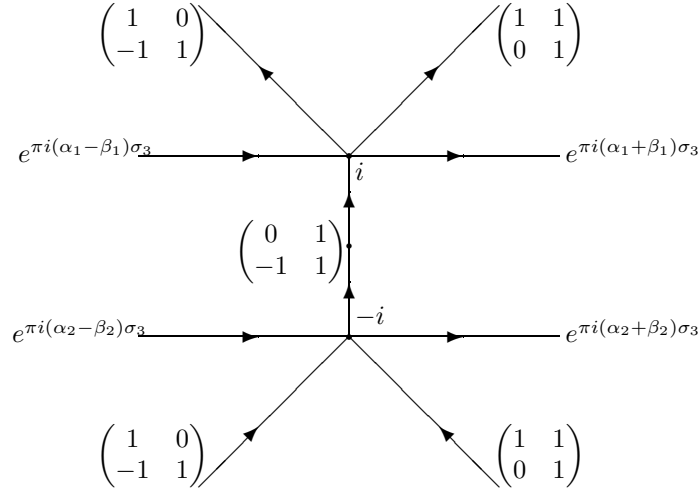
\subsubsection*{RH problem for $\Phi$}
\begin{itemize}
    \item[(a)] $\Phi:\mathbb C\setminus \Sigma \to \mathbb C^{2\times 2}$ is analytic, with
    \begin{align*}&\Sigma=\cup_{k=1}^9\Sigma_k,&& \Sigma_1=i+e^{\frac{i\pi}{4}}\mathbb R^+,
    &&&\Sigma_2=i+e^{\frac{3i\pi}{4}}\mathbb R^+,\\
&\Sigma_3=i-\mathbb R^+,&&\Sigma_4=-i-\mathbb R^+,
    &&&\Sigma_5=-i+e^{-\frac{3i\pi}{4}}\mathbb R^+,\\
    &\Sigma_6=-i+e^{\frac{i\pi}{4}}\mathbb R^+,&&\Sigma_7=-i+\mathbb R^+,&&&\Sigma_8=i+\mathbb
    R^+,\\
    &\Sigma_9=[-i,i].\end{align*}
    \item[(b)] The jump conditions are:
    \begin{equation}\label{jump Phi}\Phi_+(\zeta)=\Phi_-(\zeta)V_k,
                \qquad  \mbox{ for $\zeta \in \Sigma_k$},\end{equation} where
                \begin{align}
                &V_1=\begin{pmatrix}1&1\\0&1\end{pmatrix},
                &&V_2=\begin{pmatrix}1&0\\-1&1\end{pmatrix},\\
                                &V_3=e^{\pi i(\alpha_1-\beta_1)\sigma_3},&&
                V_4=e^{\pi i(\alpha_2-\beta_2)\sigma_3},\\
                &V_5=\begin{pmatrix}1&0\\-1&1\end{pmatrix},
                &&V_6=\begin{pmatrix}1&1\\0&1\end{pmatrix},\\
                &V_7=e^{\pi i(\alpha_2+\beta_2)\sigma_3},
                &&V_8=e^{\pi i(\alpha_1+\beta_1)\sigma_3},\\
                &V_9=\begin{pmatrix}0&1\\-1&1\end{pmatrix}.
                \end{align}
    \item[(c)] As $\zeta\to\infty$, we have
    \begin{equation}\label{Phi as}\Phi(\zeta)=
    (I+\frac{\Psi_1}{\zeta}+\frac{\Psi_2}{\zeta^2}+\bigO(\zeta^{-3}))\widehat P^{(\infty)}(\zeta)e^{-\frac{|s|}{4}\zeta\sigma_3},
    \end{equation}
    where
    \begin{equation}\label{whPinf}\widehat P^{(\infty)}(\zeta)=P^{(\infty)}(\zeta)
      \times \
        \begin{cases}
1,& -1<\Im\zeta<1\\
e^{\pi i(\alpha_1-\beta_1)\sigma_3},& \Im\zeta>1\\
e^{-\pi i(\alpha_2-\beta_2)\sigma_3},& \Im\zeta<-1
\end{cases}    ,
    \end{equation}
    with $P^{(\infty)}$ given by (\ref{Pinfty}).
    \end{itemize}
    $\Phi$ has  singular behavior near $\pm i$ which
    is inherited from $\Psi$. The precise conditions follow from
    (\ref{def Phi}), (\ref{Psi i}), (\ref{Psi -i}), (\ref{Psi i 2}, and (\ref{Psi -i 2}).

\subsubsection{$0<t\le \om(n)/n$. Construction of a local parametrix near $1$ in terms of $\Phi$}
We search for $P$ in the form
\begin{equation}\label{def P}
P(z)=E(z)\Phi(\frac{1}{t}\ln z; -2int)W(z),
\end{equation}
which means, in particular, that we evaluate $\Phi(\zeta;s)$, and thus $\Psi(\zeta;s)$, at $\zeta=\frac{1}{t}\ln z$ and $s=-i|s|=-2int$. The singularities $z=e^{\pm it}$ correspond to the values $\zeta=\pm i$.
In (\ref{def P}), $E$ has to be an analytic matrix-valued function in $\mathcal U$, and $W$ is
given by
\begin{equation}\label{def W}
W(z)=\begin{cases} -z^{\frac{n}{2}\sigma_3}
f_t(z)^{-\frac{\sigma_3}{2}}\sigma_3,&\mbox{ for $|z|<1$,}\\
z^{\frac{n}{2}\sigma_3} f_t(z)^{\frac{\sigma_3}{2}}\sigma_1,&\mbox{
for $|z|>1$.}\end{cases}\qquad \si_1=\begin{pmatrix}0&1\cr 1&0\end{pmatrix}.
\end{equation}
Choose $\Sigma_S$ such that $\frac{1}{t}\ln(\Sigma_S)\subset \Sigma\cup
i\mathbb R$ in $\mathcal U$. Then one verifies, using the jump conditions (\ref{jump Phi}) for $\Phi$ and the jump matrices (\ref{jump S}) for $S$, that $P$ satisfies the jump
relation $P_+=P_-J_S$ on $\mathcal U\cap \Sigma_S$, and that $P$ is
 meromorphic in $\mathcal U\setminus \Sigma_S$ with possible singularities at $z_1,z_2$. By the condition (d) of the RH problem for $S$, (\ref{def Phi}), and the condition (d) of the RH problem for $\Psi$, the singularities of $S(z)P(z)^{-1}$ at $z_1$ and $z_2$ are removable, so the condition (d) of the RH problem for $P(z)$ is satisfied.

It remains to choose $E$ in such a way that the matching condition (\ref{matching PN0}) for $z\in\partial\mathcal U$ as $n\to\infty$ holds, uniformly for $0<t<t_0$.
Define \begin{equation}\label{def E}
E(z)=\sigma_1\left(\mathcal D_{in,t}(z)\mathcal D_{out,t}(z)\right)^{-\frac{1}{2}\sigma_3}\widehat P^{(\infty)}(\frac{1}{t}\ln z)^{-1}.
\end{equation}
It is straightforward to verify that $E$ is analytic in $\mathcal U$ (in particular,  the branch cuts for $\mathcal D_{in,t}$ and $\mathcal D_{out, t}$ cancel out with those of $\widehat P^{(\infty)}$).

\begin{proposition}\label{PropEst}Let $\widetilde n=\min\{n,\sqrt{n/t}\}$, and assume that $|||\beta|||<1$.
We have
\begin{equation}\label{matching P}
P(z)N(z)^{-1}=\widetilde n^{-(\beta_1+\beta_2)\sigma_3}(I+\bigO(t(nt)^{-1+|||\bt|||}))\widetilde n^{(\beta_1+\beta_2)\sigma_3}, \quad \mbox{$n\to\infty$, $z\in\partial\mathcal U$,}
\end{equation}
uniformly for $0<t<t_0$ with $t_0$ sufficiently small and uniformly in $z\in\partial\mathcal U$.
\end{proposition}

\begin{proof}
Let us first consider the case where $c_0\leq nt\leq C_0$, with some $c_0>0$ small and some 
$C_0>0$ large. The constants $c_0$, $C_0$ will be fixed below.
Then $|\frac{1}{t}\ln z|> \delta n$ for $z\in\partial \mathcal U$, so $s=-2int$ remains bounded
and bounded away from zero, and by
(\ref{def P}) and (\ref{Phi as}), we have
\be\label{matchingPN}
P(z)N(z)^{-1}=E(z)(I+\bigO(n^{-1}))\widehat P^{(\infty)}(\frac{1}{t}\ln z) z^{-\frac{n}{2}\sigma_3}W(z)N(z)^{-1},\quad
z\in\partial\mathcal U,\quad n\to\infty.
\ee
Recall that we assume that the problem for $\Psi$ is solvable for $c_0\le nt\le C_0$.
Therefore, by general properties of Painlev\'e RH problems, the estimate for the error term here is valid uniformly for all $c_0\le nt\le C_0$.
By (\ref{def N}) and (\ref{def W}), we obtain
\be\label{WN}
z^{-\frac{n}{2}\sigma_3}W(z)N(z)^{-1}=
\left(\mathcal D_{in,t}(z)\mathcal D_{out,t}(z)\right)^{\frac{1}{2}\sigma_3}\sigma_1,\qquad
\mbox{ for $z\in\partial\mathcal U$.}
\ee
Substituting this into (\ref{matchingPN}), we see the reason for the definition of $E$ in (\ref{def E}). Furthermore, we set
\begin{equation}\label{whE}
\widehat E(z)=n^{(\beta_1+\beta_2)\sigma_3}E(z)
\end{equation}
so that $\widehat E$ is bounded in $n$ (uniformly for $z\in\partial\mathcal U$ and uniformly for $t<t_0$): this follows easily from (\ref{whPinf}) and (\ref{Pinfty}).
In particular,
\be\label{whE1}
\widehat E(z)=\sigma_1\left(\mathcal D_{in,t}(z)\mathcal D_{out,t}(z)\right)^{-\frac{1}{2}\sigma_3}
(\ln z-it)^{\beta_1\sigma_3}(\ln z+it)^{\beta_2\sigma_3},\qquad -1<\zeta<1.
\ee
We can now write (\ref{matchingPN}) as follows:
\begin{eqnarray}
P(z)N(z)^{-1}&=&n^{-(\beta_1+\beta_2)\sigma_3}\widehat E(z)(I+\bigO(n^{-1}))\widehat E(z)^{-1}n^{(\beta_1+\beta_2)\sigma_3}\nonumber\\
\label{matchingPN3}&=&\widetilde n^{-(\beta_1+\beta_2)\sigma_3}(I+\bigO(n^{-1}))\widetilde n^{(\beta_1+\beta_2)\sigma_3},\quad z\in\partial\mathcal U,\quad n\to\infty.
\end{eqnarray}
This proves that (\ref{matching P}) holds (uniformly) for $c_0\leq nt\leq C_0$, 
$z\in\partial\mathcal U$.

Next, suppose $C_0<nt\le\om(n)$. In this case we cannot use the expansion (\ref{Phi as}) since 
the argument $s$ of $\Psi_1(s)$ is not bounded.
Instead we use the large $|s|=2nt$ asymptotics for $\Psi$. For that, we need $C_0$ to be sufficiently large.
The asymptotics will be valid for the whole region $C_0<nt< t_0$.
Note that, for $z\in\partial \mathcal U$ and $t$ sufficiently small (i.e., $t<t_0$),  $\zeta=\frac{1}{t}\ln z$ is sufficiently large in absolute value to lie outside of the regions $\mathcal U_1$ and $\mathcal U_2$ defined in Section \ref{section: local infty}.
By (\ref{def P}), (\ref{def U}), (\ref{def tilde R}), and (\ref{def R}), we have
\begin{eqnarray}
P(z)N(z)^{-1}&=&
E(z)\Phi(\frac{1}{t}\ln z;-2int)W(z)N(z)^{-1}\\
&=&E(z)(nt)^{-\frac{1}{2}(\beta_1+\beta_2)\sigma_3}R(\frac{1}{t}\ln z;-2int)(nt)^{\frac{1}{2}(\beta_1+\beta_2)\sigma_3}\nonumber\\ &&\qquad \times \widehat P^{(\infty)}(\frac{1}{t}\ln z)
z^{-\frac{n}{2}\sigma_3}W(z)N(z)^{-1},\quad z\in\partial\mathcal U.
\end{eqnarray}
Using (\ref{whE}), and (\ref{WN}), we obtain for $z\in\partial\mathcal U$:
\begin{align}
P(z)N(z)^{-1}&=
n^{-(\beta_1+\beta_2)\sigma_3}\widehat E(z)(nt)^{-\frac{1}{2}(\beta_1+\beta_2)\sigma_3}R(\frac{1}{t}\ln z;-2int)(nt)^{\frac{1}{2}(\beta_1+\beta_2)\sigma_3}\nonumber\\
&\qquad\qquad\qquad \times\
\widehat E(z)^{-1}n^{(\beta_1+\beta_2)\sigma_3}\nonumber \\
\label{PNs}
&=\left(\frac{n}{t}\right)^{-\frac{1}{2}(\beta_1+\beta_2)\sigma_3}\widehat E(z)R(\frac{1}{t}\ln z;-2int)
\widehat E(z)^{-1}\left(\frac{n}{t}\right)^{\frac{1}{2}(\beta_1+\beta_2)\sigma_3}.
\end{align}
Therefore, by (\ref{as R-1}), we have (uniformly for $C_0/n<t<t_0$, $z\in \partial\mathcal U$)
\be\label{PNmed}
P(z)N(z)^{-1}=\widetilde n^{-(\beta_1+\beta_2)\sigma_3}(I+\bigO(t(nt)^{-1+|||\bt|||}))
\widetilde n^{(\beta_1+\beta_2)\sigma_3}.\quad n\to\infty,\quad z\in\partial\mathcal U.
\ee

If $nt<c_0$,  we can use the small $|s|$ asymptotics for $\Psi(\zeta;s)$ for large values of $\zeta=\frac{1}{t}\ln z$. We need to consider this case separately from $c_0\le nt \le C_0$ since $s=0$ is a branching point for the Painlev\'e functions. 
By (\ref{whPsi}), (\ref{def R0}), (\ref{whPsi as2}), and (\ref{RHP Ph: c}), we have for $z\in\partial\mathcal U$ and $\arg z\notin (t,2\pi-t)$,
\begin{equation}\Psi(\frac{1}{t}\ln z;-2int)=\left(I+\bigO(n^{-1})\right)\widehat P^{(\infty)}(\frac{1}{t}\ln z)z^{-\frac{n}{2}\sigma_3}.\end{equation}
This implies that
\begin{equation}
P(z)N(z)^{-1}=\widetilde n^{-(\beta_1+\beta_2)\sigma_3}(I+\bigO(n^{-1}))\widetilde n^{(\beta_1+\beta_2)\sigma_3},\quad
\mbox{ as $n\to\infty$.}
\end{equation}
For $2\pi-t<\arg z<2\pi$ and $0<\arg z<t$, the same estimate can be proved similarly.
\end{proof}

 For later use, we note that 
\begin{align}
&\label{Eh}\widehat E^{-1}(z)\widehat E'(z)=h(z)\sigma_3,\\
& h(z)=-\frac{1}{2}\sum_{j=1}^{+\infty}jV_jz^{j-1}+\frac{1}{2}\sum_{j=-1}^{-\infty}jV_jz^{j-1}-\frac{\beta_1}{z-e^{it}}+\frac{\beta_1}{z\ln z - itz}
\nonumber \\
&\label{h}\qquad\qquad\qquad\qquad\qquad -\frac{\beta_2}{z-e^{-it}}+\frac{\beta_2}{z\ln z+itz}-\frac{\alpha_1-\beta_1+\alpha_2-\beta_2}{2z}.
\end{align}

\subsubsection{$0<t\le\om(n)/n$. Final transformation}\label{secUps1}
Let the boundary $\partial \cal U$ be oriented clockwise.
Define
\begin{equation}\label{def Ups}
\Upsilon(z)=\begin{cases}\widetilde n^{(\beta_1+\beta_2)\sigma_3}S(z)N(z)^{-1}
\widetilde n^{-(\beta_1+\beta_2)\sigma_3},&
z\in\mathbb C\setminus (\mathcal
U\cup\Sigma_S),\\
\widetilde n^{(\beta_1+\beta_2)\sigma_3}S(z)P(z)^{-1}\widetilde n^{-(\beta_1+\beta_2)\sigma_3},&
z\in \mathcal U\setminus\Sigma_S,
\end{cases}
\end{equation}
where $\widetilde n$ is given in the Proposition \ref{PropEst}.

Then, from the RH conditions (b) and (d) for $S$, and from the conditions (b) and (d) for $P$, it follows that $\Upsilon$ is analytic inside $\mathcal U$.
Similarly, on $C$, the jumps for $S$ and $N$ are the same, so $\Upsilon$ is analytic on $C$.
 On $(\Sigma_{in}\cup\Sigma_{out})\setminus\overline{\mathcal U}$, we have an exponentially small jump as $n\to\infty$. Indeed, on these contours
\[
\Upsilon_+(z)=\Upsilon_-(z)\widetilde n^{(\beta_1+\beta_2)\sigma_3}N(z)J_S(z)N^{-1}(z)\widetilde n^{-(\beta_1+\beta_2)\sigma_3}.
\]
Therefore, for $z\in \Sigma_{out}\setminus\overline{\mathcal U}$,
\be
\Upsilon_+(z)=\Upsilon_-(z)\begin{pmatrix}
1& 0\cr
z^{-n}/(\mathcal D_{in,t}(z)\mathcal D_{out,t}(z)\wt n^{2(\bt_1+\bt_2)})& 1
\end{pmatrix}=\Upsilon_-(z)(I+\bigO(e^{-\ep n})),
\ee
and, for $z\in \Sigma_{in}\setminus\overline{\mathcal U}$,
\be
\Upsilon_+(z)=\Upsilon_-(z)\begin{pmatrix}
1& -z^{n}\mathcal D_{in,t}(z)\mathcal D_{out,t}(z)\wt n^{2(\bt_1+\bt_2)}\cr
0& 1
\end{pmatrix}=\Upsilon_-(z)(I+\bigO(e^{-\ep n})),
\ee
with some $\ep>0$,
uniformly for $0<t\le t_0$.

For $z\in\partial\mathcal U$, $\Upsilon$ has, as $n\to\infty$, a uniformly in $t$ and $z$ small 
jump by Proposition \ref{PropEst}:
\[\Upsilon_+(z)=\Upsilon_-(z)\widetilde n^{(\beta_1+\beta_2)\sigma_3}P(z)N^{-1}(z)\widetilde n^{-(\beta_1+\beta_2)\sigma_3}=\Upsilon_-(z)(I+\bigO(t(nt)^{-1+|||\beta|||})).\]
By the standard theory for normalized RH problems with small jumps, it follows that
\begin{equation}\Upsilon(z)=I+\bigO(t(nt)^{-1+|||\beta|||}),\qquad \frac{d\Upsilon(z)}{dz}=\bigO(t(nt)^{-1+|||\beta|||}),\label{as R}\end{equation} as
$n\to\infty$, uniformly for $z$ off the jump contour for $\Upsilon$, and uniformly for $0<t<t_0$.

\subsection{$\om(n)/n<t<t_0$. Local parametrices}\label{sec75}

The parametrix of the previous section is valid for the whole region $0<t<t_0$. However, the structure of
the large $n$ expansion for $Y$ which follows from it is too cumbersome for the detailed analysis in
the next section. Therefore, we will now construct a more explicit solution for the case
$\om(n)/n<t<t_0$.
In this case $\zeta={1\over t}\ln z$ is not necessarily large on $\partial\mathcal U$. However,
$|s|=2nt$ is large, and we will construct the large $s$ asymptotic expansion for $Y$. The construction here is very similar to that of Section \ref{sec large s}.

First, we need to modify the $S$-RH problem. Namely, in addition to the lens around the arc $(t,2\pi-t)$,
we now open the lens around the complementary arc in the same way as well. Thus $z_1$, $z_2$ are the end-points of the lenses. The jump conditions on the new lens for $S$ are easily written down. We obtain the same $S$-RH problem as considered for the case of 2 separate FH singularities (see \cite{DIK2}).
We now surround the points $z_1=e^{it}$, $z_2=e^{i(2\pi-t)}$ by small neighborhoods $\wt{\mathcal U_1}$,
$\wt{\mathcal U_2}$, resp.,
which are the images of the
neighborhoods $\mathcal U_1$, $\mathcal U_2$ of the points $\pm i$ from section
\ref{sec large s} under the inverse of the mapping $\zeta={1\over t}\ln z$.
The neighborhoods $\mathcal U_j$ are fixed in the $\zeta$ plane. Thus, in the $z$-plane,
$\wt{\mathcal U_1}$, $\wt{\mathcal U_2}$ contract if $t$ decreases with $n\to\infty$.

For a parametrix outside these neighborhoods, we take $N(z)$ of section \ref{secN} as before.
We now construct parametrices in $\wt{\mathcal U_j}$ matching with $N(z)$ to leading order at the
boundaries.

We look for a parametrix in $\wt{\mathcal U_1}$ in the form:
\be\label{def P1}
\wt P_1(z)=\wt E_1(z)M^{(\al_1,\bt_1)}\left(nt(\zeta(z)-i)\right)\Omega_1(z) W(z),\qquad
\zeta={1\over t}\ln z,
\ee
where
\be
\Omega_1(z)=\begin{cases}
e^{i{\pi\over 2}(\al_1-\bt_1)\si_3},& \Im\zeta>1\cr
e^{-i{\pi\over 2}(\al_1-\bt_1)\si_3},& \Im\zeta<1
\end{cases},
\ee
$W(z)$ is given by (\ref{def W}), and $M^{(\al_1,\bt_1)}(\lb)$ is the solution to the RH-problem of Section
\ref{secaux} with $\al=\al_1$, $\bt=\bt_1$. The matrix $\wt E_1(z)$ is analytic in $\wt{\mathcal U_1}$
and will now be determined from the matching condition.
We now use the large argument expansion (\ref{RHP Ph: c}) for
$M^{(\al_1,\bt_1)}(nt(\zeta(z)-i))$ for $z\in\partial \wt{\mathcal U_1}$. Recalling also (\ref{WN}), we can write for $z\in\partial\wt U_1$
\begin{multline}
\wt P_1(z)N(z)^{-1}=
\wt E_1(z)\left(I+\frac{M_1^{(\al_1,\bt_1)}}{nt(\zeta-i)}+\bigO((nt)^{-2})\right)\\
\times (nt(\zeta-i))^{-\bt_1\si_3}
e^{{i\over 2}nt\si_3}\Omega_1(\mathcal D_{in,t}\mathcal D_{out,t})^{\si_3/2}\si_1.
\end{multline}
We now let
\be
\wt E_1(z)=\si_1(\mathcal D_{in,t}\mathcal D_{out,t})^{-\si_3/2}\Omega^{-1}_1
(nt(\zeta-i))^{\bt_1\si_3}e^{-{i\over 2}nt\si_3}.
\ee
It is easy to check that $\wt E_1(z)$ is analytic in $\wt{\mathcal U_1}$.
Furthermore, since
\[
\mathcal D_t(z)\equiv
 \mathcal D_{in,t}\mathcal D_{out,t}(z-z_1)^{-2\bt_1}(z-z_2)^{-2\bt_2}
 \]
is uniformly bounded in $\wt{\mathcal U_j}$, $j=1,2$, we can write
\be
\wt E_1(z)=t^{\bt_2\si_3}n^{-\bt_1\si_3}\wh E_1(z),
\ee
where
\be\label{hatE_1}
\wh E_1(z)=
\si_1\mathcal D_t(z)^{-\si_3/2}\left(\frac{t}{z-z_2}\right)^{\bt_2\si_3}
\left(\frac{\ln z-it}{z-z_1}\right)^{\bt_1\si_3}\Omega^{-1}_1e^{-{i\over 2}nt\si_3}
\ee
is bounded in $n$ uniformly for $\om(n)/n<t<t_0$, $z\in\wt{\mathcal U_1}$.
Therefore,
\be\label{jumpP_1N}
\wt P_1(z)N(z)^{-1}=
t^{\bt_2\si_3}n^{-\bt_1\si_3}\left(I+\frac{\wh E_1(z)M_1^{(\al_1,\bt_1)}\wh E_1(z)^{-1}}{nt(\zeta-i)}+\bigO((nt)^{-2})
\right)n^{\bt_1\si_3}t^{-\bt_2\si_3}
\ee
uniformly for $\om(n)/n<t<t_0$, $z\in\partial\wt{\mathcal U_1}$.

Furthermore, one easily verifies that $\wt P_1$ has the same jumps as $S$ in
$\wt{\mathcal U_1}$, and $\wt P_1 S^{-1}$ is bounded at $z_1$. Thus,
$\wt P_1$ gives a parametrix for $S$ in $\wt{\mathcal U_1}$ with the matching condition
(\ref{jumpP_1N}) with $N(z)$ at the boundary of that region.

Similarly, we obtain that the following function gives a parametrix for $S$ in $\wt{\mathcal U_2}$:
\be\label{def P2}
\wt P_2(z)=\wt E_2(z)M^{(\al_2,\bt_2)}\left(nt(\zeta(z)+i)\right)\Omega_2(z) W(z),\qquad
\zeta={1\over t}\ln z,
\ee
where
\be
\Omega_2(z)=\begin{cases}
e^{i{\pi\over 2}(\al_2-\bt_2)\si_3},& \Im\zeta>-1\cr
e^{-i{\pi\over 2}(\al_2-\bt_2)\si_3},& \Im\zeta<-1
\end{cases},
\ee
with
the prefactor
\be
\wt E_2(z)=t^{\bt_1\si_3}n^{-\bt_2\si_3}\wh E_2(z),
\ee
where
\be\label{hatE_2}
\wh E_2(z)=\si_1\mathcal D_t(z)^{-\si_3/2}\left(\frac{t}{z-z_1}\right)^{\bt_1\si_3}
\left(\frac{\ln z+it}{z-z_2}\right)^{\bt_2\si_3}\Omega^{-1}_2e^{{i\over 2}nt\si_3}
\ee
is analytic in $\wt{\mathcal U_2}$ and bounded in $n$ uniformly for $\om(n)/n<t<t_0$, $z\in\wt{\mathcal U_2}$.

The matching condition with $N(z)$ is
\be\label{jumpP_2N}
\wt P_2(z)N(z)^{-1}=
t^{\bt_1\si_3}n^{-\bt_2\si_3}\left(I+\frac{\wh E_2(z)M_1^{(\al_2,\bt_2)}
\wh E_2(z)^{-1}}{nt(\zeta+i)}+\bigO((nt)^{-2})
\right)n^{\bt_2\si_3}t^{-\bt_1\si_3}
\ee
uniformly for $\om(n)/n<t<t_0$, $z\in\partial\wt{\mathcal U_2}$.

\subsubsection{$\om(n)/n<t<t_0$. Final transformation}\label{sec 751}
Let the boundaries $\partial \mathcal U_j$, $j=1,2$ be oriented clockwise.
Set
\begin{equation}\label{def tilde Ups2}
\widetilde \Upsilon(z)=\begin{cases}S(z)\wt P_1(z)^{-1},& z\in \wt{\mathcal U_1},\\
S(z)\wt P_2(z)^{-1},&z\in \wt{\mathcal U_2},\\
S(z)N(z)^{-1},& z\in \mathbb C
\setminus\left(\overline{\wt{\mathcal U_1}}\cup \overline{\wt{\mathcal U_2}}\right).\end{cases}
\end{equation}
Next, define $\Upsilon$ by
\begin{equation}\label{def Ups2}
\Upsilon(z)=\left(\frac{n}{t}\right)^{\frac{\beta_1+\beta_2}{2}\sigma_3}\widetilde\Upsilon(z)
\left(\frac{n}{t}\right)^{-\frac{\beta_1+\beta_2}{2}\sigma_3}.
\end{equation}

Then $\Upsilon$ is analytic and, in particular, has no jumps
inside $\wt{\mathcal U_j}$, $j=1,2$, and on $C$.
Exactly as in Section \ref{secUps1}, we see that the jumps of $\Upsilon$ on the rest of the lenses
are identity plus an exponentially small in $nt$ addition
($I+\bigO(e^{-\ep nt})$, $\ep>0$), uniformly in $\om(n)/n<t<t_0$.
Furthermore, using (\ref{jumpP_1N}), (\ref{jumpP_2N}), we obtain:
\begin{multline}
\Upsilon_+(z)=\Upsilon_-(z)J_1(z), \qquad z\in\partial\wt{\mathcal U_1},\label{J1}\\
J_1(z)=(nt)^{-\frac{\bt_1-\bt_2}{2}\si_3}
\left(I+\frac{\wh E_1(z)M_1^{(\al_1,\bt_1)}\wh E_1(z)^{-1}}{nt(\zeta-i)}+\bigO((nt)^{-2})\right)
(nt)^{\frac{\bt_1-\bt_2}{2}\si_3},
\end{multline}
and
\begin{multline}
\Upsilon_+(z)=\Upsilon_-(z)J_2(z), \qquad z\in\partial\wt{\mathcal U_2},\label{J2}\\
J_2(z)=(nt)^{\frac{\bt_1-\bt_2}{2}\si_3}
\left(I+\frac{\wh E_2(z)M_1^{(\al_2,\bt_2)}\wh E_2(z)^{-1}}{nt(\zeta+i)}+\bigO((nt)^{-2})\right)
(nt)^{-\frac{\bt_1-\bt_2}{2}\si_3},
\end{multline}
uniformly for $\om(n)/n<t<t_0$, and for $z\in\partial\wt{\mathcal U_j}$, $j=1,2$, resp.
We also have that $\Upsilon(\infty)=I$.

Again, by the standard theory, but now for RH problems on contracting contours, it follows that
\begin{equation}\Upsilon(z)=I+\bigO((nt)^{-1+|||\beta|||}),\qquad 
\frac{d\Upsilon(z)}{dz}=\bigO(t^{-1}(nt)^{-1+|||\beta|||}),\label{as R2}\end{equation} as
$n\to\infty$, uniformly for $z$ off the jump contour for $\Upsilon$, and uniformly for $\om(n)/n<t<t_0$.

These estimates would be sufficient to obtain the asymptotics for the Toeplitz determinants $D_n(f_t)$ if $|||\beta|||<1/2$, but to extend the results to the full range of $|||\bt|||<1$,
we need a modification of the $\Upsilon$-RH problem similar to the modification of the $R$-RH problem in Section \ref{sec large s}.  We assume for definiteness that $\Re\bt_1>\Re\bt_2$.
Then the jump matrices (\ref{J1}) and (\ref{J2}) behave for large $n$ as
\be
J_1(z)=I+\ell_1(z)\si_-+\bigO((nt)^{-1}),\qquad
J_2(z)=I+\ell_2(z)\si_++\bigO((nt)^{-1}),
\ee
where
\be\label{ell}
\ell_1(z)=\frac{(nt)^{-1+\bt_1-\bt_2}}{{1\over t}\ln z -i}(\wh E_1 M_1^{(\al_1,\bt_1)}\wh E^{-1}_1)_{21},\qquad
\ell_2(z)=\frac{(nt)^{-1+\bt_1-\bt_2}}{{1\over t}\ln z +i}(\wh E_2 M_1^{(\al_2,\bt_2)}\wh E^{-1}_2)_{12}.
\ee
Let (cf. (\ref{XR}))
\be\label{UZ}
\Upsilon(z)=\wh\Upsilon(z)Z(z),
\ee
 where $Z(z)$ is the solution of the normalized at infinity RH problem with jumps $I+\ell_1(z)\si_-$ on
 $\partial\wt{\mathcal U_1}$ and $I+\ell_2(z)\si_+$ on $\partial\wt{\mathcal U_2}$ oriented clockwise.
As in section \ref{sec large s} for $\wh R$, we conclude that
\be\label{whUpsest}
\wh \Upsilon(z)=I+\bigO((nt)^{-1}),\qquad \wh \Upsilon'(z)=\bigO(t^{-1}(nt)^{-1}),
\ee
uniformly for $\om(n)/n<t<t_0$ and $z$ off the contour for $\Upsilon(z)$.

Just as the X-RH problem of Section \ref{sec large s}, the $Z$-RH problem is solved explicitly, and
in particular, we obtain for $z\in \wt{\mathcal U_1}$:
\be\label{Z1}
Z(z)=\left(I+\frac{1}{z-e^{it}}
\begin{pmatrix}a&0\cr b&0\end{pmatrix}+
\frac{1}{z-e^{-it}}
\begin{pmatrix}0&c\cr 0&d\end{pmatrix}\right)
\left(I-\ell_1(z)\si_-\right),
\ee
where
\be\label{abcd}
b=\wh\ell_1(e^{it})\de,\qquad c=\wh\ell_2(e^{-it})\de,\qquad
a=\frac{\wh\ell_1(e^{it})\wh\ell_2(e^{-it})}{2i\sin t}\de,\qquad d=-a,
\ee
with the notation
\[
\wh\ell_1(z)=\ell_1(z)(z-e^{it}),\qquad \wh\ell_2(z)=\ell_2(z)(z-e^{-it}),\qquad
\de=\left(1-\frac{\wh\ell_1(e^{it})\wh\ell_2(e^{-it})}{4\sin^2 t}\right)^{-1}.
\]
We can expand (\ref{Z1}) further as follows:
\be\label{Z2}
Z(z)=I-\tau(z)
\begin{pmatrix}c&0\cr d&0\end{pmatrix}+
\frac{1}{z-e^{-it}}
\begin{pmatrix}0&c\cr 0&d\end{pmatrix}
-\pi(z)\si_-,\qquad z\in \wt{\mathcal U_1},
\ee
where
\[
\pi(z)=\frac{1}{z-e^{it}}(\wh\ell_1(z)-\wh\ell_1(e^{it})),\qquad
\tau(z)=\frac{1}{z-e^{it}}\left(\frac{\wh\ell_1(z)}{z-e^{-it}}-\frac{\wh\ell_1(e^{it})}{2i\sin t}\right).
\]
Note that
\[
\pi(e^{it})={d\over dz}\wh\ell_1(z)_{z=e^{it}},\qquad
\pi'(e^{it})={1\over 2}{d^2\over dz^2}\wh\ell_1(z)_{z=e^{it}},
\]
and similarly, $\tau(z)$ and its derivatives at $e^{it}$ can be expressed in terms of $\wh\ell_1(z)$.

Using (\ref{ell}) we then obtain:
\begin{multline}
\wh\ell_j(z_j)=\bigO(t(nt)^{-1+|||\bt|||}),\qquad \pi(e^{it})=\bigO(t(nt)^{-1+|||\bt|||}),\\
\pi'(e^{it})=\bigO(t(nt)^{-1+|||\bt|||}),\qquad \tau(e^{it})=\bigO(t^{-1}(nt)^{-1+|||\bt|||}),\\
\tau'(e^{it})=\bigO(t^{-2}(nt)^{-1+|||\bt|||}),\qquad \de=1+\bigO((nt)^{-2+2|||\bt|||}),\\
c=\bigO(t(nt)^{-1+|||\bt|||}),\qquad d=\bigO(t(nt)^{-2+2|||\bt|||}),\qquad |||\bt|||<1,
\end{multline}
as $n\to\infty$ uniformly for $\om(n)/n<t<t_0$.
Therefore, using (\ref{Z2}), we can write
\be\label{Z est}
Z(e^{it})=I+\bigO((nt)^{-1+|||\bt|||})
\ee
and
\begin{multline}\label{ZZ est}
(nt)^{\frac{\beta_1-\beta_2}{2}\sigma_3}(Z^{-1}{dZ\over dz})(e^{it})
(nt)^{-\frac{\beta_1-\beta_2}{2}\sigma_3}\\
=\frac{(nt)^{\bt_1-\bt_2}c}{4\sin^2 t}\si_+ +\bigO(n^{-1+|||\bt|||})+\bigO(t^{-1}(nt)^{-2+2|||\bt|||}),
\end{multline}
uniformly for  $\om(n)/n<t<t_0$. Here we explicitly wrote the $12$ matrix element
as we will need to analyse it in more detail below.

\section{Asymptotics for $D_n(f_t)$}\label{section: as Toep}

\subsection{Asymptotic form of the differential identity. Proof of Theorem \ref{theorem: extension}}

In the previous section, we performed a series of transformations
$Y\mapsto T\mapsto S\mapsto \Upsilon$ for $0<t\le\om(n)/n$, where $\Upsilon$ can be expressed explicitly in terms of $S$ (and hence, of $Y$) and in terms of the local parametrix $P$ as in (\ref{def Ups}), for $z$ near $z_1,z_2$. Since the asymptotics for $\Upsilon$ are known, see (\ref{as R}), we can obtain the asymptotics for the right hand side in the differential identity (\ref{differentialidentity}) in terms of the local parametrix $P$, for $0<t\le\om(n)/n$.  

For $\om(n)/n<t<t_0$, we have the series of transformations
$Y\mapsto T\mapsto S \mapsto\widetilde\Upsilon\mapsto \Upsilon$, where the asymptotics are known for $\Upsilon$, and where $\Upsilon$ can be expressed explicitly in terms of $Y$ and the local parametrices $\widetilde P_1, \widetilde P_2$. Thus, we can obtain the asymptotics for the right hand side of 
(\ref{differentialidentity}) in this case as well.

Combining the asymptotic behavior for ${1\over i}\frac{d}{dt}\ln D_n(f_t)$ in those two cases, 
we obtain the following result.

\begin{proposition}\label{prop diff id F}
Suppose that $\Re\alpha_1,\Re\alpha_2,\Re(\alpha_1+\alpha_2)> -{1\over 2}$, 
$\al_k\pm\bt_k\neq -1,-2,\dots$, $k=1,2$, $(\al_1+\al_2)\pm(\bt_1+\bt_2)\neq -1,-2,\dots$,
and that $|||\beta|||<1$. Let $\sigma(s)$ be the solution to (\ref{sigmaPV}) analyzed above, 
and let $\mathcal P$ be an open subset of
the $s=-2int$ plane $\mathbb C$ containing the set $\Omega$ of all the (finitely many) nonzero points 
where the $\Psi$-RH problem is not solvable.
Let $\om(x)\equiv\om(x;|||\bt|||)$ be a positive, smooth function for $x$ sufficiently large, with the following behavior:
\be
\om(n)\to\infty,\quad
\om(n)=o(n^{\ep}),\qquad \ep=\min\left\{1, \frac{2}{1+2|||\bt|||}\right\},\qquad\mbox{as $n\to\infty$}.
\ee
There holds the following asymptotic expansion:
\begin{multline}\label{diF}
{1\over i}\frac{d}{dt}\ln D_n(f_t)=n(\bt_2-\bt_1)+d_1(t;\alpha_1,\beta_1,\alpha_2,\beta_2)
+d_2(n,t;\alpha_1,\beta_1,\alpha_2,\beta_2)\\
+d_3(t;\alpha_1,\beta_1,\alpha_2,\beta_2)+\mathcal E_{n,t},
\end{multline}
as $n\to\infty$, where $\mathcal E_{n,t}$ is such that
\be\label{A est}
\left|\int_0^{t}\mathcal E_{n,x}dx\right|=\bigO(\om(n)^{-1+|||\bt|||})+\bigO(n^{-2}\om(n)^{1+2|||\bt|||})=o(1),
\ee
uniformly for $0<t<t_0$ and $-2int\in \mathbb C\setminus \mathcal P$
(the path of integration in (\ref{A est}) avoids the points of $\Omega$), and where
\begin{align}
&d_1(t;\alpha_1,\beta_1,\alpha_2,\beta_2)=-\alpha_1\sum_{j\neq 0}jV_je^{ijt}+\alpha_2\sum_{j\neq 0}jV_je^{-ijt}+(\alpha_2-\alpha_1)(\beta_1+\beta_2)\nonumber\\
&\qquad\qquad\qquad\qquad+i(\beta_1+\beta_2)\sum_{j=1}^{+\infty}j(V_j-V_{-j})\sin(jt),\label{d1}\\
&\label{d2}d_2(n,t;\alpha_1,\beta_1,\alpha_2,\beta_2)=
((\beta_1+\beta_2)^2-4\alpha_1\alpha_2)\frac{\cos t}{2i\sin t}
+\frac{1}{it}\sigma(-2int),\\
&d_3(t;\alpha_1,\beta_1,\alpha_2,\beta_2)=2\sigma_s\left[-\sum_{j=1}^{+\infty}j(V_j+V_{-j})\cos(jt)\right. \nonumber\\
&\label{d3}\qquad\qquad\qquad\qquad\qquad\qquad\qquad\left. +\frac{\bt_1-\bt_2}{2i}
\left(\frac{\cos t}{\sin t}-\frac{1}{t}\right)-\alpha_1-\alpha_2\right].
\end{align}
\end{proposition}

\begin{proof}
In the derivation below, we assume $\al_k\neq 0$, $k=1,2$, so that we can use Proposition \ref{prop21}.
Once (\ref{diF}) is proved under this assumption, the general result follows immediately from the uniformity 
of the error term in $\al_1$, $\al_2$. This uniformity is easy to verify from the constructions above. 
Alternatively, one can consider the case $\al_k=0$ separately using the corresponding differential identity: see Remark 
\ref{remark-di}.

For simplicity of the notation, we also assume below that $\Re\al_k\ge 0$, $k=1,2$ (in this case,
$\wt Y=Y$ in the Proposition \ref{prop21}). The extension to the case $-1/2<\Re\al_k<0$ is an easy exercise.

The plan of the proof is as follows. First, we express the differential identity of Proposition \ref{prop21} in
terms of the parametrices of the previous section (separately for the regions $0<t\le \om(n)/n$, 
$\om(n)/n<t<t_0$) and estimate the error terms. The error term estimation is especially involved in the latter
region. To show that (\ref{A est}) is $o(1)$, we use, in particular, large oscillations of $\mathcal E_{n,t}$.  
This difficulty is caused by the presence of $\beta$-singularities, the situation in the case of $\bt_k=0$, $k=1,2$,
and even in the case $|||\bt|||<1/2$, is simpler.
Second, we compute the leading asymptotic terms in the differential identity from the parametrices.

\subsubsection*{Transformation of the differential identity and estimates for the error terms}
Using the transformation $Y\mapsto T\mapsto S$, we can write the differential identity 
(\ref{differentialidentity}) in the form
\be\label{diS}
{1\over i}\frac{d}{dt}\ln D_n(f_t)=
\sum_{k=1}^2 (-1)^k\left[
n(\al_k+\bt_k)
+2\al_k z_k \left(S^{-1}{dS\over dz}\right)_{+,22}(z_k)\right],
\ee
where the limit $\left(S^{-1}{dS\over dz}\right)_{+,22}(z_k)$ is taken as $z\to z_k$ from the inside of the unit circle and outside the lenses.

Consider the case $0<t\le \om(n)/n$. For simplicity, we again assume that there are no points of $\Omega$ on
$0<t\le \om(n)/n$: cf. the first paragraph of section \ref{sec74}).
By (\ref{def Ups}), we obtain for $z\in \mathcal U$:
\begin{multline}\label{exprT2}
\left(S^{-1}{dS\over dz}\right)_{22}(z)=\left(P^{-1}{dP\over dz}\right)_{22}(z)+A_{n,t}(z),\\
A_{n,t}(z)=\left(P^{-1}\wt n^{-(\beta_1+\beta_2)\sigma_3}\Upsilon^{-1}{d\Upsilon\over dz}\wt n^{(\beta_1+\beta_2)\sigma_3}P\right)_{22}(z).
\end{multline}
By (\ref{def P}), (\ref{def W}), (\ref{whE}), and (\ref{h}), this can written for $|z|<1$ as
\begin{align}
\left(P^{-1}{dP\over dz}\right)_{22}(z)&=-\frac{n}{2z}+\frac{1}{2}\frac{f_t'}{f_t}(z)+\left(\Phi^{-1}{d\Phi\over dz}\right)_{22}+h(z)\left(\Phi^{-1}\sigma_3\Phi\right)_{22},\label{exprT3}\\
A_{n,t}(z)&=\left(\Phi^{-1}\left(\frac{n}{\tilde n}\right)^{-(\beta_1+\beta_2)\sigma_3}\widehat E^{-1} \Upsilon^{-1}{d\Upsilon\over dz}\widehat E\left(\frac{n}{\tilde n}\right)^{(\beta_1+\beta_2)\sigma_3}\Phi\right)_{22}(z),\label{Ant}
\end{align}
where $\Phi=\Phi(\frac{1}{t}\ln z;-2int)$.

We now show that the integral $\int_{0}^{t} |A_{n,t}(z_k)|dt$, $0<t\le\om(n)/n$ is small for large $n$ and $k=1,2$
if $|||\bt|||<1$.

By (\ref{as R}), we have, uniformly in $z$ and $0<t<t_0$, that 
\be\label{YY-}
\Upsilon^{-1}(z){d\Upsilon\over dz}(z)=\bigO(t(nt)^{-1+|||\beta|||}), \qquad n\to\infty,
\ee
and we obtain by the fact that $\widehat E(z_k)$ is bounded in $n$ (uniformly for $0<t<t_0$):
\begin{equation}\label{last term1}
A_{n,t}(z_k)=\left(\Phi^{-1}\left(\frac{n}{\tilde n}\right)^{-(\beta_1+\beta_2)\sigma_3}\bigO(t(nt)^{-1+|||\beta|||})\left(\frac{n}{\tilde n}\right)^{(\beta_1+\beta_2)\sigma_3}\Phi\right)_{22}(z_k).
\end{equation}
Now take the constants $c_0$, $C_0$ from the proof of Proposition \ref{PropEst}.

If $c_0<nt\le C_0$, both $n/\widetilde n$ and $\Phi$ are bounded, and we obtain
\begin{equation}
A_{n,t}(z_k)=\bigO(t(nt)^{-1+|||\beta|||}),\qquad k=1,2,\label{A1est}
\end{equation}
uniformly for $\frac{c_0}{n}<t\le \frac{C_0}{n}$.

If $0<t\le c_0/n$ we use the small $s$ asymptotics for $\Phi$. By (\ref{whPsi})--(\ref{whPsi2}), (\ref{def R0}), (\ref{Rsol}), and (\ref{delta}),
\[
A_{n,t}(z_k)=\left(P_0^{-1}(\lambda_k)\bigO(t(nt)^{-1+|||\beta|||})P_0(\lambda_k)\right)_{22},\qquad k=1,2,\qquad \lambda_1=0, \ \lambda_2=s,
\]
uniformly.
By (\ref{def P00}) and (\ref{def P00half}), this implies the estimate (\ref{A1est}) after a straightforward calculation.
Thus, (\ref{A1est}) holds uniformly for $0<t\le C_0/n$ as $n\to\infty$.

Finally, we set $C_0/n<t\le\om(n)/n$.
Let us consider the case $z\to z_1$ as the case of $z\to z_2$ is dealt with similarly.  
Combining (\ref{A}) with (\ref{def U}), (\ref{def tilde R}), and (\ref{def R}), we have
in the neighborhood $z(\mathcal U_1)\subset \mathcal U$,
where $z(\mathcal U_1)$ is the image of $\mathcal U_1$ under the inverse of the map
$\zeta=\frac{1}{t}\ln z$,
\begin{multline}
A_{n,t}(z)=\left(P_1^{-1}(\zeta)(nt)^{-\frac{1}{2}(\beta_1+\beta_2)\sigma_3}R^{-1}(\zeta)
\bigO(t(nt)^{-1+|||\beta|||})
R(\zeta)(nt)^{\frac{1}{2}(\beta_1+\beta_2)\sigma_3}P_1(\zeta)\right)_{22},\\
\zeta=\frac{1}{t}\ln z.
\end{multline}
Note that by (\ref{as R-1}), $R(\zeta)=I+O((nt)^{-1+|||\bt|||})$ for $z\in z(\mathcal U_1)$ uniformly in $t>C_0/n$.
By this observation and (\ref{def P+}), we can write further in $z(\mathcal U_1)$:
\begin{multline}\label{Ainterim}
A_{n,t}(z)=\left(M^{(\al_1,\bt_1)}(nt(\zeta-i))^{-1}E_1^{(0)}(\zeta)^{-1}
(nt)^{-\frac{1}{2}(\beta_1-\beta_2)\sigma_3}
\bigO(t(nt)^{-1+|||\beta|||})
\right.\\
\left.
\times
(nt)^{\frac{1}{2}(\beta_1-\beta_2)\sigma_3} E_1^{(0)}(\zeta)M^{(\al_1,\bt_1)}(nt(\zeta-i))\right)_{22},
\end{multline}
where $E^{(0)}(\zeta)=(nt)^{\bt_2\si_3}E_1(\zeta)$, with $E_1$ given by (\ref{def E1}).
Note that $E^{(0)}(\zeta)$
is uniformly bounded in $n$ for any $t>C_0/n$ and for $z\in z(\mathcal U_1)$.

Substituting (\ref{M3}) if $2\al_1\neq 1,2,\dots$ or (\ref{M3half}) if $2\al_1= 1,2,\dots$ into this expression
gives for $z_1$ uniformly in $t$ (the same estimate for $z_2$ is obtained similarly):
\be\label{A1 first case final}
A_{n,t}(z_k)=\bigO(n^{-1+2|||\bt|||}t^{2|||\bt|||}),\qquad k=1,2,\qquad \frac{C_0}{n}<t\le\frac{\om(n)}{n}.
\ee
Recalling (\ref{A1est}) we conclude that the contribution of this term
to the logarithm of the determinant is, uniformly in $t$,
\be\label{A1 int est}
\int_{0}^{t}|A_{n,t}(z_k)|dt=\bigO(n^{-2}\om(n)^{1+2|||\bt|||}),
\qquad 0<t\le\om(n)/n.
\ee
This is small for $|||\bt|||<1$ if $\om(n)=o(n^{\ep})$ with $\ep=2/(1+2|||\bt|||)$.

\bigskip

If $\om(n)/n<t<t_0$, we consider only the neighborhood of $z_1$ as the contribution of $z_2$ is dealt with similarly. For $\om(n)/n<t<t_0$, $z\in\wt{\mathcal U_1}$,
we obtain instead of (\ref{exprT2}):
\begin{multline}\label{exprT4}
\left(S^{-1}{dS\over dz}\right)_{22}(z)=\left(\wt P^{-1}_1{d\wt P_1\over dz}\right)_{22}(z)+A_{n,t}(z),\\
A_{n,t}(z)=\left(\wt P^{-1}_1
\left(\frac{n}{t}\right)^{-\frac{\beta_1+\beta_2}{2}\sigma_3}\Upsilon^{-1}{d\Upsilon\over dz}
\left(\frac{n}{t}\right)^{\frac{\beta_1+\beta_2}{2}\sigma_3}
\wt P_1\right)_{22}(z)
\end{multline}
with $\Upsilon(z)$ from Section \ref{sec 751}.
By (\ref{def P1}) and (\ref{def W}),  this can written out for $|z|<1$ as
\begin{align}
\left(\wt P^{-1}_1{d\wt P_1\over dz}\right)_{22}(z)&=
-\frac{n}{2z}+\frac{1}{2}\frac{f_t'}{f_t}(z)+\left(M^{-1}{dM\over dz}\right)_{22}+\wt h_1(z)\left(M^{-1}\sigma_3 M\right)_{22},\label{exprT5}\\
A_{n,t}(z_1)&\equiv\lim_{z\to z_1}A_{n,t}(z)\nonumber\\
&=\lim_{z\to z_1}\left(M^{-1}\wh E^{-1}_1 (nt)^{\frac{\beta_1-\beta_2}{2}\sigma_3}
\Upsilon^{-1}{d\Upsilon\over dz}(nt)^{-\frac{\beta_1-\beta_2}{2}\sigma_3}
\wh E_1M\right)_{22},\label{A2}
\end{align}
where $M=M^{(\al_1,\bt_1)}(nt(\frac{1}{t}\ln z-i))$,  $\wh E_1$ is given by (\ref{hatE_1}), and
\be\label{wt h1}
\wt h_1(z)\si_3=\wh E_1^{-1}(z){d\over dz}\wh E_1(z),\qquad
\wt h_1(z)=h(z)-\frac{\bt_2}{z\ln z+itz},
\ee
in terms of $h(z)$ given by (\ref{h}).

We now estimate the error term (\ref{A2}). 
A straightforward estimate by (\ref{as R2}) shows
the smallness of the error term only for $|||\bt|||<1/2$. Therefore, we will use (\ref{UZ}).
We assume (for simplicity only) that $\Re\beta_1>\Re\beta_2$.

Substituting (\ref{UZ}) into (\ref{A2}), we obtain
\[
A_{n,t}=B_{n,t}+C_{n,t},
\]
where
\be
B_{n,t}=\lim_{z\to z_1}\left(M^{-1}\wh E^{-1}_1 (nt)^{\frac{\beta_1-\beta_2}{2}\sigma_3}
Z^{-1}{dZ\over dz}(nt)^{-\frac{\beta_1-\beta_2}{2}\sigma_3}
\wh E_1M\right)_{22}\label{B1}
\ee
and
\be
C_{n,t}=\lim_{z\to z_1}\left(M^{-1}\wh E^{-1}_1 (nt)^{\frac{\beta_1-\beta_2}{2}\sigma_3}
Z^{-1}\wh\Upsilon^{-1}{d\wh\Upsilon\over dz}Z(nt)^{-\frac{\beta_1-\beta_2}{2}\sigma_3}
\wh E_1M\right)_{22}.\label{C1b}
\ee

The estimates (\ref{whUpsest}), (\ref{Z est}) give
\be\label{C2est}
C_{n,t}=\bigO(t^{-1}(nt)^{-1+|||\bt|||}),
\ee
and therefore, uniformly in $t$,
\be\label{C2intest}
\int_{\om(n)/n}^{t}|C_{n,t}|dt=\bigO(\om(n)^{-1+|||\bt|||}),\qquad  \om(n)/n<t<t_0.
\ee

For $B_{n,t}$, we obtain from (\ref{ZZ est}):
\begin{multline}\label{B2}
B_{n,t}=B^{(1)}+B^{(2)},\qquad
B^{(1)}=\frac{c (nt)^{\bt_1-\bt_2}}{4\sin^2 t}\lim_{z\to z_1}
\left(M^{-1}\wh E^{-1}_1\si_+
\wh E_1M\right)_{22}(e^{it}),\\
|B^{(2)}|=
\bigO(n^{-1+|||\bt|||})+\bigO(t^{-1}(nt)^{-2+2|||\bt|||}).
\end{multline}
Here, uniformly in $t$,
\be\label{B2est}
\int_{\om(n)/n}^{t}|B^{(2)}|dt=\bigO(n^{-1+|||\bt|||})+\bigO(\om(n)^{-2+2|||\bt|||}),\qquad  \om(n)/n<t<t_0.
\ee
Now using the definition of $c$ in (\ref{abcd}) and of $\wh E_2(z)$ in (\ref{hatE_2}),
we can write:
\begin{multline}\label{B3}
\frac{c (nt)^{\bt_1-\bt_2}}{4\sin^2 t}=
\frac{(nt)^{-1+2(\bt_1-\bt_2)}}{4\sin^2 t}
t e^{-it}\de(\wh E_2 M_1\wh E^{-1}_2)_{12}(e^{-it})\\=
\frac{(nt)^{-1+2(\bt_1-\bt_2)}}{4\sin^2 t}
t e^{-it}\de
{\mathcal D}(e^{-it})\left(\frac{-2i\sin t}{t}\right)^{2\bt_1}e^{-2it\bt_2}e^{-int}(\si_1\Omega_2^{-1}
M_1\Omega_2\si_1)_{12}\\=
e^{-int}\frac{(nt)^{-1+2(\bt_1-\bt_2)}}{t}\de\ep(t),
\end{multline}
where $\ep(t)$ is independent of $n$ and analytic in $0< t \le t_0$ (as follows
from the uniform boundedness of ${\mathcal D}(z)$).

Let
\be\label{2hatE_1}
\wh{\wh E}_1(t)=e^{-{int\over 2}\si_3}\wh E_1(e^{it})=
\si_1\mathcal D_t(z)^{-\si_3/2}\left(\frac{t}{2i\sin t}\right)^{\bt_2\si_3}
e^{-it\bt_1\si_3}\Omega^{-1}_1,
\ee
where $\wh E_1(z)$ is given by (\ref{hatE_1}). So defined $\wh{\wh E}_1$ is independent of $n$.
Substituting (\ref{B3}), (\ref{2hatE_1}), and (\ref{M3}) (or (\ref{M3half})) into (\ref{B2}), we obtain
\be\label{B4}
B^{(1)}=-e^{-2int}\frac{(nt)^{-1+2|||\bt|||}}{t}\de\ep_1(t),
\ee
where (see (\ref{ME}))
\[
\ep_1(t)=\ep(t)(\wh{\wh E}_1(t)L(0))_{21}(\wh{\wh E}_1(t)L(0))_{22}
\]
is independent of $n$ and analytic in $0< t \le t_0$.

Note that, as is established by an easy calculation using the definition of $\de$ in (\ref{abcd}), 
$\de=1+\bigO((nt)^{-2+2|||\bt|||})$ and its derivative $d\de(t)/dt=\bigO(t^{-1}(nt)^{-2+2|||\bt|||})$, uniformly in $t$. 
Moreover, we obtain from (\ref{2hatE_1}) that both $\wh{\wh E_1}(t)$ and its derivative
are uniformly bounded.
Now the estimate (\ref{B4}) implies by integration by parts that, uniformly in $t$,
\be\label{B5est}
\left|\int_{\om(n)/n}^{t}B^{(1)}dt\right|=\bigO(\om(n)^{-2+2|||\bt|||}),\qquad \om(n)/n<t<t_0.
\ee

Combining (\ref{B2est}), (\ref{B5est}),  and (\ref{C2intest}), we finally see that the integral of (\ref{A2}) is estimated as follows, uniformly for $\om(n)/n<t<t_0$,
\be\label{A2est}
\left|\int_{\om(n)/n}^{t}A_{n,t}(z_1)dt\right|=
\left|\int_{\om(n)/n}^{t}(B^{(1)}+B^{(2)}+C_{n,t})dt\right|=
\bigO(\om(n)^{-1+|||\bt|||}).
\ee

Together with (\ref{A1 int est}), this will imply below that the contributions
of the $\mathcal E$ terms to the integral from $0$ to $t$, $0<t<t_0$ of the differential identity (\ref{diS}) are
uniformly small for $|||\bt|||<1$.

\subsubsection*{Calculation of the main asymptotic terms}

We now turn to computing the contribution of (\ref{exprT3}) and (\ref{exprT5}).
Consider first (\ref{exprT3}). This corresponds to the case $0<t\le\om(n)/n$.
We use (\ref{def Phi}), (\ref{Psi i})--(\ref{Psi -i 2}), and (\ref{Eh}) to obtain
\begin{multline}\label{767}
\left(P^{-1}{dP\over dz}\right)_{22,+}(z_k)=-\frac{n}{2z_k}+\lim_{z\to z_k}\left(\frac{1}{2}\frac{f_t'}{f_t}(z)-\frac{\alpha_k}{z\ln z-z\ln z_k}\right)\\+\frac{1}{tz_k}\left(F_k^{-1}F_k'\right)_{22}(\zeta_k)+h(z_k)\left(F_k^{-1}\sigma_3F_k\right)_{22}(\zeta_k),
\end{multline}
where we wrote $\zeta_k=\frac{1}{t}\ln z_k=\pm i$.
Using (\ref{fDD}) we obtain
\be\label{fpf}
\lim_{z\to z_k}\left(\frac{1}{2}\frac{f_t'}{f_t}(z)-\frac{\alpha_k}{z\ln z-z\ln z_k}\right)=
\frac{1}{2}V'(z_k)-\frac{\alpha_{k'}-\beta_1-\beta_2}{2z_k}+\frac{\alpha_{k'}}{z_k-z_{k'}}.
\ee
Here $k'$ is equal to $1$ if $k=2$, and equal to $2$ if $k=1$.
Substituting (\ref{767}) into (\ref{exprT2}) and that, in turn, into (\ref{diS}), we obtain
\begin{multline}\label{diF2}
{1\over i}\frac{d}{dt}\ln D_n(f_t)=
\sum_{k=1}^2 (-1)^k\left[
n\beta_k +\alpha_kz_kV'(z_k)+\alpha_k(\beta_1+\beta_2)
+\frac{2}{t}\al_k \left(F_k^{-1}F_k'\right)_{22}(\zeta_k)\right. \\
\left.
+2\al_kz_kh(z_k) \left(F_k^{-1}\sigma_3F_k\right)_{22}(\zeta_k)\right]+2i\alpha_1\alpha_2\frac{\cos t}{\sin t} +\wt{\mathcal E}_{n,t},
\end{multline}
where
\be\label{A def}
\wt{\mathcal E}_{n,t}=2\sum_{k=1}^2(-1)^k \al_k z_k A_{n,t}(z_k).
\ee

Now we use the identities from Proposition \ref{prop F1} and Proposition \ref{prop F2} to conclude that
\be\label{diF3}
\begin{aligned}
{1\over i}\frac{d}{dt}\ln D_n(f_t)=n(\beta_2-\beta_1)+2i\alpha_1\alpha_2\frac{\cos t}{\sin t}
+\frac{2}{t}(c_2-c_1)\\
+\frac{1}{it}\sigma+2\sigma_s\left[
z_1h(z_1)+z_2h(z_2)\right] -(\beta_1+\beta_2)\left[
z_1h(z_1)-z_2h(z_2)\right]\\-\alpha_1\sum_{j\neq 0}jV_je^{ijt}+\alpha_2\sum_{j\neq 0}jV_je^{-ijt}+(\alpha_2-\alpha_1)(\beta_1+\beta_2)+\wt{\mathcal E}_{n,t}.
\end{aligned}
\ee
From (\ref{h}), we note that
\begin{multline}\label{h z_k}
h(z_k)z_k=
-\frac{1}{2}\sum_{j=1}^{+\infty}jV_jz_k^j+\frac{1}{2}\sum_{j=-1}^{-\infty}jV_jz_k^j+(-1)^k\frac{\beta_{k'}z_k}{2i\sin t}-(-1)^k\frac{\beta_{k'}}{2it}\\-\frac{\alpha_1+\alpha_2-\beta_k'}{2}.
\end{multline}
Substituting this into (\ref{diF3}), we obtain
\begin{multline}
\label{differentialidentityF4}
{1\over i}\frac{d}{dt}\ln D_n(f_t)=n(\beta_2-\beta_1)+2i\alpha_1\alpha_2\frac{\cos t}{\sin t}+\frac{2}{t}(c_2-c_1)+\frac{1}{it}\sigma\\
-\alpha_1\sum_{j\neq 0}jV_je^{ijt}+\alpha_2\sum_{j\neq 0}jV_je^{-ijt}+(\alpha_2-\alpha_1)(\beta_1+\beta_2)\\
+i(\beta_1+\beta_2)\sum_{j=1}^{+\infty}j(V_j-V_{-j})\sin(jt)
+\frac{(\beta_1+\beta_2)^2}{2i}\left(\frac{\cos t}{\sin t}  -{1\over t}\right)\\
+2\sigma_s\left[
-\sum_{j=1}^{+\infty}j(V_j+V_{-j})\cos(jt)+\frac{\beta_1e^{-it}-\beta_2e^{it}}{2i\sin t}+\frac{\beta_2-\beta_1}{2it}-\alpha_1-\alpha_2+\frac{\bt_1+\bt_2}{2}\right] \\+\wt{\mathcal E}_{n,t}.
\end{multline}
Substituting here the values (\ref{c1c2}) of $c_1$ and $c_2$, and simplifying further, we obtain
the differential identity (\ref{diF}) for $0<t\leq {\om(n)/n}$. Set
\[
\mathcal E_{n,t}=\wt{\mathcal E}_{n,t},\qquad 0<t\leq {\om(n)/n}.
\]
From (\ref{A1 int est}) and (\ref{A def}),
we have the estimate (\ref{A est}) uniformly for $0<t\le \om(n)/n$.

\bigskip

Let us now consider the region $\om(n)/n<t<t_0$, i.e., consider (\ref{exprT5}).
We assume for simplicity that $\al_k$, $k=1,2$ is not half integer. The case $2\alpha_1$ integer can be 
treated similarly by (\ref{M3half})--(\ref{mhalf}).
For $z$ approaching $z_1$ from the inside of the unit circle and outside the lens, we use the representation of
$M=M^{(\al_1,\bt_1)}(n\ln z-int)$ in region III of Figure \ref{Figure: M}, i.e. (\ref{M3}), and obtain:
\be
\left(M^{-1}{dM\over dz}\right)_{22}=
\left[(L(\lb)^{-1}{dL\over d\lb}(\lb))_{22}-\frac{\al_1}{\lb}\right]\frac{n}{z},\qquad \lb=n\ln z-int,\qquad 2\alpha_1\neq 0,1,\ldots 
\ee
up to the terms of order $\lb^{2\al_1}$ that disappear (or would be removed if we considered the case 
$\Re\al<0$) in the (``regularized'' for $\Re\al<0$) limit $z\to z_1$.
Substituting here the explicit formula (\ref{ME}) for $L$ and setting $z=z_1$ in the terms which are not unbounded
as $z\to z_1$, we obtain:
\be
\left(M^{-1}{dM\over dz}\right)_{22}=-\left[\frac{\bt_1}{2\al_1}+\frac{\al_1}{\lb}\right]\frac{n}{z}
\ee
in the vicinity of $z_1$ inside the unit circle, outside the lens. Similarly, we have in the same limit
\be
\left(M^{-1}\sigma_3 M\right)_{22}=\left(L^{-1}(0)\sigma_3 L(0)\right)_{22}=\frac{\bt_1}{\al_1}.
\ee
To evaluate $\wt h_1(z)$ at $z_1$, we use its definition (\ref{wt h1}) and (\ref{h z_k}). Collecting our results
together, we can write
\begin{multline}
z_1\left(M^{-1}{dM\over dz}\right)_{22}+z_1\wt h_1(z_1)\left(M^{-1}\sigma_3 M\right)_{22}=
-\left[\frac{\bt_1}{2\al_1}+\frac{\al_1}{n\ln z-int}\right]n\\
-\frac{\bt_1}{2\al_1}\left(\sum_{j=1}^{+\infty}jV_jz_1^j-\sum_{j=-1}^{-\infty}jV_jz_1^j+
\frac{\bt_2 z_1}{i\sin t}+\al_1+\al_2-\bt_2\right).
\end{multline}
This formula holds
in the vicinity of $z_1$ inside the unit circle, outside the lens. 
Now substituting it into (\ref{exprT5}) and using (\ref{fpf}), we obtain:
\begin{multline}\label{diP1}
-2 \al_1 z_1\left(\wt P_1^{-1}{d\wt P_1\over dz}\right)_{+,22}(z_1)
=n(\al_1+\bt_1)-\al_1z_1{d\over dz}V(z_1)-\frac{\al_1\al_2z_1}{i\sin t}+\al_1(\al_2-\bt_1-\bt_2)\\
+\bt_1\left(\sum_{j=1}^{+\infty}jV_jz_1^j-\sum_{j=-1}^{-\infty}jV_jz_1^j+
\frac{\bt_2 z_1}{i\sin t}+\al_1+\al_2-\bt_2\right).
\end{multline}

The analysis of the neighborhood of $z_2$ is similar, and we obtain
\begin{multline}\label{diP2}
2\al_2 z_2 \left(\wt P_2^{-1}{d\wt P_2\over dz}\right)_{+,22}(z_2)
=-n(\al_2+\bt_2)+\al_2 z_2{d\over dz}V(z_2)-\frac{\al_1\al_2z_2}{i\sin t}-\al_2(\al_1-\bt_1-\bt_2)\\
-\bt_2\left(\sum_{j=1}^{+\infty}jV_jz_2^j-\sum_{j=-1}^{-\infty}jV_jz_2^j+
\frac{\bt_1 z_2}{i\sin t}+\al_1+\al_2-\bt_1\right).
\end{multline}

Substituting (\ref{diP1}) and (\ref{diP2}) into (\ref{diS}), we finally obtain
that, for $\om(n)/n<t<t_0$,
\be\label{di large t}
{1\over i}{d\over dt}\ln D_n=
S_1+S_2+\wt{\mathcal E}_{n,t}
\ee
where $\wt{\mathcal E}_{n,t}$ is given by (\ref{A def}),
\be\label{S1}
S_1=n(\bt_2-\bt_1)-\al_1z_1{d\over dz}V(z_1)+\al_2z_2{d\over dz}V(z_2)
+(\al_2-\al_1)(\bt_1+\bt_2),
\ee
and
\begin{multline}\label{S2}
S_2=-n(\bt_2-\bt_1)+\bt_1\left(\sum_{j=1}^{+\infty}jV_jz_1^j-\sum_{j=-1}^{-\infty}jV_jz_1^j\right)
-\bt_2\left(\sum_{j=1}^{+\infty}jV_jz_2^j-\sum_{j=-1}^{-\infty}jV_jz_2^j\right)\\
+2\left(\bt_1\bt_2-\al_1\al_2\right)\frac{\cos t}{i \sin t}+(\al_1+\al_2)(\bt_1-\bt_2).
\end{multline}
Let us compare these expressions with (\ref{diF}) obtained for $0<t\le\om(n)/n$.
First, we see that $S_1$ coincides with the sum of $n(\bt_2-\bt_1)$ and a part of $d_1$.
Now consider $d_2+d_3$ for large $s=-2int$.
Substituting there the expansion (\ref{asinfty}) for $\si(s)$, we obtain that
\[
n(\bt_2-\bt_1)+d_1+d_2+d_3=S_1+S_2+\Theta_{n,t},\qquad  \om(n)/n<t<t_0.
\]
where $\Theta_{n,t}$ is a term arising from the error term and
from $\gamma(s)$ in (\ref{asinfty}) and which, in particular,
because of the oscillatory factors $e^{\pm i|s|}$, becomes of order $\om(n)^{-2+2|||\bt|||}$ after integration w.r.t. $t$
(cf. (\ref{B5est})):
\be\label{Th2}
\left|\int_{\om(n)/n}^{t}\Theta_{n,\tau}d\tau\right|=\bigO(\om(n)^{-2+2|||\bt|||}),\qquad \om(n)/n<t<t_0,
\ee
uniformly in $t$.
Set
\[
\mathcal E_{n,t}=\wt{\mathcal E}_{n,t}+\Theta_{n,t},\qquad  \om(n)/n<t<t_0.
\]

Thus, expressions (\ref{diF}), (\ref{d1}), (\ref{d2}), (\ref{d3}) remain valid also for the
region $\om(n)/n<t<t_0$. It remains to verify the smallness of the error term in this region.
This, however, follows immediately from (\ref{A def}), (\ref{A2est}),
similar estimates
for $A_{n,t}(z_2)$,  and from (\ref{Th2}).
\end{proof}

\begin{remark}\label{extension}
Integrating (\ref{di large t}) between $t$ and $t_0$ with $\om(n)/n\le t<t_0$ and using the expansion (\ref{as FH2}) for 
$D_n(f_{t_0})$, we obtain the same expansion for $D_n(f_t)$ with the error term 
$\bigO(\om(n)^{-1+|||\bt|||})$, and $\om(x)$ such that
$\om(n)\to\infty$, $\om(n)=o(n)$ as $n\to\infty$. (Further limitations for the order of
$\om(n)$ at infinity in Proposition \ref{prop diff id F} came from 
the interval $0<t\le\om(n)/n$ that we do not need in this case.)
Thus, the formula (\ref{as FH2})
remains valid in the region $\om(n)/n\le t<t_0$. The uniformity of the error term is
easy to verify. 
Furthermore, as is clear from our constructions above, the function $\om(n)$ can be replaced by a 
sufficiently large constant $s_0$, and the error terms containing $\om(n)$ or $nt$, by appropriate 
estimates.
We therefore proved Theorem \ref{theorem: extension}.
\end{remark}

\subsection{Integration of the differential identity. Proof of Theorems \ref{theorem: Toeplitz}, \ref{theorem: Toeplitz 2}.}

Integrating (\ref{diF}) between $0$ and $t$ gives
\begin{multline}\label{intid}
\ln D_n(f_t)=\ln D_n(0)+int(\beta_2-\beta_1)+i\int_0^t d_1(\tau;\alpha_1,\beta_1,\alpha_2,\beta_2)d\tau\\
+i\int_0^t d_2(n,\tau;\alpha_1,\beta_1,\alpha_2,\beta_2)d\tau + i\int_0^t d_3(n,\tau;\alpha_1,\beta_1,\alpha_2,\beta_2)d\tau
+\bigO(n^{-1+|||\bt|||}),
\end{multline}
uniformly for $0<t<t_0$. If $\bt_1$, $\bt_2$ are not zero or purely imaginary, then the contour of
integration in (\ref{intid}) is chosen to avoid possible poles of $\si(s)$. This is the reason
for the remark at the beginning of Section \ref{sec74}. Note that, by Theorem \ref{theorem: PV}, there are
no poles for $0<t<c_0/n$ and $t>C_0/n$ if $c_0$ is sufficiently small, and $C_0$, sufficiently large.

Using the definitions (\ref{WH}), we obtain
\begin{multline}\label{I1}
i\int_0^t d_1(\tau;\alpha_1,\beta_1,\alpha_2,\beta_2)d\tau=
(\alpha_2-\alpha_1)(\beta_1+\beta_2)it
-\al_1(V(e^{it})-V(1))\\-\al_2(V(e^{-it})-V(1))+
{1\over 2}(\bt_1+\bt_2)\ln\frac{b_+(e^{it})b_+(e^{-it})}{b_-(e^{it})b_-(e^{-it})}
+(\bt_1+\bt_2)\ln\frac{b_-(1)}{b_+(1)}.
\end{multline}
To integrate $d_2$, we add and subtract to it $\sigma(0)/(it)$, and recall from (\ref{as0}) that
\[
\sigma(0)=2\al_1\al_2-{1\over 2}(\bt_1+\bt_2)^2.
\]
We then obtain:
\begin{multline}\label{I2}
i\int_0^t d_2(n,\tau;\alpha_1,\beta_1,\alpha_2,\beta_2)d\tau=
\int_{0}^{-2int}\frac{1}{s}\left(\sigma(s)-2\alpha_1\alpha_2+\frac{1}{2}(\beta_1+\beta_2)^2\right)ds\\
+\left({1\over 2}(\beta_1+\beta_2)^2-2\alpha_1\alpha_2\right)\ln\frac{\sin t}{t}.
\end{multline}
To integrate $d_3$, write it first in the form
\be\label{I3}
d_3(n,t;\alpha_1,\beta_1,\alpha_2,\beta_2)=\si_s(s)\Lambda(t)=
\left(\si_s(s)-\frac{\bt_2-\bt_1}{2}\right)\Lambda(t)+\frac{\bt_2-\bt_1}{2}\Lambda(t),
\ee
where the expression for $\Lambda(t)$ is clear from the r.h.s. of (\ref{d3}). Note that
$\Lambda(t)$ is uniformly bounded in $t$. Using this fact and the large $s$ expansion (\ref{asinfty})
(which is differentiable in $s$), we can estimate the integral of the first term in the r.h.s. of (\ref{I3}) as follows
\be
\left|
\int_0^t   \left(\si_s(s)-\frac{\bt_2-\bt_1}{2}\right)\Lambda(t) dt\right|<
\frac{\mbox{Const}}{n}\int_0^{-2int}\left| \si_s(s)-\frac{\bt_2-\bt_1}{2}\right|ds=\bigO(n^{-1}).
\ee
For the second term in the r.h.s. of (\ref{I3}), we easily obtain
\be
\frac{\bt_2-\bt_1}{2}
\int_0^t  \Lambda(\tau) d\tau=(\al_1+\al_2)(\bt_1-\bt_2)it-\frac{(\bt_1-\bt_2)^2}{2}\ln\frac{\sin t}{t}+
\frac{\bt_1-\bt_2}{2}\ln\frac{b_+(e^{it})b_-(e^{-it})}{b_+(e^{-it})b_-(e^{it})}.
\ee

Substituting (\ref{I1}), (\ref{I2}), and (\ref{I3}) into (\ref{intid}), we obtain (\ref{asT}) and thus prove both
Theorem \ref{theorem: Toeplitz} and Theorem \ref{theorem: Toeplitz 2}.

\section{Toeplitz determinant for $|||\bt|||=1$}\label{section: beta1}

\subsection{Proof of Theorem \ref{theorem: Toeplitz 3}}
Let $\Re\al_1, \Re\al_2, \Re(\al_1+\al_2)>-1/2$.
Assume that $s=-2int$ is bounded
away from the set $\Omega$ where the $\Psi$-RH problem is not solvable.
Let
\[\Re\bt_1=\Re\bt_2,\] and define
\[
\bt_1^-=\bt_1,\qquad \bt_2^-=\bt_2-1.
\]
Then for the symbol $f^-\equiv f_t(z;\al_1,\al_2,\bt_1^-,\bt_2^-)$, we have $|||\bt^-|||=1$.
We will now find an asymptotic formula for $D_n(f^{-})$ for large $n$ uniform for $0<t<t_0$.
Our approach is based on the following identity (see \cite{DIK2}, Theorem 1.18):
\[
D_{n}(f^-)=z_2^n\frac{\wh \phi_n(0)}{\chi_n}D_n(f),\qquad n=1,2,\dots,
\]
which, since $D_n(f)=\prod_{j=0}^{n-1}\chi_j^{-2}$, can be written in the form
\be\label{mainsec9}
D_{n-1}(f^-)=z_2^{n-1}\wh \phi_{n-1}(0)\chi_{n-1}D_n(f),
\ee
convenient for us. Here $\wh \phi_{n-1}(0)$, $\chi_{n-1}$, refer to the polynomials orthogonal
w.r.t. $f\equiv f_t(z;\al_1,\al_2,\bt_1,\bt_2)$, i.e., corresponding to $|||\bt|||=0$. Theorem \ref{theorem: Toeplitz 2}  can be used
to write the asymptotics for $D_n(f)$, so it remains to estimate the prefactor in the r.h.s. of
(\ref{mainsec9}) as $n\to\infty$ using the results of Section \ref{section: RHP Y}.

Note that the parametrix/solution constructed in Section \ref{sec74} for $0<t\le 1/n$ remains 
valid for the case $0<t\le C_0/n$, where $C_0$ is a constant. 
Moreover, the parametrix/solution of Section \ref{sec75} for
$\om(n)/n<t<t_0$ remains valid for the case $C_0/n<t<t_0$ where $C_0$ is sufficiently large.
Both give expansions uniform in $t$. In the present section, we adopt this choice of solutions. Accordingly,
fix $C_0$ sufficiently large.

First, consider $0<t\le C_0/n$. Then, using (\ref{def Y}), (\ref{def Ups}) with $\wt n$ replaced by $n$, 
and (\ref{Dg1}),  we have
\begin{multline}\label{phichi}
\wh \phi_{n-1}(0)\chi_{n-1}=-\lim_{z\to\infty}z^{-n+1}Y_{21}(z)\\=
-\lim_{z\to\infty}z^{-n+1}\left(n^{-(\bt_1+\bt_2)\si_3}\Upsilon(z)n^{(\bt_1+\bt_2)\si_3}
\mathcal D_{out,t}(z)^{\si_3}z^{n\si_3}\right)_{21}\\=
-n^{2(\bt_1+\bt_2)}\left(\lim_{z\to\infty} z \Upsilon_{1,21}(z)+\bigO(t^{2}\Psi_2(s))\right).
\end{multline}
Here $\Upsilon_{1,21}(z)$ is the first term in the asymptotic expansion (\ref{as R}),
which (as well as the estimate $\bigO(t^{2}\Psi_2(s))$ for the error term) is obtained by the standard theory from the jump conditions for $\Upsilon$.
Namely, first,
using (\ref{def P}), (\ref{WN}), (\ref{whE}), and eventually (\ref{Psi as}), we write
(as $t\le C_0/n$)
\begin{multline}
n^{(\bt_1+\bt_2)\si_3}P(z)N^{-1}(z)n^{-(\bt_1+\bt_2)\si_3}=
\wh E(z)\Psi(\zeta,s)z^{\frac{n}{2}\si_3}P^{(\infty)}(\zeta)^{-1}\wh E(z)^{-1}\\ =
\wh E(z)\left(I+\frac{\Psi_1(s)}{t^{-1}\ln z}+\bigO(t^2\Psi_2(s))\right)\wh E(z)^{-1},
\qquad z\in\partial U,
\end{multline}
and then
\begin{multline}\label{Upsint}
\lim_{z\to\infty} z \Upsilon_{1,21}(z)=
\lim_{z\to\infty} \frac{z}{2\pi i}\int_{\partial U}
\frac{(\wh E(u)\Psi_1\wh E(u)^{-1})_{21}}{t^{-1}\ln u}\frac{du}{u-z}\\=
-\frac{1}{2\pi i}\int_{\partial U} \frac{(\wh E(u)\Psi_1\wh E(u)^{-1})_{21}}{t^{-1}\ln u}du.
\end{multline}
Here $\Psi_j(s)$ are the coefficients in the large $\zeta$ expansion (\ref{Psi as}).
Computing the residue at the simple pole $z=1$,
we obtain
\be\label{zU}
\lim_{z\to\infty} z \Upsilon_{1,21}(z)=
t\left(\wh E(1)\Psi_1\wh E(1)^{-1}\right)_{21}.
\ee
Thus, we have to evaluate $\wh E(1)$. We use the expression (\ref{whE1}).
First, from (\ref{Dl1}), (\ref{Dg1}), remembering the definition of the branches
given after (\ref{Dg1}), we obtain by a straightforward analysis of triangles that
\be
\mathcal D_{in,t}(1)\mathcal D_{out,t}(1)=b_0\frac{b_+(1)}{b_-(1)}(2\sin t)^{2(\bt_1+\bt_2)}
e^{i(\pi+t)(\al_1-\al_2)}.
\ee
Substituting this into (\ref{whE1}), and recalling the definition of the branches of
$\zeta\pm i$, we obtain
\be\label{E1}
\wh E(1)=\si_1\left(b_0\frac{b_+(1)}{b_-(1)}\right)^{-\si_3/2}
 \left(\frac{t}{\sin t}\right)^{(\bt_1+\bt_2)\si_3}e^{-i(\pi+t)(\al_1-\al_2)\si_3/2}
 e^{i\pi(3\bt_1+\bt_2)\si_3/2}.
 \ee
Therefore, substituting (\ref{E1}) into (\ref{zU}), and that, in turn, into (\ref{phichi}),
and recalling the definition (\ref{C1})
\[
\Psi_{1,12}=r(s),
\]
we obtain uniformly in $t$
\be\label{as phichi 1}
\begin{aligned}
\wh \phi_{n-1}(0)\chi_{n-1}&=
-r(s)b_0^{-1}\frac{b_-(1)}{b_+(1)}
 t\left(\frac{nt}{\sin t}\right)^{2(\bt_1+\bt_2)}
 e^{-i(\pi+t)(\al_1-\al_2)}e^{i\pi(3\bt_1+\bt_2)}\\ &+
 \bigO(\Psi_{2,12}(s)t^2n^{2(\bt_1+\bt_2)}),\qquad 0<t\le C_0/n, \quad s=-2int.
\end{aligned}
\ee

Since $\Psi(s)$ is bounded for $0<t\le C_0/n$, we obtain by (\ref{mainsec9}) the first part of
the r.h.s. of (\ref{asDb1}).

Now, let $C_0/n<t<t_0$ and consider a more general case $|||\bt|||<1$. 
As before, but now using (\ref{def Ups2}), we have
\begin{multline}\label{phichi2}
\wh \phi_{n-1}(0)\chi_{n-1}=-\lim_{z\to\infty}z^{-n+1}Y_{21}(z)\\=
-\lim_{z\to\infty}z^{-n+1}\left(\left(\frac{n}{t}\right)^{-\frac{\bt_1+\bt_2}{2}\si_3}
\Upsilon(z)\left(\frac{n}{t}\right)^{\frac{\bt_1+\bt_2}{2}\si_3}
\mathcal D_{out,t}(z)^{\si_3}z^{n\si_3}\right)_{21}\\=
-\left(\frac{n}{t}\right)^{\bt_1+\bt_2}
\left(\lim_{z\to\infty} z \Upsilon_{1,21}(z)+\bigO(t^{-1}n^{-2}(nt)^{\bt_1-\bt_2})+
\bigO(t^{-1}n^{-2}(nt)^{-\bt_1+\bt_2})\right),
\end{multline}
and we also have from the jump conditions (\ref{J1}), (\ref{J2}) in Section \ref{sec 751} that
\begin{multline}\label{Upsint2}
\lim_{z\to\infty} z \Upsilon_{1,21}(z)=
-\frac{1}{2\pi i}\int_{\partial\wt U_1}
\frac{(\wh E_1(u) M_1^{(\al_1,\bt_1)}\wh E_1(u)^{-1})_{21}}{t^{-1}\ln u-i}du (nt)^{-1+\bt_1-\bt_2}\\
\qquad\qquad-\frac{1}{2\pi i}\int_{\partial\wt U_2}
\frac{(\wh E_2(u) M_1^{(\al_2,\bt_2)}\wh E_2(u)^{-1})_{21}}{t^{-1}\ln u+i}du (nt)^{-1-\bt_1+\bt_2}\\
=z_1\frac{n^{2\bt_1-1}}{t^{2\bt_2}}(\wh E_1(z_1) M_1^{(\al_1,\bt_1)}\wh E_1(z_1)^{-1})_{21}
+z_2\frac{n^{2\bt_2-1}}{t^{2\bt_1}}(\wh E_2(z_2) M_1^{(\al_2,\bt_2)}\wh E_2(z_2)^{-1})_{21}.
\end{multline}
Here $M_1^{(\al_1,\bt_1)}$ is given by (\ref{M1}), and $\wh E_1$, $\wh E_2$,
by (\ref{hatE_1}), (\ref{hatE_2}), resp.
Let $z\to z_1$ in such a way that $\zeta>1$. Then we obtain
\be
\mathcal D_t(z_1)=b_0\frac{b_+(z_1)}{b_-(z_1)}z_1^{-2\bt_1}z_2^{-\al_2-\bt_2}
e^{-i\pi(\al_1+\bt_1+\al_2+\bt_2)}(z_1 e^{2\pi i})^{\al_2-\bt_2},
\ee
and, by (\ref{hatE_1}),
\be\label{E11}
\wh E_1(z_1)=
\si_1\mathcal D_t(z_1)^{-\si_3/2}\left(\frac{t}{2e^{5\pi i/2}\sin t}\right)^{\bt_2\si_3}z_1^{-\bt_1\si_3}
e^{-\frac{i}{2}nt\si_3}e^{-\frac{i}{2}\pi(\al_1-\bt_1)\si_3}.
\ee

Similarly,
\be
\mathcal D_t(z_2)=b_0\frac{b_+(z_2)}{b_-(z_2)}z_2^{-2\bt_2}z_1^{-\al_1-\bt_1}
e^{-i\pi(\al_1+\bt_1+\al_2+\bt_2)}z_2^{\al_1-\bt_1},
\ee
and, by (\ref{hatE_2}),
\be\label{E22}
\wh E_2(z_2)=
\si_1\mathcal D_t(z_2)^{-\si_3/2}\left(\frac{t}{2e^{3\pi i/2}\sin t}\right)^{\bt_1\si_3}z_2^{-\bt_2\si_3}
e^{\frac{i}{2}nt\si_3}e^{-\frac{i}{2}\pi(\al_2-\bt_2)\si_3}.
\ee

Substituting (\ref{E11}), (\ref{E22}), into (\ref{Upsint2}), and that, in turn, into (\ref{phichi2}),
we finally obtain uniformly in $t$
\be
\begin{aligned}\label{as phichi 2}
\wh \phi_{n-1}(0)\chi_{n-1}b_0&=
n^{2\bt_1-1}z_1^{-n+1}\frac{b_-(z_1)}{b_+(z_1)}
\frac{\Gamma(1+\al_1-\bt_1)}{\Gamma(\al_1+\bt_1)}e^{i(\pi-2t)\al_2}(2\sin t)^{-2\bt_2}\\&+
n^{2\bt_2-1}z_2^{-n+1}\frac{b_-(z_2)}{b_+(z_2)}
\frac{\Gamma(1+\al_2-\bt_2)}{\Gamma(\al_2+\bt_2)}e^{i(-\pi+2t)\al_1}(2\sin t)^{-2\bt_1}\\&+
\bigO(n^{-2+2\bt_1}t^{-1-2\bt_2})+\bigO(n^{-2+2\bt_2}t^{-1-2\bt_1}),\qquad C_0/n<t<t_0.
\end{aligned}
\ee

\begin{remark}
We obtained the expansion (\ref{as phichi 2}) only under the condition $|||\bt|||<1$,
$\Re\al_j>-1/2$.
For a fixed $t>0$ this expansion coincides with the result which follows
from Theorem 1.8 in \cite{DIK2}.
\end{remark}

\begin{remark}\label{2beta1}
Let $\Re\bt_1=\Re\bt_2$ as this is the case we need here.
Using the large $s$ expansion for $r(s)$ (obtained similarly to that of $q(s)$ from (\ref{555})):
\be
\begin{aligned}\label{as r}
r(s)=&-2|s|^{-1-\bt_2}e^{-i|s|/2}e^{i\pi(\al_1-3\bt_1-\bt_2)}\frac{\Gamma(1+\al_1-\bt_1)}
{\Gamma(\al_1+\bt_1)}\\
&-2|s|^{-1-\bt_1}e^{i|s|/2}e^{-i\pi(\al_2+3\bt_1+\bt_2)}\frac{\Gamma(1+\al_2-\bt_2)}
{\Gamma(\al_2+\bt_2)}\\
&+
\bigO(|s|^{-2-2\bt_2}),\qquad s\to -i\infty,\qquad\Re\bt_1=\Re\bt_2,
\end{aligned}
\ee
where $\al_k\pm\bt_k\neq -1,-2,\dots$, $\Re\al_k>-1/2$,
and the estimate
\be\label{psi2est}
\bigO(\Psi_{2, 12}(s))=\bigO(|s|^{-2-2\bt_2}),\qquad s\to -i\infty,\qquad\Re\bt_1=\Re\bt_2,
\ee
the reader can easily verify that (\ref{as phichi 2}) agrees with (\ref{as phichi 1})
for $t\in(C_0'/n, C_0''/n)$, $C_0'<C_0<C_0''$.
\end{remark}

Formulae (\ref{as phichi 2}) and (\ref{mainsec9}) yield the second part of (\ref{asDb1}).

\begin{remark}
Similar to the derivation of the small $s$ asymptotics for $\sigma$, we obtain from (\ref{whPsi as2})
that
\begin{multline}\label{r small s}
r(s)=-\frac{2}{|s|}e^{s/2}e^{i\pi(\al_1-\al_2-3\bt_1-\bt_2)}\frac{\Gamma(1+\al_1+\al_2-\bt_1-\bt_2)}
{\Gamma(\al_1+\al_2+\bt_1+\bt_2)}\\
\times\left(1+\bigO(|s\ln|s||)+\bigO(|s|^{1+2(\al_1+\al_2)})\right),\qquad s\to-i0_+,
\end{multline}
for $\al_k\pm\bt_k\neq -1,-2,\dots$, $(\al_1+\al_2)\pm(\bt_1+\bt_2)\neq -1,-2,\dots$,
$\Re\al_k,\Re(\al_1+\al_2)>-1/2$.
\end{remark}

\subsection{Special case $\al_1=\al_2=\bt_1^-=\bt_2^-+1=1/2$}\label{sec sp case}

In the case $\al_1=\al_2=\bt_1=\bt_2=1/2$ (i.e., $\bt_1^-=\bt_2^-+1=1/2$)
relevant for the problem of Toeplitz eigenvalues \cite{DIK3}, the situation simplifies.
The problem for $\Psi$ is then solved in Section \ref{section: degenerate}
in elementary functions , and we obtain from (\ref{def Psi deg}) that
\be
r(s)=-\frac{2i}{s^2}\left(e^{-\frac{s}{2}}  - e^{\frac{s}{2}}\right)=
-\frac{\sin nt}{(nt)^2},\qquad s=-2int.
\ee
Therefore, using (\ref{mainsec9}), (\ref{as phichi 1}), we obtain
\be\label{asd1}
D_{n-1}(f^-)=
e^{-i(n-1)t}{b_0^{-1}\over \sin t}\left[\frac{b_-(1)}{b_+(1)}\sin nt+\bigO(n^{-1})\right]
D_n(f),
\ee
with the error term uniform for $0<t\le C_0/n$.
On the other hand, (\ref{mainsec9}) and (\ref{as phichi 2}) give in this case
\begin{multline}\label{asd2}
D_{n-1}(f^-)=
e^{-i(n-1)t}b_0^{-1}|z_1-z_2|^{-1}\left[
z_1^{-n+1}\frac{b_-(z_1)}{b_+(z_1)}
\left(\frac{z_2}{z_1}\right)^{1/2}e^{-i\pi/2}\right.\\
\left.+
z_2^{-n+1}\frac{b_-(z_2)}{b_+(z_2)}
\left(\frac{z_1}{z_2}\right)^{1/2}e^{i\pi/2}
+\bigO(n^{-1})\right]D_n(f),
\end{multline}
with the error term uniform for $C_0/n<t<t_0$.
Note that for $t$ of order $1/n$ or less, the formula (\ref{asd2}) reduces to (\ref{asd1}).
Therefore, (\ref{asd2}) holds uniformly for {\it all} $0<t<t_0$.

\section*{Acknowledgements}
TC was supported by the European Research Council under the European Union's Seventh Framework Programme (FP/2007/2013)/ ERC Grant Agreement n.\, 307074, by FNRS, and by the Belgian Interuniversity Attraction Pole P07/18. IK was partially supported by EPSRC Platform Grant.

\end{document}